\documentclass[a4paper,UKenglish,cleveref, autoref, thm-restate]{lipics-v2021}

\hideLIPIcs  

\nolinenumbers

\usepackage{listings}
\usepackage{times}
\usepackage{soul}
\usepackage{url}
\usepackage[utf8]{inputenc}
\usepackage{graphicx}
\usepackage{amsmath}
\usepackage{amsthm}
\usepackage{booktabs}
\usepackage[mathlines]{lineno}
\usepackage{etoolbox}
\usepackage{url}
\usepackage[obeyFinal]{todonotes}
\usepackage[utf8]{inputenc}
\usepackage{csquotes}
\usepackage{rotating}
\usepackage{graphicx}
\usepackage{multirow}
\usepackage{multicol}
\usepackage{tikz}
\usetikzlibrary{arrows}
\usepackage{macros}

\usepackage{comment}

\newif\ifdraft
\newif\ifshortversion

\ifshortversion
  \includecomment{shortonly}
  \excludecomment{longonly}
\else
  \includecomment{longonly}
  \excludecomment{shortonly}
\fi

\usepackage[capitalise]{cleveref}
\crefformat{footnote}{#2\footnotemark[#1]#3}

\newtheorem{fact}[theorem]{Fact}
\title{
ABox Abduction for Inconsistent Knowledge Bases under Repair Semantics
}

\author{Anselm Haak}{Knowledge Representation Group, Paderborn University, Germany}{anselm.haak@uni-paderborn.de}{https://orcid.org/0000-0003-1031-5922}{}

\author{Patrick Koopmann}{Knowledge in Artificial Intelligence, Vrije Universiteit Amsterdam, The Netherlands}{p.k.koopmann@vu.nl}{https://orcid.org/0000-0001-5999-2583}{}

\author{Yasir Mahmood}{Data Science Group, Paderborn University, Germany}{yasir.mahmood@uni-paderborn.de}{https://orcid.org/0000-0002-5651-5391}{}

\author{Anni-Yasmin Turhan}{Knowledge Representation Group, Paderborn University, Germany}{turhan@uni-paderborn.de}{https://orcid.org/0000-0001-6336-335X}{}

\authorrunning{A. Haak, P. Koopmann, Y. Mahmood, A. Turhan} 

\Copyright{Anselm Haak, Patrick Koopmann, Yasir Mahmood and Anni-Yasmin Turhan} 



\category{} 

\relatedversion{} 

\acknowledgements{
  We thank all anonymous reviewers for their valuable feedback that helped us improve the exposition of our paper.
  The third author appreciates funding by the  Deutsche Forschungsgemeinschaft (DFG, German Research Foundation), grant TRR 318/1 2021 – 438445824.
}

\EventEditors{John Q. Open and Joan R. Access}
\EventNoEds{2}
\EventLongTitle{42nd Conference on Very Important Topics (CVIT 2016)}
\EventShortTitle{CVIT 2016}
\EventAcronym{CVIT}
\EventYear{2016}
\EventDate{December 24--27, 2016}
\EventLocation{Little Whinging, United Kingdom}
\EventLogo{}
\SeriesVolume{42}
\ArticleNo{23}

\begin{document}

\maketitle

\begin{abstract}
Given a knowledge base (KB)  with a non-entailed fact, the ABox abduction problem asks for possible extensions of the KB  that would entail this fact. This problem has many applications, ranging from diagnosis to explainability and repair. ABox abduction has been well-investigated for consistent KBs and classical semantics,  but little is known for the case of inconsistent KBs, which can be caused by erroneous data.  
In this paper we define suitable notions of abduction in this setting and propose criteria that guide abduction towards \enquote{useful} hypotheses. To regain meaningful reasoning in the presence of inconsistencies, we use well-established repair semantics.
We provide a comprehensive landscape of the complexity of ABox abduction under repair semantics, treating different variants of the abduction problem for the light-weight description logics \DLLite and \ELbot. 
\end{abstract}

\section{Introduction}

In computational logic, abductive reasoning can be defined as follows:
Given two formulas $K$ (the \emph{background knowledge}) and $\Obs$ (the \emph{observation}),
we are looking for a \emph{hypothesis} $H$ s.t. $K\wedge H\models\Obs$. The original idea dates back
to~\cite{peirce1878deduction} as a type of non-monotonic
reasoning towards finding a plausible
explanation $H$ for an unexpected observation $\Obs$, given our knowledge $K$
about the situation. Since then, it has been
suggested for various additional use cases: to provide possible explanations for black box classifiers
in \emph{machine learning}~\cite{ignatiev2019abduction}, to \emph{explain
missing entailments} from knowledge bases (\enquote{if $H$ was
in your knowledge base $K$, $\Obs$ would be entailed})~\cite{du2015towards,AlrabbaaBFHKKKK24}, and to suggest
possible solutions for \emph{repairing incomplete knowledge bases}~\cite{WeiKleinerDragisicLambrix2014,du2017practical,Haifani2022}.
Diagnosis, explainability and repair are well-motivated in the area of ABox reasoning with description logic
(DL) ontologies~\cite{SchlobachC03,BaaderKKN22}: here, the ABox $\Amc$ contains a set of facts like in a
database, which is used
together with an ontology or TBox $\Tmc$ to derive new implicit facts.
\emph{ABox abduction} is the
special case of abduction in which $K$ is such a DL knowledge base $\tup{\Tmc,\Amc}$, $
\Obs$ consists of facts, and we are looking for a set $\Hmc$ of facts
as hypothesis that would ensure $\tup{\Tmc, \Amc \cup \Hmc} \models
\Obs$~\cite{Klarman2011,Calvanese2013,Ceylan2020,Koopmann21a}.
To avoid nonsensical explanations,
one usually additionally requires \emph{consistency} ($\tup{\Tmc, \Amc \cup \Hmc} \not\models \bot$),
\emph{non-triviality} ($\Hmc \neq \Obs$) and sometimes puts restrictions on the \emph{signature} of $
\Hmc$~\cite{Elsenbroich2006,Koopmann2020}.

\newcommand{\dl}[1]{\textit{#1}}

\newcommand{\exA}{\dl{High}}
\newcommand{\exB}{\dl{Low}}
\newcommand{\exC}{\dl{GlycemicCrisis}}
\newcommand{\exs}{\dl{glucoseLevel}}
\newcommand{\exr}{\dl{experiences}}
\newcommand{\exE}{\dl{OverdosedInsulin}}
\newcommand{\exF}{\dl{Ketoacidosis}}
\newcommand{\exO}{\dl{DiabeticComa}}
\newcommand{\exa}{\dl{patient}}
\newcommand{\exb}{\dl{l}}
\newcommand{\exc}{c}

\begin{example}
 Consider the following excerpt from a medical ontology $\Tmc$ on diabetes:
 \begin{align*}
  \exA\sqcap\exB\sqsubseteq \bot\ \ \ \
  \exists\exs.\exA&\sqsubseteq \exC \\
  \exists\exs.\exB&\sqsubseteq\exC\\
  \exists\exs.\exA\sqcap\exE&\sqsubseteq \exO\\
  \exC\sqcap\exF&\sqsubseteq \exO
 \end{align*}
 The following ABox $\Amc_1$ contains medical evidence about a patient obtained through a glucose monitor:
 \[
  \exs(\exa,\exb)\qquad \exA(\exb)
 \]
 We observe that the patient passed out. Based on his glucose reading, a reasonable explanation is that he took too much insulin. Indeed, using ABox abduction
 for the observation $\exO(\exa)$, a possible hypothesis would be $\{\exE(\exa)\}$.
\end{example}
If $\tup{\Tmc, \Amc}$ is inconsistent, every fact is entailed, and under standard semantics, both abductive
and deductive reasoning cannot be used any more to make useful inferences.
Unfortunately, consistency cannot always be assumed in realistic settings, as
data often contains erroneous information. To overcome this issue, a range of
inconsistency-tolerant semantics have been proposed~\cite{PARACONSISTENT,BienvenuB16,COST-BASED}, in particular repair-based semantics,
\emph{repair semantics} for short,
which are
defined based on maximal consistent subsets of the ABox called \emph{repairs}.
\begin{example}
 The doctor uses a finger stick sensor to get additional evidence about the
 glucose level of the patient, which leads to the additional ABox fact $\exB(\exb)$.
  The new ABox $\Amc_2$ is inconsistent and has two repairs, namely $\Amc_1$ and $\{\exs(\exa,\exb)$, $\exB(\exb)\}$.
  Under \emph{Brave} semantics, which considers repairs in isolation, $\exA(\exb)$, $\exB(\exb)$ and $
  \exC(\exa)$ are all entailed, while under \emph{AR} semantics, which considers all repairs at the
  same time, of those only $\exC(\exa)$ is entailed. We can adapt the definition of
  abduction to consider those entailment relations instead of classical
  entailment. Then, $\{\exE(\exa)\}$ and $\{\exF(\exa)\}$ are both hypotheses
  under Brave semantics, while only $\{\exF(\exa)\}$ is a hypothesis under AR semantics. Indeed, the cautious doctor should investigate this hypothesis first.
\end{example}
It turns out that considering repair semantics changes a few things compared to
classical semantics. The consistency requirement
$\tup{\Tmc, \Amc \cup \Hmc} \not\models \bot$ becomes obsolete, and we may
instead require that hypotheses should not introduce new conflicts in the knowledge base
(that they are \emph{conflict-confining}), or at least focus on those hypotheses
that minimise the newly introduced conflicts. Also, allowing for the
observation itself as the trivial hypothesis does not as
easily trivialise the problem as under classical semantics: It might not always satisfy the conditions,
and there are even cases in which other hypotheses introduce less new conflicts than the observation
itself would.

While entailment under repair semantics is
well-investigated, and there is work on explaining missing
entailments under these semantics, to the best of our knowledge,
the only work that considers abduction under repair semantics is by Du et al.~\cite{du2015towards},
who present
a practical system for computing
conflict-confining hypotheses with rewritable ontology languages under IAR 
semantics, another variant of repair semantics.
Abduction has also been investigated for paraconsistent semantics~\cite{BienvenuIK24},
a different inconsistency-tolerant semantics, but only for propositional logic.
A more complete picture for ABox abduction under repair-based semantics,
considering different semantics and logics, and also analysing
how to appropriately define abduction in these settings, has been missing so far.

One central motivation of abduction is to explain missing entailments, and indeed
different authors have investigated this problem using a different strategy:
Specifically,
Bienvenu et al.~\cite{bienvenu2019computing} and Lukasiewicz et al.~\cite{lukasiewicz2022explanations} explain missing
entailments under repair semantics by pointing to facts that \enquote{block}
the entailment by creating conflicts in parts of the ABox that would otherwise
lead to the entailment. In other words, while our explanations consist of facts that would need
to be added to make the entailment true, those explanations consist of facts that need to be
\emph{removed}. However, it is easy to construct examples where an abductive solution exists,
while it is not possible to make the entailment true by removing axioms.
Our results thus complement these works, by providing explanations
in case their approaches are incomplete.

We provide the first comprehensive analysis of abduction under repair semantics.
Since complexity under these semantics easily
trivialises in logics that are \ExpTime-hard, we consider two
prominent tractable description logics,
\ELbot and DL-Lite, which play a central role in many large-scale knowledge bases.
We consider different variants of the problem that are suitable in the
context of repair semantics, for example by requiring the aforementioned conflict-confinement
as well as signature-restrictions, and also investigate
different optimality criteria for hypotheses. Along the way, we show different properties
of hypotheses under repair semantics that help to develop a better
understanding of abduction for the analysed semantics. Some of our results also apply
to the consistent setting, but have to the best of our knowledge not been published before.
Our results are shown in~\Cref{tab:results}.

\begin{table*}
	\centering
\resizebox{.99\textwidth}{!}{
	\begin{tabular}{l@{\hspace{7pt}}l@{\hspace{5pt}}|cccc@{\hspace{5pt}}|ccccc}
		\toprule
		& &
		\multicolumn{4}{c|}{\textbf{Existence Problem}}& 
		\multicolumn{5}{c}{\textbf{Verification Problem}}\\ 
		& 
		& general 
		& signature 
		& conflict-conf.
		& non-trivial 
		& $\leq$-min
		& $\subseteq$-min
		& $\subseteq_c$-min
    & $\leq_c$-min
		& conflict-conf. \\
		\midrule 
		\multirow{2}{*}{$\DLLite$}
		& \brave
		& \NL$^{\text{T}\ref{thm:dllite-exist-cc}}$
		& \NL$^{\text{T}\ref{thm:dllite-exist-brave-sig-nontriv}}$
		& \NL$^{\text{T}\ref{thm:dllite-exist-cc}}$
		& \NL$^{\text{T}\ref{thm:dllite-exist-brave-sig-nontriv}}$ 
		& \NL$^{\text{T}\ref{thm:dllite-verif-gen-card}}$
		& \NL$^{\text{T}\ref{thm:dllite-verif-brave-subset}}$
		& \NL$^{\text{T}\ref{thm:dllite-verif-cc-cmin}}$
		& \NL$^{\text{T}\ref{thm:dllite-verif-cc-cmin}}$
		& \NL$^{\text{T}\ref{thm:dllite-verif-cc-cmin}}$ \\
		& \ar
		& \NL$^{\text{T}\ref{thm:dllite-exist-cc}}$
		& \coNP$^{\text{T}\ref{thm:dllite-exist-ar-sig-nontriv}}$
		& \NL$^{\text{T}\ref{thm:dllite-exist-cc}}$
		& \coNP$^{\text{T}\ref{thm:dllite-exist-ar-sig-nontriv}}$ 
		& \co\NP$^{\text{T}\ref{thm:dllite-verif-gen-card}}$
		& \DP$^{\text{T}\ref{thm:dllite-verif-ar-subset}}$
		& \NL$^{\text{T}\ref{thm:dllite-verif-cc-cmin}}$
		& \NL$^{\text{T}\ref{thm:dllite-verif-cc-cmin}}$
		& \NL$^{\text{T}\ref{thm:dllite-verif-cc-cmin}}$ \\
		\multirow{2}{*}{$\ELbot$}
		& \brave 
		& \Ptime$^{\text{T}\ref{thm:elbot-exist-general}}$
		& \NP$^{\text{T}\ref{thm:elbot-exist-sig}}$
		& \SigmaP$^{\text{T}\ref{thm:elbot-exist-cc}}$
		& \NP$^{\text{T}\ref{thm:elbot-exist-nontriv}}$
		& \NP$^{\text{T}\ref{thm:elbot-verif-general}}$
		& \DP$^{\text{T}\ref{thm:elbot-verif-subset}}$
		& \PiP$^{\text{T}\ref{thm:elbot-verif-cmin-subset}}$
		& --
		& \DP$^{\text{T}\ref{thm:elbot-verif-cc}}$ \\
		& \ar 
		& \coNP$^{\text{T}\ref{thm:elbot-exist-general}}$
		& \SigmaP$^{\text{T}\ref{thm:elbot-exist-sig}}$
		& \coNP$^{\text{T}\ref{thm:elbot-exist-cc}}$
		& \SigmaP$^{\text{T}\ref{thm:elbot-exist-nontriv}}$
		& \co\NP$^{\text{T}\ref{thm:elbot-verif-general}}$
		& \PiP$^{\text{T}\ref{thm:elbot-verif-subset}}$
		& \co\NP$^{\text{T}\ref{thm:elbot-verif-cmin-subset}}$
    & \co\NP$^{\text{T}\ref{thm:elbot-verif-cmin-card}}$
		& \co\NP$^{\text{T}\ref{thm:elbot-verif-cc}}$ \\
		\bottomrule
	\end{tabular}

}
	\caption{Complexity for existence and verification of hypotheses with different properties and repair semantics.
	All entries are completeness results, with superscripts referencing corresponding theorems. 
	Verifying general hypotheses has the same complexity as with $\leq$-minimality.}\label{tab:results}
\end{table*}

\section{Preliminaries} 
For a general introduction to description logics, we refer the reader to~\cite{DBLP:books/daglib/0041477}.
We assume familiarity with computational complexity~\cite{DBLP:books/daglib/0092426}, in particular with the complexity classes
$\NL, \Ptime, \NP, \co\NP$, $\SigmaP$ and~$\PiP$, as well as
$\DP$, the class of decision problems representable as the intersection of a problem in $\NP$ and a problem in $\co\NP$.

\subsection{The Description Logics \ELbot and \DLLite}\label{sec:dl-defs}
We use the following enumerable sets of names: \NC for \emph{concept names}, \NR for \emph{role names}, and \NI for \emph{individuals}.
An ABox is a finite set of \emph{concept assertions } (also known as \emph{flat} or \emph{simple} assertions) of the form $A(a)$ and
\emph{role assertions} of the form $r(a,b)$ for $A\in\NC$, $r\in\NR$, $a,b\in\NI$.
A \emph{TBox} is a finite set of \emph{concept inclusions} of the form $C\subsum D$ for concepts $C$ and $D$ and role inclusions of the form $Q \sqsubseteq S$ for roles $Q$ and $S$. 
Concept and role inclusions are also called \emph{axioms}.
A KB is a tuple $\tup{\calT,\calA}$ for a TBox $\calT$ and ABox $\calA$.
By \emph{signature} of a KB $\calK$, we mean the set of (concept, role, and individual) names appearing in $\calK$.

We next specify the syntax for \ELbot and the considered \DLLite dialects, with semantics being defined as usual~\cite{DBLP:books/daglib/0041477}.
For \ELbot concepts, the syntax is given by the grammar:
\[C \Coloneqq C \sqcap C \mid \exists r.C \mid A \mid \top \mid \bot ~.\]
An \emph{\ELbot KB} restricts the TBox to concept inclusions over \ELbot concepts. 

We consider the $\DLLite$ dialects $\DLLiteR$ and \DLLiteCore.
In \DLLiteR (underlying the OWL 2 QL profile),
TBoxes may contain \emph{concept inclusions} of the form $B \subsum C$ and \emph{role inclusions} of the form $Q \sqsubseteq S$,
where $B, C, Q$ and $S$ are generated by the following grammar:
\begin{align*}
B &\Coloneqq A \mid \exists Q, \qquad C \Coloneqq B \mid \neg B, \qquad Q \Coloneqq R \mid R^-, \qquad S \Coloneqq Q \mid \neg Q,
\end{align*}
where $A \in \NC$ and $R \in \NR$.
\DLLiteCore restricts \DLLiteR by disallowing role inclusions, so that only concept inclusions of the above form are allowed.
Semantics for all three DLs are defined as usual.
For an ABox \calA and TBox \calT, we say that \calA is \emph{\calT-consistent}, if $\tup{\calT,\calA}\not\models\bot$, and \emph{\calT-inconsistent} otherwise.
Further, we say that \calA is a \calT-support of a concept assertion $\alpha$, if $\tup{\calT, \calA} \models \alpha$.

For the rest of the paper, the general term \DLLite refers to either \DLLiteCore or \DLLiteR.
We do so since all of our results apply to both DLs: our lower bounds apply to \DLLiteCore and our upper bounds apply to \DLLiteR.

\subsection{Repair Semantics}
If a knowledge base is inconsistent, repair semantics can \enquote{restore} consistent versions and admit meaningful reasoning again.
We focus here on ABox repairs and define these as well as two common kinds of repair semantics next.

Let $\calK = \tup{\calT, \calA}$ be an inconsistent knowledge base.
A \emph{repair} of $\calK$ is a \calT-consistent subset $\calR \subseteq \calA$ and subset-maximal with this property, i.e., there is no \calT-consistent subset $\calR' \subseteq \calA$ that is a strict superset of \calR.
The somewhat dual notion is a \emph{conflict} or \emph{conflict set} \conflict, which is a \calT-inconsistent subset of the ABox and subset-minimal with this property.
We denote by $\Conf(\calK)$ the set of conflicts of \calK.
We recall entailment under \brave~\cite{BienvenuR13} and \ar semantics~\cite{LLRRS-JWS-15} for concept assertions $\alpha$:
\begin{itemize}
	\item $\calK\models_{\brave} \query$ iff there is a repair $\calR$ of $\calK$ s.t.\ $\tup{\calT,\calR} \models \query$.
	\item $\calK\models_{\ar} \query$ iff $\tup{\calT,\calR}\models \query$ for every repair $\calR$ of $\calK$.
\end{itemize}
The complexity of entailment under repair semantics is well understood~\cite{BienvenuB16}:
Checking entailment of concept assertions under \brave semantics is \NL-complete for $\DLLiteCore$ and \DLLiteR, and \NP-complete for $\EL_\bot$ in combined complexity, whereas under \ar semantics it is \coNP-complete for all three DLs.

Besides the complexity of entailment under repair semantics,
both \DLLite dialects share the following properties for \DLLite KBs $\calK = \tup{\calT, \calA}$ \cite{BienvenuB16}:
\begin{itemize}
  \item conflicts of \calK are of size $2$, and
  \item subset-minimal \calT-supports of concept assertions $\alpha$ are of size $1$.
\end{itemize}

\section{ABox Abduction for Inconsistent KBs}
\label{sec:abd}

The central task of abduction is to compute hypotheses. We define these for non-entailed concept assertions  under repair semantics.
\begin{definition}\label{def:hypotheses}
  Let $\calK = \tup{\calT, \calA}$ be an inconsistent KB, $\alpha$
  a concept assertion
  (called an \emph{observation}) and $\calS \in \{\brave, \ar\}$ such that $\calK \not\models_\calS \alpha$.
  Then, the pair $\tup{\calK, \alpha}$ is called an \emph{$\calS$-abduction problem}.
  A solution for such a problem, called \emph{\calS-hypothesis}, is an ABox \Hyp using only individuals occurring in  \calK and $\alpha$ s.t.\ $\tup{\calT, \calA \cup \Hyp} \models_\calS \alpha$.
  
  Additionally, for a set $\Sigma$ (signature) containing individual, concept and role names, we define the triple $\tup{\calK, \alpha, \Sigma}$ to be a \emph{\SignatureRestricted \calS-abduction problem}.
  A solution for such a problem is an ABox \Hyp using only symbols from $\Sigma$ that is an \calS-hypothesis for $\tup{\calK, \alpha}$.
\end{definition}
For an \calS-abduction problem $\tup{\calK, \alpha}$ we require that \calK is inconsistent and $\calK\not\models_\calS \alpha$.
This means we consider only the so-called \emph{promise problem}, i.e.\ the problem restricted to these particular inputs.
The restriction aligns with the intuition that one asks for an $\calS$-hypothesis if it is already known that the knowledge base is
inconsistent and the observation is not $\calS$-entailed in $\calK$.
In contrast, if we instead assume that \calK is consistent and $\alpha$ is not entailed by \calK under classical semantics,
we obtain classical abduction problems.
In this case, we call an ABox \Hyp \emph{hypothesis for $\alpha$ under classical semantics},
if $\tup{\calT, \calA \cup \Hyp}\not\models \bot$ and $\tup{\calT, \calA \cup \Hyp} \models \alpha$.
Note also that we do not admit fresh individuals beyond the specified signature as it is done in~\cite{Koopmann21a}.
This can sometimes make the problem \ExpTime-hard already in the classical case, and is left for future work.

To obtain hypotheses that are meaningful for explanation purposes, one often considers additional properties of hypotheses as well as minimality criteria that yield \emph{preferred hypotheses}.
We transfer some of this terminology already defined for abduction under classical semantics, and extend it to repair semantics by additional properties and minimality notions based on conflicts.
Of those, the notion we call \emph{conflict-confining} was already used in~\cite{du2017practical}.
For two sets $S_1, S_2$ we may abbreviate $\lvert S_1\rvert \leq \lvert S_2\rvert$ by $S_1 \leq S_2$.
\begin{definition}\label{def:properties}
  Let ${\calS} \in {\{\brave,\ar\}}$, $\tup{\calK, \alpha}$ be an \calS-abduction problem, and 
  ${\preceq} \in {\{\subseteq, \leq\}}$.
  An ABox \Hyp is
  \begin{enumerate}
    \item\label{def:conf-conf} \emph{conflict-confining for $\calK = \tup{\calT, \calA}$}, provided that
    \mbox{$\Conf(\tup{\calT, \calA \cup \Hyp}) = \Conf(\calK)$}.
  \end{enumerate}
  If \Hyp is an \calS-hypothesis for $\tup{\calK, \alpha}$, we call it
  \begin{enumerate}
    \setcounter{enumi}{1}
    \item\label{def:non-triv} \emph{non-trivial}, if $\alpha \not\in \Hyp$,
    \item\label{def:hyp-minimal} \emph{$\preceq$-minimal}, if there is no $\calS$-hypothesis $\Hyp'$ for $\tup{\calK, \alpha}$ s.t.\ $\Hyp' \prec \Hyp$, and
    \item\label{def:conf-minimal}  \emph{$\preceq_c$-minimal}, if there is no \calS-hypothesis $\Hyp'$ for $\tup{\calK, \alpha}$ s.t.\ $\Conf(\tup{\calT, \calA \cup \Hyp'}) \prec \Conf(\tup{\calT, \calA \cup \Hyp})$.
  \end{enumerate}
\end{definition}
While Condition~\ref{def:non-triv} and~\ref{def:hyp-minimal} are standard for abduction,
Condition~\ref{def:conf-conf} and~\ref{def:conf-minimal}
adapt the idea of a hypothesis not (resp., minimally) introducing new inconsistencies to the KB, which is already inconsistent to begin with.
Conflict-confinement can be equivalently defined by requiring that $\tup{\calT, \calR \cup \Hyp} \not\models\bot$ for every repair \calR of \calK.
Furthermore, the minimality in Condition~\ref{def:conf-minimal} generalises this notion and allows to minimally add new conflicts. 
Note that hypotheses introducing new conflicts can be desirable, as erroneous facts might need to implicitly be replaced by new facts, leading to conflicts when considering both together. 

We use the term \emph{subset-minimal} for $\subseteq$-minimal and \emph{cardinality-minimal} for $\leq$-minimal.
Further, we also consider minimal variants of hypotheses with additional properties:
For example, a $\subseteq$-minimal conflict-confining \ar-hypothesis need only be $\subseteq$-minimal among all conflict-confining \ar-hypotheses.

Conflict-confinement is a general property of an ABox and does not depend on the semantics or additional properties of abduction problems.
Hence, we can already establish the complexity of checking for this property here.

\begin{restatable}{lemma}{ccCompl}\label{lem:cc-compl}
	Given a KB \calK and an ABox \calH, checking whether \calH is conflict-confining for \calK is
  \begin{enumerate*}[label=(\arabic*)]
    \item \NL-complete, if \calK is a \DLLite KB, and
    \item \coNP-complete, if \calK is an \ELbot KB.
  \end{enumerate*}
\end{restatable}
\begin{proof}[Proof sketch]
  \NL-membership for \DLLite follows from the fact that conflicts are of size at most $2$, so we can iterate over all possible conflicts in logarithmic space and verify that they are indeed fresh conflicts introduced by \calH.
  Hardness can be shown by reduction from directed non-reachability, constructing for a directed graph $G = (V,E)$ and $s,t \in V$ the TBox 
  \begin{align*}
    \calT_{\text{unreach}} \coloneqq {} &\{\, A_v \sqsubseteq A_w \mid (v,w) \in E \text{ and } t \not\in \{v,w\} \,\} \cup {} \\
    &\{A_t \sqsubseteq A_s\} \cup \{\, A_v \sqsubseteq \neg A_t \mid (v,t) \in E \,\}.
  \end{align*}
  This TBox encodes edges of $G$ by concept inclusions, but replaces all edges incident to $t$ by a single edge from $t$ to $s$ and disjointness of $t$ with all of its predecessors.
  This ensures that existence of an $s$-$t$-path in $G$ is equivalent to $A_t$ being satisfiable w.r.t.\ $\calT_{\text{unreach}}$, and hence to $\{A_t(a)\}$ being conflict-confining for $\tup{\calT_{\text{unreach}}, \emptyset}$.

  For membership in case of \ELbot, we universally guess a potential conflict and verify that it is not a fresh conflict introduced by \calH in polynomial time.
  Hardness is obtained by a straightforward reduction from non-entailment under \brave semantics.
\end{proof}

We investigate the following reasoning problems for a given (\SignatureRestricted) $\Sem$-abduction problem.
\begin{definition}[Reasoning Problems]\label{def:problems}\ 
	Given an \calS-abduction problem $\tup{\calK, \alpha}$ (resp.\ \SignatureRestricted \calS-abduction problem $\tup{\calK, \alpha, \Sigma}$),
	\begin{enumerate}[topsep=-\parskip+\lineskip]
		\item the \emph{existence problem} asks whether $\tup{\calK, \alpha}$ (resp.\ $\tup{\calK, \alpha, \Sigma}$) has a solution, and
		\item the \emph{verification problem} asks whether a given ABox \calH is a hypothesis for $\tup{\calK, \alpha}$ (resp.\ for $\tup{\calK, \alpha, \Sigma}$).
	\end{enumerate}
\end{definition}
Note that \SignatureRestriction and non-triviality are trivial to check for a given hypothesis, which is why we do not consider them for
the complexity of verification.
On the other hand, existence of minimal hypotheses is equivalent to existence of any hypothesis for all considered cases.

Under some repair semantics, already standard reasoning tasks can behave in unexpected ways.
This also holds true for abduction under repair semantics, with sometimes different effects depending on the DL.
For instance, there is an interesting form of non-monotonicity in the case of $\subseteq$-minimal \ar-hypotheses, which applies to both \DLLite and \ELbot.
\begin{observation}\label{obs:ar-non-convex}
  Let $\tup{\calK, \alpha}$ be an \ar-abduction problem, where \calK is an \DLLite or \ELbot KB.
  Then, the set of \ar-hypotheses for $\tup{\calK, \alpha}$ does not need to be convex.
\end{observation}
\begin{example}\label{ex:ar-non-convex}
  We demonstrate this by constructing a \DLLite KB $\calK = \tup{\calT, \calA}$ that can easily be transformed into an \ELbot KB, and ABoxes $\calB_1 \subsetneq \calB_2 \subsetneq \calB_3$ such that $\calB_1$ and $\calB_3$ are \ar-hypotheses for $\tup{\calK, A(a)}$, but $\calB_2$ is not,
  as follows: 
  \begin{align*}
  	\calT &\coloneqq \{ B_1 \sqsubseteq \neg B_2,~ 
  	C_1 \sqsubseteq \neg C_2,\quad 
  	B_1 \sqsubseteq A,~ 
  	B_3 \sqsubseteq A\}, \\
    \calA &\coloneqq \{C_1(a), C_2(a)\}, \qquad \calB_1 \coloneqq \{B_1(a)\}, \quad \calB_2 \coloneqq \calB_1 \cup \{B_2(a)\}, \quad \calB_3 \coloneqq \calB_2 \cup \{B_3(a)\}.
  \end{align*}
  The TBox \calT can readily be transformed into a \ELbot KB by changing the disjointness axioms to the appropriate form.
\end{example}
This observation implies that for $\subseteq$-minimality,
it is generally not sufficient to \emph{locally} check all subsets that remove one assertion at a time.
This leads to \PiP-hardness for verification of $\subseteq$-minimal \ar-hypotheses in \ELbot (see \Cref{thm:elbot-verif-ar-subset}).
In contrast, we can circumvent the need for a global check in \DLLite by building hypotheses bottom-up from singleton sets, using that in this case \calT-supports are always of size $1$.
This results in \DP-completeness here, as shown in \Cref{thm:dllite-ar-hyps} and \Cref{thm:dllite-verif-ar-subset}.

\subsection{Effect of the Trivial Hypothesis}
\label{sec:triv}

Because we only consider concept assertions as observations $\alpha$,
the ABox $\{\alpha\}$ is already  a candidate hypothesis for $\tup{\calK,\alpha}$ for any KB~\calK.
We study next how such \emph{trivial hypotheses} affect certain cases of the existence and the verification problem in our framework.
Regarding \brave semantics, the following is easy to see.
\begin{proposition}\label{prop:brave-triv}
  Let $\tup{\calK, \alpha}$ be a \brave-abduction problem.
  There is a \brave-hypothesis for $\tup{\calK, \alpha}$ iff $\{\alpha\}$ is a \brave-hypothesis for $\tup{\calK, \alpha}$.
\end{proposition}
\begin{proof}
  $\{\alpha\}$ is a \brave-hypothesis for $\tup{\calK, \alpha}$ if $\{\alpha\}$ is \calT-consistent, as in this case $\alpha$ is contained in some repairs of $\tup{\calT, \calA \cup \{\alpha\}}$.
  Otherwise, if $\{\alpha\}$ is \calT-inconsistent, there is no \brave-hypothesis for $\tup{\calK, \alpha}$.
\end{proof}

We next show that for existence of general \ar-hypotheses, it is sufficient to check the trivial hypothesis.
Furthermore, for the trivial hypothesis, the properties of being an \ar-hypothesis and being conflict-confining coincide.
Note that the trivial hypothesis $\{\alpha\}$ being conflict-confining for the KB \calK means that there is \emph{no} repair \calR of \calK s.t.\ $\calR \cup \{\alpha\}$ is \calT-inconsistent.
Intuitively, this can be seen as \emph{non-entailment} of the negation of $\alpha$ under \brave semantics.
This provides some intuition for why the complexity of checking whether an ABox is conflict-confining for a KB has the same complexity as the complement of \brave-entailment for both DLs, cf.\ \cref{lem:cc-compl}.

\begin{restatable}{lemma}{arTriv}\label{lem:ar-triv}
  Let $\tup{\calK, \alpha}$ be an \ar-abduction problem with $\calK = \tup{\calT, \calA}$.
  The following are equivalent:
  \begin{enumerate}
    \item there is an \ar-hypothesis for $\tup{\calK, \alpha}$,
    \item $\{\alpha\}$ is an \ar-hypothesis for $\tup{\calK, \alpha}$, and
    \item $\{\alpha\}$ is conflict-confining for \calK.
  \end{enumerate}
\end{restatable}
\begin{proof}[Proof sketch]
  \emph{(2) $\Rightarrow$ (1)} is obvious.
  
  \emph{(3) $\Rightarrow$ (2)}: For any repair \calR of $\tup{\calT, \calA \cup \{\alpha\}}$, we argue that $\alpha \in \calR$ and hence $\alpha$ is \ar-entailed:
  If instead $\alpha \not\in \calR$, then \calR is a repair of \calK, so $\calR \cup \{\alpha\}$ is \calT-consistent contradicting maximality of \calR.

  \emph{(1) $\Rightarrow$ (3)}: If $\{\alpha\}$ is not conflict-confining for \calK, there is a conflict of $\tup{\calT, \calA \cup \{\alpha\}}$ containing $\alpha$.
  This yields a repair~\calR of $\tup{\calT, \calA \cup \{\alpha\}}$ not containing $\alpha$.
  Now for any ABox \Hyp, there is a repair $\calR'$ of $\tup{\calT, \calA \cup \Hyp}$ with $\calR \subseteq \calR'$.
  But since $\calR \cup \{\alpha\}$ is \calT-inconsistent, we have $\tup{\calT, \calR'} \not\models \alpha$, so \Hyp is not an \ar-hypothesis for $\alpha$.
\end{proof}

The above properties simplify the corresponding existence problems, and guarantee that $\leq$-minimal \ar-hypotheses are of size $1$.
For \ar semantics, existence of conflict-confining hypotheses is similarly affected, while for \brave semantics, 
the existence even becomes trivial if the concept name in $\alpha$ is satisfiable w.r.t.\ \calT.

\section{The Case of \DLLite} 
\label{sec:dllite}

We establish the complexity results for abduction problems over \DLLite KBs, and establish a number of interesting properties of these problems and the corresponding hypotheses.
An overview of the results is found in \cref{tab:results}, whereas this section is structured by the involved proof techniques of the results.
We begin by noting that verification of (general) \calS-hypotheses is essentially equivalent to \calS-entailment.
By \cref{prop:brave-triv} and \cref{lem:ar-triv}, this extends to $\leq$-minimal hypotheses under both repair semantics.

\begin{restatable}{theorem}{dlliteVerifGenCard}\label{thm:dllite-verif-gen-card}
  Verification of general and of  $\leq$-minimal \calS-hypotheses is
  \begin{enumerate*}[label=(\arabic*)]
    \item \NL-complete for $\calS = \brave$, and
    \item \coNP-complete for $\calS = \ar$.
  \end{enumerate*}
\end{restatable}
\begin{proof}[Proof sketch]
  The upper bounds are obtained from those for \calS-entailment, as argued above.
  The lower bound for \brave semantics follows by a reduction from directed reachability, based on the following construction.
  Given a directed graph $G = (V,E)$ and $s,t \in V$, define
  \[\calT \coloneqq \{\, A_v \sqsubseteq A_w \mid (v,w) \in E \,\}, \quad \Hyp_{\text{reach}} \coloneqq \{A_s(a)\}.\]
	Now there is an $s$-$t$-path in $G$ iff $\tup{\calT, \Hyp_{\text{reach}}} \models A_t(a)$.
  Adding a dummy inconsistency to obtain a KB $\calK_\text{reach}$ yields the desired reduction.
  Note that this is a simpler variant of the construction used in the proof of \cref{lem:cc-compl}.
  
  The lower bound for general \ar-hypotheses follows by reduction from unsatisfiability, adapting a construction
  from~\cite{bienvenu2019computing}.
	For a CNF $\varphi(x_1, \dots, x_n) \dfn \{c_1, \dots, c_k\}$, set $\calK_{\text{unsat}} \coloneqq \tup{\calT, \calA}$ where
  \begin{align*}
	  \calT \coloneqq {} & \{\exists U \subsum A,~\exists P^{-} \subsum \neg \exists N^{-} \} \cup \{ \exists P \subsum \neg \exists U^{-},~\exists N \subsum \neg \exists U^{-} \}, \\
	  \calA \coloneqq {} & \{P(c_j,x_i) \mid x_i \in c_j \} \cup \{N(c_j,x_i) \mid \neg x_i \in c_j \}, \text{ and} \\
    \Hyp \coloneqq {} & \{U(a,c_j) \mid j \leq k \}.
  \end{align*}
  It is known that $\tup{\calT, \calA \cup \Hyp} \models_\ar A(a)$ iff $\varphi$ is unsatisfiable, yielding the desired reduction.

  For $\leq$-minimal \ar-hypotheses, we use the same construction, but first transform $\varphi$ into a normal form,
  where~$\varphi$ has the property that removing the last clause~$c$ from~$\varphi$ always yields a satisfiable formula.
  This then allows us to obtain a singleton hypothesis candidate in the construction, which ensures $\leq$-minimality.
  Given a formula $\varphi = \{c_1, \dots, c_k\}$ with clauses $c_i$, the normal form can be achieved by letting $\varphi' \coloneqq \{c_1', \dots, c_{k+2}'\}$ with
  \begin{align*}
    c_i' &\coloneqq c_i \cup \{x_{n+1}\} \text{ for } 1 \leq i \leq k, \\
    c_{k+1}' &\coloneqq \neg x_{n+1} \lor x_{n+2}, \text{ and} \\
    c_{k+2}' &\coloneqq \neg x_{n+2}
  \end{align*}
  We observe that $\varphi'$ and $\varphi$ are equisatisfiable, but removing the clause $c_{k+2}'$ from $\varphi'$ yields a satisfiable sub-formula.

  Now for a formula $\varphi$ in our normal form with last clause~$c$, let \calT and \calA be as above, and define
  \[\calK' \coloneqq \tup{\calT, \calA \cup \{\, U(a,c') \mid c' \in \varphi \setminus \{c\} \,\}}\]
  Here, $\Hyp' \coloneqq \{U(a,c)\}$ is an \ar-hypothesis for $\tup{\calK', A(a)}$ iff $\varphi$ is unsatisfiable, since $\varphi \setminus \{c\}$ is never unsatisfiable.
  As $|\Hyp'| = 1$, $\leq$-minimality is ensured.
\end{proof}

Turning towards the existence problem, \cref{prop:brave-triv} and \cref{lem:ar-triv} show that for general \brave-hypotheses as well as for general and for conflict-confining \ar-hypotheses, it is sufficient to check the trivial hypothesis.
We next show that for \DLLite, the same applies to conflict-confining \brave-hypotheses by showing that here, \brave-hypotheses have a very simple structure.

\begin{restatable}{lemma}{dlliteSingletonHyps}\label{lem:dllite-singleton-hyps}
  Let $\tup{\calK, \alpha}$ be a \brave-abduction problem, where $\calK = \tup{\calT, \calA}$ is a \DLLite KB, and assume there is a \brave-hypothesis \Hyp for $\tup{\calK, \alpha}$.
  Then there is an assertion $\beta \in \Hyp$ s.t.\ $\{\beta\}$ is a \brave-hypothesis for $\tup{\calK, \alpha}$.
  Furthermore, $\Conf(\tup{\calT, \calA \cup \{\beta\}}) \subseteq \Conf(\tup{\calT, \calA \cup \Hyp})$ and for any conflict $\{\alpha, \gamma\} \in \Conf(\tup{\calT, \calA \cup \{\alpha\}})$, we have $\{\beta, \gamma\} \in \Conf(\tup{\calT, \calA \cup \{\beta\}})$.
\end{restatable}

In particular, the last part of the lemma shows that existence of any conflict-confining \brave-hypothesis is equivalent to $\{\alpha\}$ being one.
Now, using the upper bound for checking that an ABox is conflict-confining in \Cref{lem:cc-compl}, we obtain following complexity results.
Surprisingly, the complexity under \ar semantics drops below even that of \ar entailment.
This can be explained by the relationship to brave non-entailment noted before \Cref{lem:ar-triv}.

\begin{restatable}{theorem}{dlliteExistCc}\label{thm:dllite-exist-cc}
  The existence problems for general and for conflict-confining \calS-hypotheses are \NL-complete for $\calS \in \{\brave, \ar\}$.
\end{restatable}
\begin{proof}[Proof sketch]
  For general \brave-hypotheses, membership readily follows from \cref{prop:brave-triv}.
  The remaining cases are equivalent to checking that $\{\alpha\}$ is conflict-confining by lemmata~\ref{lem:ar-triv} and~\ref{lem:dllite-singleton-hyps}, and the complexity follows from \cref{lem:cc-compl}.
  Hardness for \brave-hypotheses can be shown by a similar reduction as the one used in \Cref{lem:cc-compl}, adding a dummy inconsistency.
\end{proof}

Building on these techniques, we next show that verification of conflict-confining and conflict-minimal hypotheses enjoys the same low complexity.

\begin{restatable}{theorem}{dlliteVerifCcCmin}\label{thm:dllite-verif-cc-cmin}
  Verification of conflict-confining and of $\preceq_c$-minimal \calS-hypotheses is \NL-complete for ${\preceq} \in \{\subseteq, \leq\}$ and $\calS \in \{\brave, \ar\}$.
\end{restatable}
\begin{proof}[Proof sketch]
  Hardness for all cases follows by reduction from directed reachability as in the proof of \cref{thm:dllite-verif-gen-card}, noting that $\Hyp_{\text{reach}}$ is conflict-confining for $\calK_{\text{reach}}$.
  To verify that \Hyp is a conflict-confining \brave-hypothesis, verify separately that it is a \brave-hypothesis and that it is conflict-confining, each possible in \NL.
  For the latter, use that conflicts are of size at most $2$, similar to \cref{thm:dllite-exist-cc}.
  For~\ar, first note that here, a hypothesis is $\preceq_c$-minimal iff it is conflict-confining by \cref{lem:ar-triv}, so we can focus on the conflict-confining case.
  We show that conflict-confining hypotheses have a simpler structure than general ones:
  together with the TBox, such a hypothesis already entails $\alpha$ classically, leading to an \NL-algorithm.
  For the case of $\preceq_c$-minimal \brave-hypotheses, we first look for a singleton hypothesis $\{\beta\} \subseteq \Hyp$, which must exist in any hypothesis by \cref{lem:dllite-singleton-hyps}.
  This simplifies the minimality-checks, yielding \NL-algorithms for both cases by carefully comparing conflicts for $\{\alpha\}$ and $\{\beta\}$ and using \cref{lem:dllite-singleton-hyps}.
\end{proof}

The remaining cases under \brave semantics behave similarly, with membership again following from \cref{lem:dllite-singleton-hyps}.
In case of existence of non-trivial hypotheses for some observation $A(a)$, one may be tempted to only consider the direct subsumees of $A$ at least for \brave semantics.
Notably, this is not sufficient, as demonstrated by the following example, which is also incorporated in the construction for hardness.
\begin{example}\label{ex:non-triv}
  Define
  \[\calT \dfn \{B \subsum \exists r, \exists r \subsum A, A \subsum \neg \exists r^-, C \sqsubseteq \neg C\},\]
  and let $\calK \coloneqq \tup{\calT, \{C(a)\}}$.
  Then $\tup{\calK, A(a)}$ is a \brave-abduction problem, as \calK is obviously inconsistent and $\calK \not\models_\brave A(a)$.
  The ABox $\calB \dfn \{r(a,a)\}$, based on the direct subsumee $\exists r$ of $A$, is not a \brave-hypothesis for $\tup{\calK, A(a)}$, as $\tup{\calT, \calB} \models \bot$.
	Still, there is such a hypothesis, namely $\Hyp \dfn \{B(a)\}$.

  Intuitively, not allowing fresh individuals in a hypothesis has the effect that we have to find a subsumee satisfiable with the limited number of individuals present in the KB.
  Note that the last TBox axiom and the assertion $C(a)$ in the ABox are only necessary to ensure inconsistency, and do not interfere with the rest of the construction.
\end{example}

\begin{restatable}{theorem}{dlliteExistBraveSigNontriv}\label{thm:dllite-exist-brave-sig-nontriv}
    The existence problems for \SignatureRestricted and for non-trivial \brave-hypotheses are \NL-complete.
\end{restatable}

\begin{restatable}{theorem}{dlliteVerifBraveSubset}\label{thm:dllite-verif-brave-subset}
    Verification of $\subseteq$-minimal \brave-hypotheses is \NL-complete.
\end{restatable}

The remaining cases of existence of \ar-hypotheses would seem to require guessing a hypothesis and verifying using \ar-entailment, which would result in a \SigmaP-algorithm.
At first, this seems further evidenced by the non-convexity of \ar-hypotheses, seen in \cref{obs:ar-non-convex}:
For example, for \mbox{$\subseteq$-minimality}, it is not sufficient to check all direct subsets only removing one element.
Surprisingly, we can still get around this obstacle, due to the following property.

\begin{restatable}{lemma}{dlliteArHyps}\label{thm:dllite-ar-hyps}
  Let $\tup{\calK, \alpha}$ be an \ar-abduction problem, where $\calK = \tup{\calT, \calA}$ is a \DLLite KB, and let \calB be a set of assertions.
  There is an \ar-hypothesis $\Hyp \subseteq \calB$ for $\tup{\calK, \alpha}$ iff for each repair \calR of \calK, there is a \calT-support $\{\beta\} \subseteq \calA \cup \calB$ of $\alpha$ s.t.\ $\calR \cup \{\beta\}$ is \calT-consistent.
\end{restatable}
\begin{proof}[Proof sketch]
  ($\Leftarrow$) Let $\calR_1, \dots, \calR_m$ be the repairs of \calK and $\beta_1, \dots, \beta_m$ the \calT-supports of $\alpha$, where $\calR_i \cup \{\beta_i\}$ is \calT-consistent.
  Each repair \calR of $\tup{\calT, \calA \cup \{\beta_1, \dots, \beta_m\}}$ must contain at least one of the $\beta_i$, and hence a \calT-support of $\alpha$:
  Otherwise, $\calR \subseteq \calA$ and $\calR \cup \{\beta_j\}$ must be \calT-consistent for some $j$, leading to contradiction.

  ($\Rightarrow$) Any repair \calR of \calK is contained in some repair $\calR'$ of $\tup{\calT, \calA \cup \Hyp}$.
  As \Hyp is an \ar-hypothesis for $\alpha$, there is some \calT-support $\{\beta\} \subseteq \calR'$ of $\alpha$.
  As $\calR \cup \{\beta\} \subseteq \calR'$ and $\calR'$ is a repair, $\calR \cup \{\beta\}$ is \calT-consistent.
\end{proof}

We use this to find hypotheses $\Hyp \subseteq \calB$ in \coNP.
This resolves the remaining cases, starting with the remaining existence problems.

\begin{restatable}{theorem}{dlliteExistArSigNontriv}\label{thm:dllite-exist-ar-sig-nontriv}
  The existence problems for \SignatureRestricted and for non-trivial \ar-hypotheses are \coNP-complete.
\end{restatable}
\begin{proof}[Proof sketch]
  The following is a \coNP-algorithm for existence of \SignatureRestricted \ar-hypotheses by \cref{thm:dllite-ar-hyps}:
  Given a \SignatureRestricted \ar-abduction problem $\tup{\calK, \alpha, \Sigma}$, let \calB be the set of assertions over $\Sigma$.
  Guess a repair of \calK and check for a \calT-support $\{\beta\} \subseteq \calA \cup \calB$ s.t.\ $\calR \cup \{\beta\}$ is \calT-consistent.
  Both checks are in $\NL \subseteq \Ptime$.
  The non-trivial case can be handled with the same algorithm, only changing \calB to be the set of assertions over the signature of $\tup{\calK, \alpha}$, except for $\alpha$ itself.

  Hardness for existence of \SignatureRestricted \ar-hypotheses is shown by adapting the construction used in \cref{thm:dllite-verif-gen-card}: 
  we choose the signature $\Sigma \coloneqq \{\, U, a, c_j \mid j \leq k \,\}$ and prevent the only remaining unintended hypothesis $\{U(a,a)\}$ by adding the axiom $\exists U \sqsubseteq \neg \exists U^-$.

  To show hardness for non-trivial \ar-hypotheses, we instead reduce from checking whether a given $\forall$-QBF $\forall X \varphi(X)$ is true, where $\varphi = \{C_1, \dots, C_k\}$ is in DNF.
  We encode $\varphi$ into the KB $\calK = \tup{\calT, \calA}$ as follows:
  \begin{align*}
    \calT \coloneqq {} &\{\, C_j \sqsubseteq A_\varphi \mid 1 \leq j \leq m \,\} \cup \{\, T_{x_i} \sqsubseteq \neg C_j \mid \neg x_i \in C_j \,\} \cup \{\, F_{x_i} \sqsubseteq \neg C_j \mid x_i \in C_j \,\} \cup {} \\
    & \{\, T_{x_i} \sqsubseteq \neg F_{x_i} \mid 1 \leq i \leq n \,\} \text{ and} \\
    \calA \coloneqq {} & \{\, T_{x_i}(a), F_{x_i}(a) \mid 1 \leq i \leq n \,\}.
  \end{align*}
  Here, \calT encodes that $\varphi$ is true, if at least one clause $C_j$ is true, and a clause $C_j$ is falsified when at least one of its literals is falsified.
  It remains to show that $\forall X \varphi(X)$ is true iff there is a non-trivial \ar-hypothesis for $\tup{\calK, A_\varphi(a)}$.
  For this, note that the repairs of \calK correspond to assignments over $X$, and that the only non-trivial assertions that can help entailment of $A(a)$ are assertions of the form $C_j(a)$.
\end{proof}

Even though verification of $\subseteq$-minimal \ar-hypotheses has higher complexity than the previous cases, the upper bound again relies on \cref{thm:dllite-ar-hyps}.
  
\begin{restatable}{theorem}{dlliteVerifArSubset}\label{thm:dllite-verif-ar-subset}
  Verification of $\subseteq$-minimal \ar-hypotheses is \DP-complete. 
\end{restatable}
\begin{proof}[Proof sketch]
  To verify a $\subseteq$-minimal \ar-hypothesis $\Hyp = \{\beta_1, \dots, \beta_n\}$, one needs to check that it is an \ar-hypothesis in \coNP, and that it is $\subseteq$-minimal.
  For the latter, define the direct subsets $\calB_i \coloneqq \Hyp \setminus \{\beta_i\}$, noting that \Hyp is not $\subseteq$-minimal iff there is any \ar-hypothesis $\Hyp' \subseteq \calB_i$ for some $i$.
  Hence, we can check non-minimality in \coNP using $n$ subsequent runs of the algorithm from the membership proof of \cref{thm:dllite-exist-ar-sig-nontriv}, replacing $\calB$ by the different $\calB_i$.
  This yields an \NP-algorithm for $\subseteq$-minimality, and \DP-membership in total.
  For hardness, we reuse the construction in the proof of~\cref{thm:dllite-verif-gen-card}.
  Specifically, we reduce from the known \DP-hard problem of checking whether a set of clauses $\psi \subseteq \varphi$ is a \emph{minimal unsatisfiable subset} (MUS) of $\varphi$~\cite{Liberatore05}.
\end{proof}

\section{The Case of \ELbot} 
\label{sec:elbot}

For complexity in \ELbot, we begin with cases that behave analogous to the corresponding cases for \DLLite.
Verification of general hypotheses is again essentially equivalent to entailment for both semantics.
Moreover,
\cref{prop:brave-triv} and \cref{lem:ar-triv} allow us to extend results to the $\leq$-minimal case and to handle existence of general hypotheses.

\begin{restatable}{theorem}{elbotVerifGeneral}\label{thm:elbot-verif-general}
  Verification of general \calS-hypotheses as well as $\leq$-minimal \calS-hypotheses is
  \begin{enumerate*}[label=(\arabic*)]
    \item \NP-complete for $\calS = \brave$, and
    \item \coNP-complete for $\calS = \ar$.
  \end{enumerate*}
\end{restatable}

\begin{restatable}{theorem}{elbotExistGeneral}\label{thm:elbot-exist-general}
  The existence problem for general \calS-hypotheses is
  \begin{enumerate*}[label=(\arabic*)]
    \item \Ptime-complete for $\calS = \brave$, and
    \item \coNP-complete for $\calS = \ar$.
  \end{enumerate*}
\end{restatable}

We next address each property of hypotheses in turn. 
An appropriate \SignatureRestriction can disallow the trivial hypothesis, so that for deciding existence in this setting, reasoning for only the trivial hypothesis is not sufficient.
This interestingly leads to existence of \brave-hypothesis still enjoying the same complexity as \brave-entailment, while the complexity for \ar semantics jumps to the second level of the polynomial hierarchy.
\begin{restatable}{theorem}{elbotExistSig}\label{thm:elbot-exist-sig}
  The existence problem for \SignatureRestricted \calS-hypotheses is
  \begin{enumerate*}[label=(\arabic*)]
    \item \NP-complete for $\calS = \brave$, and
    \item \SigmaP-complete for $\calS = \ar$.
  \end{enumerate*}
\end{restatable}
\begin{proof}[Proof Sketch]
	Membership for (1): Guess an ABox \Hyp over signature $\Sigma$ as well as a candidate repair $\calR \subseteq \calA \cup \Hyp$, and verify that \calR is \calT-consistent and $\tup{\calT, \calR} \models \alpha$ in polynomial time.
	Membership for (2): Guess an ABox \Hyp over $\Sigma$ and check that $\tup{\calT, \calA \cup \Hyp} \models_\ar \alpha$.
	The latter can be handled by an \NP-oracle, yielding $\SigmaP$-membership.

	Hardness for (1) follows by reduction from \SAT.
	Let $\varphi(X) \dfn \{c_1,\dots, c_n\}$ be a CNF with clauses $c_i$.
	We model the satisfaction of $\varphi$ by the TBox 
	\begin{align*}
		\calT\dfn {} & \{\,A_x\sqcap A_{\bar x}\subsum \bot \mid x\in X\,\}\cup {}\\
		& \{\,A_\ell\subsum A_c \mid \ell\in c,c\in\varphi\,\}\cup {} \left\{\,\bigsqcap_{c\in\varphi}A_c\subsum A_\varphi\,\right\}.
	\end{align*}
	Let $\alpha \coloneqq A_\varphi(m)$ and $\Sigma \dfn \{A_x,A_{\bar x}\mid x\in X\}\cup\{m\}$.
	The reduction actually works with an empty ABox and hence also applies to the existence of a hypothesis over $\Sigma$ in a consistent KB $\calK$.
	Any hypothesis \calH over $\Sigma$ corresponds to an assignment $\theta_\calH$ over $X$, and entailment $A_\varphi(m)$ in $\tup{\calT, \calA \cup \Hyp}$ is equivalent to $\theta_\calH \models \varphi$.
	
	Hardness for (2) follows by reduction from $\exists\forall$-QBF, utilising a similar construction as in (1).
  Given an $\exists\forall$-QBF $\Phi = \exists Y \forall Z. \varphi$ with $\varphi$ in DNF, the main differences are that we encode satisfaction of a DNF instead of a CNF, and that the assertions encoding an assignment over $Z$ are included in the KB, while those encoding an assignment over $Y$ have to be part of the hypothesis.
\end{proof}

Existence of non-trivial hypotheses turns out to have the same complexity as existence of \SignatureRestricted hypotheses, but for \brave semantics requires us to incorporate the idea from \cref{ex:non-triv} into the reduction for \NP-hardness.

\begin{restatable}{theorem}{elbotExistNontriv}\label{thm:elbot-exist-nontriv}
  The existence problem for non-trivial \mbox{\calS-hypotheses} is
  \begin{enumerate*}[label=(\arabic*)]
    \item \NP-complete for $\calS = \brave$, and
    \item \SigmaP-complete for $\calS = \ar$.
  \end{enumerate*}
\end{restatable}
\begin{proof}[Proof Sketch]
  Membership in both cases is shown similar as in the proof of \cref{thm:elbot-exist-sig}, so we focus on hardness.

	For (1), we build on the construction from the proof of \cref{thm:elbot-exist-sig}.
	The main change is to use the idea from \cref{ex:non-triv} to replace the explicit signature-restriction, enforcing use of the assertions of the form $A_\ell(m)$ for a literal $\ell$ in any hypothesis.
  This can be achieved by replacing each $A_\ell \subsum A_c$ by $A_\ell\subsum \exists r_c.B$ and using the axiom $\bigsqcap_{c\in\varphi} \exists r_c.B\subsum A_\varphi$ to entail $A_\varphi$ instead. 
	Moreover, we add $A_\varphi \sqcap B \subsum \bot$ to trigger a conflict when trying to entail $A_\varphi(m)$ using assertions $r_c(m,m)$ and $B(m)$ instead of the intended assertions of the form $A_\ell(m)$.
  Note that this makes use of the fact that fresh individuals are not allowed, and that that non-triviality rules out using the assertion $A_\varphi(m)$ in the hypothesis.

	As before, the reduction works with an empty ABox and hence also applies to consistent KBs.
	Our TBox takes the following final form:
	\begin{align*}
		\calT\dfn & \{\,A_x\sqcap A_{\bar x}\subsum \bot \mid x\in X\,\}\cup {}\\
		& \{\,A_\ell\subsum \exists r_c.B \mid \ell\in c,c\in\varphi\,\}\cup \left\{\bigsqcap_{c\in\varphi}\exists r_c.B\subsum A_\varphi\right\}.
	\end{align*}
	For (2): we reduce directly from an instance of $\exists\forall$-QBF.
	Let $\Psi\dfn \exists Y\forall Z. \psi$ be a formula  with $\psi\dfn \{t_1,\dots,t_n\}$ a DNF where each $t_i$ is a term over $X=Y\cup Z$. 
	We construct the KB $\calK\dfn \tup{\calT,\calA}$, where 
	\begin{align*}
		\calT\dfn {} & \{\,A_x\sqcap A_{\bar x}\subsum \bot \mid x\in X\, \}\cup {} \\
		& \left\{\, C\sqcap \bigsqcap_{\ell\in t} A_\ell \subsum A_\psi \mid t\in \psi\,\right\}, \text{ and }\\
		\calA\dfn {} & \{\,A_z(m),A_{\bar z}(m) \mid z\in Z\,\}.
	\end{align*}
	Moreover, we let $\alpha\dfn A_\psi(m)$ be our observation.
	Intuitively, although $\{\alpha\}$ is a trivial \ar-hypotheses, any \emph{non-trivial} \ar-hypothesis corresponds to a satisfying assignment over $Y$ for $\Psi$.
	Thus, $\Psi$ is true iff $\alpha$ admits a non-trivial \ar-hypothesis in $\calK$.
\end{proof}

Existence of conflict-confining \ar-hypotheses behaves very similar as for \DLLite, with our proofs mostly relying on \cref{lem:ar-triv}.
In sharp contrast, the setting behaves very different under \brave semantics.
While this seems counter-intuitive, the following observation applies to \ELbot, as demonstrated in the example below.

\begin{restatable}{observation}{obsElbotBraveCc}\label{obs:elbot-brave-cc}
  There is a \brave-abduction problem $\tup{\calK, \alpha}$ s.t.\ there is a conflict-confining \brave-hypothesis for $\tup{\calK, \alpha}$, but $\{\alpha\}$ is not such a hypothesis.
\end{restatable}
\begin{example}\label{ex:brave-cc}
  Let $\calK = \tup{\calT, \calA}$, where
  \begin{align*}
  	\calT &= \{\, A \sqcap B \sqsubseteq \bot, B \sqcap C \sqsubseteq \bot, C \sqcap D \sqsubseteq A \,\}, \text{ and} \\
  	\calA &= \{\, B(a), C(a) \,\},
  \end{align*}
  and let $\alpha = A(a)$. 
  It is easy to see that $\{\alpha\}$ is a \brave-hypothesis for $\tup{\calK, \alpha}$, but results in a new conflict $\{A(a),B(a)\}$ of $\tup{\calT,\calA\cup\{\alpha\}}$.
  Hence, $\{\alpha\}$ is not conflict-confining in $\calK$.
  However, $\calH\dfn \{D(a)\}$ entails $A(a)$ in the repair $\{C(a), D(a)\}$ and is consistent with all repairs of \calK. 
  Intuitively, the only repair of $\calA\cup\calH$ where $\alpha$ is entailed, is the one that already got rid of the conflict with $\alpha$.
\end{example}
This contrasts with the situation for \DLLite, where a \brave-abduction problem with this property does not exist due to \cref{lem:dllite-singleton-hyps}.
Surprisingly, the result is that adding the conflict-confining restriction decreases the complexity under \ar semantics, while it increases it for \brave semantics.
Note that this is the opposite behaviour to what was seen for \SignatureRestriction and non-triviality in Theorems~\ref{thm:elbot-exist-sig} and~\ref{thm:elbot-exist-nontriv}, where complexity increased under \ar semantics and stayed the same for \brave semantics.

\begin{restatable}{theorem}{elbotExistCc}\label{thm:elbot-exist-cc}
  The existence problem for conflict-confining \calS-hypotheses is 
  \begin{enumerate*}[label=(\arabic*)]
	  \item \SigmaP-complete for $\calS=\brave$, and
	  \item \coNP-complete for $\calS=\ar$.
  \end{enumerate*}
\end{restatable}

\begin{proof}[Proof Sketch]
	For membership in (1): although one can guess a hypothesis and a witnessing candidate repair simultaneously, the conflict-confinement of this hypothesis requires oracle calls.
	In contrast, the membership for (2) holds due to the observation that if there exists some conflict-confining \ar-hypothesis, the observation itself must be conflict-confining and hence an \ar-hypothesis. Hence, one only requires to check this later condition for the observation, which can be done in \coNP due to Lemma~\ref{lem:cc-compl}.
	
	For hardness in (1): we reduce from $\exists\forall$-QBF reusing and adapting ideas from the reduction for \ar semantics in the proof of Theorem~\ref{thm:elbot-exist-sig}.
	Intuitively, the outer quantifier simulates the existence of a hypothesis, whereas the inner one now checks the conflict-confinement.
	The required changes include (1) shifting from \ar to \brave entailment, (2) causing new conflicts when a hypothesis contains an assertion for concepts outside $\Sigma$, and (3) encoding the situation when a partial assignment over $Z$ already satisfies the formula. Here (3) is essential to encode that new conflicts are not subsumed by existing ones.
\end{proof}

Verification of conflict-confining and $\subseteq_c$-minimal hypotheses for can be treated using similar ideas.
Again, the complexity under \brave semantics increases when requiring the hypotheses to be conflict-confining (or $\subseteq_c$-minimal), while it stays the same as for verification of general hypotheses in case of \ar, due to the effect of the trivial hypothesis.

\begin{restatable}{theorem}{elbotVerifCc}\label{thm:elbot-verif-cc}
  Verification of conflict-confining \calS-hypotheses is
  \begin{enumerate*}[label=(\arabic*)]
	  \item \DP-complete for $\calS=\brave$, and
	  \item \coNP-complete for $\calS=\ar$.
  \end{enumerate*}
\end{restatable}

\begin{restatable}{theorem}{elbotVerifCminSubset}\label{thm:elbot-verif-cmin-subset}
  Verification of $\subseteq_c$-minimal \calS-hypotheses is
  \begin{enumerate*}[label=(\arabic*)]
    \item \PiP-complete for $\calS = \brave$, and
    \item \coNP-complete for $\calS = \ar$.
  \end{enumerate*}
\end{restatable}
\begin{proof}[Proof Sketch]
	We sketch the proof for (1).
	For membership: Given an ABox $\Hyp$, we check whether \Hyp is a \brave-hypothesis and guess a counter-witness $\Hyp'$ to \Hyp being $\subseteq_c$-minimal.
  Verifying that either \Hyp is not a \brave-hypothesis, or $\Hyp'$ is a \brave-hypothesis and causes fewer conflicts than \Hyp is done using an \NP-oracle.
	This yields membership in $\PiP$.

	For hardness: we reuse the reduction from the proof of Theorem~\ref{thm:elbot-exist-cc}.
	To this aim, we construct the KB $\calK'$ from \calK in the previous construction by adding to the TBox the axioms
  \[C_d \sqcap X \subsum C \quad \text{and} \quad X\sqcap Y\subsum \bot\]
  and to the ABox the assertion $Y(m)$.
	Then, let $\alpha \dfn C(m)$ as before and take $\calH \dfn \{X(m)\}$.
	Intuitively, \calH uses these new axioms to entail $\alpha$ while inducing exactly one new conflict in $\calK'$, namely $\{X(m), Y(m)\}$.
	We then argue that \calH is not a $\subseteq_c$-minimal hypotheses for $\tup{\calK, C(m)}$ iff $\tup{\calK, C(m)}$ admits a conflict-confining \brave-hypothesis.
  For this, observe that the latter is also a conflict-confining \brave-hypothesis for $\tup{\calK', C(m)}$,
  and therefore a hypothesis introducing less conflicts than~\Hyp.
\end{proof}

The case of $\leq_c$-minimal \brave-hypotheses is more challenging: Here, when adapting the algorithm for the case of $\subseteq_c$-minimality, we would have to compare the \emph{number} of conflicts introduced by the given candidate hypothesis \Hyp to the \emph{number} of conflicts introduced by a potential counter-witness.
As it is not hard to see that the number of conflicts can be exponential, this would naively require an oracle for the counting class \cP.
It seems likely that this case requires quite different techniques, and we for now only observe that the problem is certainly in \PSPACE by a straightforward algorithm that counts the number of conflicts, and \PiP-hard by the same proof as for the case of $\subseteq_c$-minimality above.
In contrast, verification of $\leq_c$-minimal \ar-hypotheses can again be handled using similar techniques as above, due to the effect of the trivial hypothesis.

\begin{restatable}{theorem}{elbotVerifCminCard}\label{thm:elbot-verif-cmin-card}
   Verification of $\leq_c$-minimal \ar-hypotheses is \coNP-complete.
\end{restatable}

Finally, $\subseteq$-minimality increases the complexity of verification for both semantics compared to the $\leq$-minimality.
In case of \ar semantics, the non-convexity of the set of \ar-hypotheses comes into play.
In contrast to the case of \DLLite, where \cref{thm:dllite-ar-hyps} allowed us to circumvent this, we can here combine the idea of \cref{ex:ar-non-convex} with an encoding of a $\forall\exists$-QBF to show \PiP-hardness using conjunction.

\begin{restatable}{theorem}{elbotVerifSubset}\label{thm:elbot-verif-subset} 
  Verification of $\subseteq$-minimal $\calS$-hypotheses is
  \begin{enumerate*}[label=(\arabic*)]
    \item \DP-complete for $\calS=\brave$, and
    \item \PiP-complete for $\calS=\ar$.
  \end{enumerate*}
\end{restatable}

\begin{proof}[Proof sketch]
  Membership for both semantics is relatively straightforward.
  For hardness under \brave semantics, we use a reduction from the combination of an entailment (\NP-hard) and a non-entailment question (\coNP-hard) under \brave semantics, which is \DP-hard. 

  Hardness for $\calS = \ar$ is more involved.
  Here, we reduce a $\forall\exists$-QBF $\Phi = \exists Y \forall Z \varphi(Y,Z)$ to an \ar-abduction problem $\tup{\calK, A(a)}$ and ABox \Hyp s.t.\ \Hyp is always an \ar-hypothesis of $\tup{\calK, A(a)}$, and there is some \ar-hypothesis $\Hyp' \subsetneq \Hyp$ of $\tup{\calK, A(a)}$ iff $\Phi$ is true.
  Intuitively, subsets $\Hyp' \subsetneq \Hyp$ encode assignments over $Y$.

  The construction builds on two main ideas:
  First, two disjoint concepts $B_1$ and $B_2$ split the set of repairs into two parts.
  Each repair \calR in the first part encodes an assignment over $Z$, and entailment of $A(a)$ in $\tup{\calT, \calR}$ is equivalent to $\varphi$ being true under assignment encoded by $\Hyp'$ and \calR.
  The second part contains a single repair ensuring that $\Hyp'$ encodes a full assignment over $Y$.
  Second, building on the idea from \cref{ex:ar-non-convex}, we use an additional axiom to ensure that \Hyp is always an \ar-hypothesis for $\tup{\calK, A(a)}$ without interfering with the rest of the construction.
\end{proof}

\section{Conclusion and Outlook}

We extend abduction from the classical setting to KBs that admit inconsistencies.
This is achieved by designing properties that govern abductive hypotheses under commonly studied repair semantics.
We present a comprehensive, detailed, and almost complete complexity landscape under all considered criteria, with only the complexity of verification for $\leq_c$-minimal \brave-hypotheses in \ELbot remaining open.
In summary, abduction for \DLLite behaves largely similarly to entailment under the corresponding semantics, with certain cases where it becomes even easier (e.g., the existence of conflict-confining \ar-hypotheses).
Nevertheless, membership in many cases for \DLLite is derived using new insights for considered properties.
For \ELbot, the situation differs, as there are cases where the complexity exceeds the one of the corresponding entailment problem under both semantics.
Our complexity classification highlights how different properties of \emph{preferred} hypotheses affect the considered (\brave and \ar) semantics.

Additional findings are certain effects observed for abduction under the considered semantics.
We mention two notable cases here.
First, \emph{non-convexity} (Observation~\ref{obs:ar-non-convex}) applies to both DLs but leads to higher complexity only for \ELbot, as a workaround exists for \DLLite (see Lemma~\ref{thm:dllite-ar-hyps}).
Second, an observation may admit a conflict-confining \brave-hypothesis even if the observation itself is not conflict-confining for \ELbot KBs (Observation~\ref{obs:elbot-brave-cc}).
Regarding \ELbot, it is also worth noting how the complexity landscape compares to the setting of consistent KBs. 
While brave semantics often behaves as in the consistent case, conflict-confinement increases complexity and has no meaningful counterpart in the consistent setting. 
Moreover, verification in nearly all cases incurs higher complexity, in contrast to the classical setting, where the problem reduces to classical entailment and can be decided efficiently. 
This can be partially explained again by the non-convex nature of the space of admissible hypotheses.

We see several directions to further pursue this work.
Beyond completing the picture by determining the precise complexity for the only remaining open case, it seems interesting, first and foremost, to study the so-called \emph{data complexity} of abductive reasoning.
Here, one looks at the complexity purely in terms of the ABox while keeping the TBox fixed, motivated by the situation in practice, where the amount of data often far exceeds the size of the TBox.
While several of our results for \DLLite already extend to this setting, as the construction for hardness uses a TBox of constant size (e.g.\ \coNP-hardness in \cref{thm:dllite-verif-gen-card}), a systematic complexity analysis remains open for further exploration.

Studying the combinations of properties and seeing their effects on the complexity is another interesting direction to pursue.
We can already make interesting observations that further motivate this investigation. 
First, considering signature-restriction or non-triviality together with conflict-confinement, 
the complexity for both semantics jumps to the second level of {the} polynomial hierarchy.
\begin{restatable}{corollary}{elbotExistSigCC}\label{cor:existence-cc+sig-el}
	For $\ELbot$, 
	the existence problem for signature-restricted (or non-trivial) conflict-confining $\calS$-hypotheses is $\SigmaP$-complete for $\calS\in\{\brave,\ar\}$. 
\end{restatable}
Interestingly, when considering \SignatureRestriction, one might expect that non-triviality can be obtained ``for free'' by omitting from the signature some of the symbols of the observation, disallowing the trivial hypothesis.
However, a non-trivial hypothesis may reuse the symbols of the observation in a different manner e.g., to satisfy an existential restriction involving the concept from the observation.
This effect is illustrated in the following example, and motivates to also study the combination of non-triviality with \SignatureRestriction.
\begin{example}\label{ex:nt+sig}
	Let $\calK\dfn \tup{\calT,\calA}$ be the KB with $\calT\dfn \{A\sqcap B \subsum C, D\sqcap \exists r.C\subsum A\}$ and $\calA\dfn \{B(m),r(m,n)\}$.
	Moreover, let $\alpha\dfn C(m)$ and $\Sigma \coloneqq \{C,D,m,n\}$.
	Here, the trivial hypothesis $\{C(m)\}$ is in fact also \SignatureRestricted, while the hypothesis $\Hyp_1\dfn\{A(m)\}$ is non-trivial, but outside our signature.
  Finally, $\Hyp \dfn \{C(n),D(m)\}$ is non-trivial and \SignatureRestricted.
  It uses $C$ and $m$, that is, both the concept name and the individual from the observation.
\end{example}
Another prominent future direction is to consider expressive observations in an abduction problem.
This includes extending our analysis from \textit{flat} to \textit{complex} concepts in an observation, as well as to Boolean (conjunctive) queries. The hardness results from our work already transfer, whereas we anticipate increased complexity in certain cases.
Additional future work includes exploring alternative definitions of \emph{preferred} hypotheses, such as semantically minimal ones~\cite{du2015towards}.
In classical abduction, a hypothesis \Hyp is \emph{semantically minimal}, if there exists no hypothesis $\Hyp'$ such that
$\tup{\calT, \calA \cup \Hyp} \models \Hyp'$, but $\tup{\calT, \calA \cup \Hyp'} \not\models\Hyp$.
We argue that while such a minimality criterion is natural for \ar semantics, its meaning is unclear for $\brave$-hypotheses. 
Further exploration of this minimality criterion is therefore left for future work.

Allowing fresh individuals in hypotheses also seems interesting, as it may impact the complexity in certain cases.
Indeed, even some of our hardness proofs rely on the fact that we do not admit fresh individuals 
(e.g., Theorems~\ref{thm:elbot-exist-general} and~\ref{thm:elbot-exist-nontriv}).
With fresh individuals in the hypotheses, an intriguing case arises for $\ELbot$:
(1) existence of non-trivial \brave-hypotheses becomes easy, since one now only needs to consider direct subsumees of the concept appearing in the
observation, and
(2) existence of conflict-confining hypothesis might get harder as there is no obvious polynomial bound for the size of hypotheses in this setting.
Indeed, it is known that admitting fresh individuals in the signature-restricted setting may lead to exponentially large hypotheses~\cite{Koopmann21a}.

\section*{Acknowledgments}

We thank all anonymous reviewers for their valuable feedback that helped us improve the exposition of our paper.
The third author appreciates funding by the  Deutsche Forschungsgemeinschaft (DFG, German Research Foundation), grant TRR 318/1 2021 – 438445824.

\bibliographystyle{plainurl}
\bibliography{bibliography}

@article{Liberatore05,
  author       = {Paolo Liberatore},
  title        = {Redundancy in logic {I:} {CNF} propositional formulae},
  journal      = {Artif. Intell.},
  volume       = {163},
  number       = {2},
  pages        = {203--232},
  year         = {2005},
  url          = {https://doi.org/10.1016/j.artint.2004.11.002},
  doi          = {10.1016/J.ARTINT.2004.11.002}
}

@inproceedings{Rosati11,
  author       = {Riccardo Rosati},
  editor       = {Toby Walsh},
  title        = {On the Complexity of Dealing with Inconsistency in Description Logic
                  Ontologies},
  booktitle    = {{IJCAI} 2011, Proceedings of the 22nd International Joint Conference
                  on Artificial Intelligence},
  pages        = {1057--1062},
  publisher    = {{IJCAI/AAAI}},
  year         = {2011},
  url          = {https://doi.org/10.5591/978-1-57735-516-8/IJCAI11-181},
  doi          = {10.5591/978-1-57735-516-8/IJCAI11-181}
}

@inproceedings{BienvenuIK24,
  author       = {Meghyn Bienvenu and
                  Katsumi Inoue and
                  Daniil Kozhemiachenko},
  editor       = {Pierre Marquis and
                  Magdalena Ortiz and
                  Maurice Pagnucco},
  title        = {Abductive Reasoning in a Paraconsistent Framework},
  booktitle    = {Proceedings of the 21st International Conference on Principles of
                  Knowledge Representation and Reasoning, {KR} 2024},
  year         = {2024},
  url          = {https://doi.org/10.24963/kr.2024/13},
  doi          = {10.24963/KR.2024/13}
}

@book{DBLP:books/daglib/0092426,
	author = {Nicholas Pippenger},
	publisher = {Cambridge University Press},
	title = {Theories of computability},
	year = {1997}}

@inproceedings{lukasiewicz2022explanations,
  author       = {Thomas Lukasiewicz and
                  Enrico Malizia and
                  Cristian Molinaro},
  title        = {Explanations for Negative Query Answers under Inconsistency-Tolerant
                  Semantics},
	booktitle = {Proceedings of the Thirty-First International Joint Conference on
Artificial Intelligence, {IJCAI-22}},
  pages        = {2705--2711},
  publisher    = {ijcai.org},
  year         = {2022},
  doi          = {10.24963/IJCAI.2022/375},
}

@article{bienvenu2019computing,
  title={Computing and explaining query answers over inconsistent {DL-Lite} knowledge bases},
  author={Bienvenu, Meghyn and Bourgaux, Camille and Goasdou{\'e}, Fran{\c{c}}ois},
  journal={Journal of Artificial Intelligence Research},
  volume={64},
  pages={563--644},
  year={2019}
}

@inproceedings{COST-BASED,
  TITLE = {Data Complexity of Querying Description Logic Knowledge Bases under Cost-Based Semantics},
  AUTHOR = {Bienvenu, Meghyn and Maniere, Quentin},
  URL = {https://hal.science/hal-05360907},
  BOOKTITLE = {Proceedings of the Annual AAAI Conference on Artificial Intelligence},
  ADDRESS = {Singapore, Singapore},
  YEAR = {2026},
  publisher = {AAAI Press}
}

@article{peirce1878deduction,
  title={Deduction, induction and hypothesis: Popular Science Monthly, v. 13},
  author={Peirce, CS},
  year={1878}
}

@inproceedings{ignatiev2019abduction,
  title={Abduction-based explanations for machine learning models},
  author={Ignatiev, Alexey and Narodytska, Nina and Marques-Silva, Joao},
  booktitle={Proceedings of the AAAI Conference on Artificial Intelligence},
  volume={33},
  pages={1511--1519},
  year={2019}
}

@inproceedings{SchlobachC03,
  author       = {Stefan Schlobach and
                  Ronald Cornet},
  title        = {Non-Standard Reasoning Services for the Debugging of Description Logic
                  Terminologies},
  booktitle    = {Proceedings of the Eighteenth International Joint Conference
                  on Artificial Intelligence 2003},
  pages        = {355--362},
  publisher    = {Morgan Kaufmann},
  year         = {2003},
  bibsource    = {dblp computer science bibliography, https://dblp.org}
}

@inproceedings{Koopmann21a,
  author       = {Patrick Koopmann},
  title        = {Signature-Based Abduction with Fresh Individuals and Complex Concepts
                  for Description Logics},
  booktitle    = {Proceedings of the Thirtieth International Joint Conference on Artificial
                  Intelligence},
  pages        = {1929--1935},
  publisher    = {ijcai.org},
  year         = {2021},
  url          = {https://doi.org/10.24963/ijcai.2021/266},
  bibsource    = {dblp computer science bibliography, https://dblp.org}
}

@article{Klarman2011,
  title={{ABox} abduction in the description logic {$\mathcal{ALC}$}},
  author={Klarman, Szymon and Endriss, Ulle and Schlobach, Stefan},
  journal={Journal of Automated Reasoning},
  volume={46},
  number={1},
  pages={43--80},
  year={2011},
  publisher={Springer}
}

@inproceedings{BaaderKKN22,
  author       = {Franz Baader and
                  Patrick Koopmann and
                  Francesco Kriegel and
                  Adrian Nuradiansyah},
  title        = {Optimal {ABox} Repair w.r.t. Static {$\mathcal{EL}$ TBoxes}: From Quantified
                  {ABoxes} Back to {ABoxes}},
  booktitle    = {The Semantic Web -- 19th International Conference, {ESWC} 2022},
  pages        = {130--146},
  year         = {2022},
  bibsource    = {dblp computer science bibliography, https://dblp.org},
  publisher    = {Springer},
  series       = {Lecture Notes in Computer Science},
  volume       = {13261},
}

@inproceedings{BienvenuR13,
  author       = {Meghyn Bienvenu and
                  Riccardo Rosati},
  editor       = {Francesca Rossi},
  title        = {Tractable Approximations of Consistent Query Answering for Robust
                  Ontology-based Data Access},
  booktitle    = {{IJCAI} 2013, Proceedings of the 23rd International Joint Conference
                  on Artificial Intelligence, Beijing, China, August 3-9, 2013},
  pages        = {775--781},
  publisher    = {{IJCAI/AAAI}},
  year         = {2013},
  url          = {http://www.aaai.org/ocs/index.php/IJCAI/IJCAI13/paper/view/6904},
  bibsource    = {dblp computer science bibliography, https://dblp.org}
}

@inproceedings{du2015towards,
  author       = {Jianfeng Du and
                  Kewen Wang and
                  Yi{-}Dong Shen},
  title        = {Towards Tractable and Practical {ABox} Abduction over Inconsistent Description
                  Logic Ontologies},
  booktitle    = {Proceedings of the Twenty-Ninth {AAAI} Conference on Artificial Intelligence},
  pages        = {1489--1495},
  publisher    = {{AAAI} Press},
  year         = {2015},
  url          = {https://doi.org/10.1609/aaai.v29i1.9393},
}

@inproceedings{BienvenuB16,
  author       = {Meghyn Bienvenu and
                  Camille Bourgaux},
  title        = {Inconsistency-Tolerant Querying of Description Logic Knowledge Bases},
  booktitle    = {Reasoning Web: Logical Foundation of Knowledge Graph Construction
                  and Query Answering --- 12th International Summer School 2016},
  series       = {Lecture Notes in Computer Science},
  volume       = {9885},
  pages        = {156--202},
  publisher    = {Springer},
  year         = {2016},
  doi          = {10.1007/978-3-319-49493-7\_5},
}

@inproceedings{WeiKleinerDragisicLambrix2014,
  author    = {Fang Wei{-}Kleiner and
               Zlatan Dragisic and
               Patrick Lambrix},
  title     = {Abduction Framework for Repairing Incomplete $\mathcal{EL}$
Ontologies: Complexity
               Results and Algorithms},
  booktitle = {Proceedings of the Twenty-Eighth {AAAI} Conference on Artificial
Intelligence},
  pages     = {1120--1127},
  year      = {2014},
  publisher = {AAAI Press},
  url       = {http://www.aaai.org/ocs/index.php/AAAI/AAAI14/paper/view/8239},
  bibsource = {dblp computer science bibliography, https://dblp.org}
}

@book{DBLP:books/daglib/0041477,
  author       = {Franz Baader and
                  Ian Horrocks and
                  Carsten Lutz and
                  Ulrike Sattler},
  title        = {An Introduction to Description Logic},
  publisher    = {Cambridge University Press},
  year         = {2017}
}

@article{LLRRS-JWS-15,
	author       = {Domenico Lembo and
	Maurizio Lenzerini and
	Riccardo Rosati and
	Marco Ruzzi and
	Domenico Fabio Savo},
	title        = {Inconsistency-tolerant query answering in ontology-based data access},
	journal      = {J. Web Semant.},
	volume       = {33},
	pages        = {3--29},
	year         = {2015},
	doi          = {10.1016/J.WEBSEM.2015.04.002},
}

@inproceedings{Elsenbroich2006,
  author    = {Corinna Elsenbroich and
               Oliver Kutz and
               Ulrike Sattler},
  title     = {A Case for Abductive Reasoning over Ontologies},
  booktitle = {Proceedings of the OWLED'06 Workshop on {OWL:} Experiences and
Directions},
  year      = {2006},
  url       = {http://ceur-ws.org/Vol-216/submission\_25.pdf},
  bibsource = {dblp computer science bibliography, https://dblp.org},
  series    = {{CEUR} Workshop Proceedings},
  volume    = {216},
  publisher = {CEUR-WS.org},
}

@inproceedings{Ceylan2020,
  author       = {{\.I}smail {\.I}lkan Ceylan and
                  Thomas Lukasiewicz and
                  Enrico Malizia and
                  Cristian Molinaro and
                  Andrius Vaicenavicius},
  title        = {Explanations for Negative Query Answers under Existential Rules},
  booktitle    = {Proceedings of the 17th International Conference on Principles of
                  Knowledge Representation and Reasoning, {KR} 2020},
  pages        = {223--232},
  year         = {2020},
  url          = {https://doi.org/10.24963/kr.2020/23},
  publisher    = {{AAAI} Press},
  bibsource    = {dblp computer science bibliography, https://dblp.org}
}

@article{Calvanese2013,
  author    = {Diego Calvanese and
               Magdalena Ortiz and
               Mantas Simkus and
               Giorgio Stefanoni},
  title     = {Reasoning about Explanations for Negative Query Answers in
{DL-Lite}},
  journal   = {J. Artif. Intell. Res.},
  volume    = {48},
  pages     = {635--669},
  year      = {2013},
  url       = {https://doi.org/10.1613/jair.3870},
  doi       = {10.1613/jair.3870},
  bibsource = {dblp computer science bibliography, https://dblp.org}
}

@inproceedings{Haifani2022,
  author       = {Fajar Haifani and
                  Patrick Koopmann and
                  Sophie Tourret and
                  Christoph Weidenbach},
  editorX       = {Jasmin Blanchette and
                  Laura Kov{\'{a}}cs and
                  Dirk Pattinson},
  title        = {Connection-Minimal Abduction in {$\mathcal{EL}$} via Translation to {FOL}},
  booktitle    = {Proceedings of the 11th International Joint Conference on Automated Reasoning {IJCAR}
                  2022},
  series       = {Lecture Notes in Computer Science},
  volume       = {13385},
  pages        = {188--207},
  publisher    = {Springer},
  year         = {2022},
  url          = {https://doi.org/10.1007/978-3-031-10769-6\_12},
  bibsource    = {dblp computer science bibliography, https://dblp.org}
}

@inproceedings{Koopmann2020,
  author       = {Patrick Koopmann and
                  Warren Del{-}Pinto and
                  Sophie Tourret and
                  Renate A. Schmidt},
  title        = {Signature-Based Abduction for Expressive Description Logics},
  booktitle    = {Proceedings of the 17th International Conference on Principles of
                  Knowledge Representation and Reasoning, {KR} 2020},
  pages        = {592--602},
  year         = {2020},
  publisher    = {{AAAI} Press},
  url          = {https://doi.org/10.24963/kr.2020/59},
  bibsource    = {dblp computer science bibliography, https://dblp.org}
}

@inproceedings{du2017practical,
  author    = {Jianfeng Du and
               Hai Wan and
               Huaguan Ma},
  title     = {Practical {TBox} Abduction Based on Justification Patterns},
  booktitle = {Proceedings of the Thirty-First {AAAI} Conference on Artificial
Intelligence},
  booktitleREST={,
               February 4-9, 2017, San Francisco, California, {USA.}},
  pages     = {1100--1106},
  year      = {2017},
  url       = {http://aaai.org/ocs/index.php/AAAI/AAAI17/paper/view/14402},
  bibsource = {dblp computer science bibliography, https://dblp.org}
}

@article{PARACONSISTENT,
  author       = {Frederick Maier and
                  Yue Ma and
                  Pascal Hitzler},
  title        = {Paraconsistent {OWL} and related logics},
  journal      = {Semantic Web},
  volume       = {4},
  number       = {4},
  pages        = {395--427},
  year         = {2013},
  url          = {https://doi.org/10.3233/SW-2012-0066},
  doi          = {10.3233/SW-2012-0066},
  bibsource    = {dblp computer science bibliography, https://dblp.org}
}

@inproceedings{AlrabbaaBFHKKKK24,
  author       = {Christian Alrabbaa and
                  Stefan Borgwardt and
                  Tom Friese and
                  Anke Hirsch and
                  Nina Knieriemen and
                  Patrick Koopmann and
                  Alisa Kovtunova and
                  Antonio Kr{\"{u}}ger and
                  Alexej Popovic and
                  Ida S. R. Siahaan},
  editor       = {Pierre Marquis and
                  Magdalena Ortiz and
                  Maurice Pagnucco},
  title        = {Explaining Reasoning Results for {OWL} Ontologies with {Evee}},
  booktitle    = {Proceedings of the 21st International Conference on Principles of
                  Knowledge Representation and Reasoning, {KR} 2024, Hanoi, Vietnam.
                  November 2-8, 2024},
  year         = {2024},
  url          = {https://doi.org/10.24963/kr.2024/67},
  doi          = {10.24963/KR.2024/67},
  bibsource    = {dblp computer science bibliography, https://dblp.org}
}

\clearpage
\appendix
\section{Full Proofs for \cref{sec:abd}}

\ccCompl*
\begin{proof}
  Let $\calK = \tup{\calT, \calA}$.
  We begin with the case of \DLLite.
  For membership, recall that conflicts in \DLLite are of size at most $2$.
  Hence, we can check whether \calH is conflict-confining for \calK by checking that no pair of assertions from \Hyp is \calT-inconsistent, and that for each pair of some $\gamma_1 \in \calA$ and $\gamma_2 \in \calH$, either $\{\gamma_1\}$ is \calT-inconsistent, or $\{\gamma_1, \gamma_2\}$ is \calT-consistent.
  As consistency and inconsistency can both be checked in \NL for \DLLite, this yields an \NL-algorithm.

  For hardness, we reduce from directed non-reachability.
  Given a directed graph $G = (V,E)$ and $s,t \in V$, define
  \begin{align*}
    \calT \coloneqq {} &\{\, A_v \sqsubseteq A_w \mid (v,w) \in E \text{ and } t \not\in \{v,w\} \,\} \cup {} \\
    &\{A_t \sqsubseteq A_s\} \cup \{\, A_v \sqsubseteq \neg A_t \mid (v,t) \in E \,\}
  \end{align*}
  and $\calK \coloneqq \tup{\calT, \emptyset}$.
  The TBox \calT for the most part simply encodes edges of $G$ as concept inclusions, but removes all edges incident on $t$, adds an edge from $t$ to $s$, and adds disjointness between $t$ and every predecessor of $t$.
  Hence, $\{A_t(a)\}$ is conflict-confining for \calK iff there is no $s$-$t$-path in $G$.
  Note that the reduction from reachability in the proof of \cref{thm:dllite-verif-gen-card} is a simpler variant of this construction.

	We now turn to \ELbot.
	For membership, one can guess a set $\calC\subseteq \calA\cup\calH$ such that (i) $\tup{\calT,\calC}\models \bot$, (ii) $\tup{\calT,\calC'}\not\models \bot$ for any $\calC'=\calC\setminus\{\gamma\}$ with $\gamma\in \calC$, and (iii) $\calC\not\subseteq \calA$.
	The verification can be done in polynomial time and ensures that $\calC$ is indeed a conflict in $\tup{\calT,\calA\cup\calH}$ but not in $\calK$, hence $\calH$ is not conflict-confining in $\calK$.
	
	For hardness, we reduce from non-entailment under \brave semantics.
	Given a KB $\calK \dfn \tup{\calT,\calA}$ and concept assertion $\alpha = C(a)$, let $\calK'\dfn \tup{\calT',\calA}$ where $\calT'\dfn \calT\cup\{A\sqcap C\subsum \bot\}$ for a fresh concept $A$.
	Moreover, define the ABox $\calH\dfn \{A(a)\}$.
	We prove the following claim 
  \begin{claim}
	  $\calK\not\models_\brave C(a)$ iff $\calH$ is conflict-confining in $\calK$.
  \end{claim}
  \begin{claimproof}
	  Suppose $\tup{\calT,\calR}\models C(a)$ for some repair $\calR$ of $\calK$. 
	  Moreover, let $\calC'\subseteq \calR$ be the smallest set such that $\tup{\calT,\calC}\models C(a)$. It follows that $\tup{\calT',\calC} \models \bot$ where $\calC\dfn \calC'\cup\calH$. As a result, $\calC$ is a new conflict in $\calK'$ and hence $\calH$ is not conflict-confining in $\calK'$.
	  
	  Conversely, suppose $\calK\not\models_\brave C(a)$.
	  Then, there is no $\calC\subseteq \calA$ such that $\tup{\calT,\calC}\models C(a)$, as otherwise one can extend $\calC$ to a repair entailing $C(a)$.
	  Moreover, there is also no set $\calC\subseteq \calA\cup\calH$ such that $\tup{\calT',\calC}\models C(a)$. 
	  Consequently, no new conflict arises due to $\calH$ in $\calK'$ since such a conflict must entail $C(a)$ in $\calT'$. We conclude that $\calH$ is conflict-confining in $\calK'$.
  \end{claimproof}
\end{proof}

\arTriv*
\begin{proof}
  The implication \emph{(2) $\Rightarrow$ (1)} is obvious.

  \emph{(3) $\Rightarrow$ (2)}:
	Let $\calK_\alpha \coloneqq \tup{\calT, \calA \cup \{\alpha\}}$.
	We show that $\calK_\alpha \models_{\text{AR}} \alpha$, that is, $\tup{\calT, \calR} \models \alpha$ for all repairs $\calR$ of $\calK_\alpha$.
  Consider any repair \calR of $\calK_\alpha$.
  It is sufficient to show that $\alpha \in R$, as this readily implies $\tup{\calT, \calR} \models \alpha$.
  For the sake of contradiction, assume $\alpha \not\in \calR$.
	As $\calR$ is a repair of $\calK_\alpha$, it is also a repair of $\calK$.
	As $\{\alpha\}$ is conflict-confining, we have $\tup{\calT, \calR \cup \{\alpha\}} \not\models \bot$.
	Therefore, $\calR$ is not maximally consistent, which is a contradiction.
	
  If $\{\alpha\}$ is conflict-confining, then it is $\preceq$-minimal by our assumption that $\calK\not\models_\ar\alpha$	and the fact that $\{\alpha\}$ is of size $1$.
  Furthermore, it is $\preceq_c$-minimal by definition as it does not introduce any additional conflicts.

  \emph{(1) $\Rightarrow$ (3)}:
  We show this by proving the contrapositive.
	Let $\calK = \tup{\calT, \calA}$.
  If $\{\alpha\}$ is not conflict-confining for $\calK$, there is a repair $\calR$ of $\calK$ such that $\tup{\calT, \calR \cup \{\alpha\}} \models \bot$.
	For the sake of contradiction, assume that there is a hypothesis $\Hyp$ of $\alpha$ in $\calK$, that is: for all repairs $\calR''$ of $\calK_{\Hyp} \coloneqq \tup{\calT, \calA \cup \Hyp}$, we have $\tup{\calT, \calR''} \models \alpha$.
	As $\calR$ is consistent with $\calT$, there is some repair $\calR'$ of $\calK_{\Hyp}$ with $\calR \subseteq \calR'$.
	But then we have $\tup{\calT, \calR'} \models \tup{\calT, \calR \cup \{\alpha\}} \models \bot$, which is a contradiction.
\end{proof}

\section{Full Proofs for \cref{sec:dllite}}

\dlliteVerifGenCard*
\begin{proof}
  We begin by showing membership.
  Given an \calS-abduction problem $\tup{\calK, \alpha}$ for some \DLLite KB $\calK = \tup{\calT, \calA}$, and an ABox \Hyp, we can verify that \Hyp is an \calS-hypothesis for $\tup{\calK, \alpha}$ by checking that $\tup{\calT, \calA \cup \Hyp} \models_\calS \alpha$.
  Hence, the complexity of \calS-entailment yields an upper bound of \NL for $\calS = \brave$ and \coNP for $\calS = \ar$.
  Further, any $\leq$-minimal \calS-hypothesis is of size $1$ by \cref{lem:ar-triv} and \cref{lem:dllite-singleton-hyps} for $\calS = \brave$ and $\calS = \ar$, resp.
  Consequently, for $\leq$-minimality we only need to check that $|\Hyp| = 1$, which does not change the complexity.

  We next show \NL-hardness for $\calS = \brave$, for both general and $\leq$-minimal hypotheses.
  This can be done by a straightforward reduction from reachability in directed graphs, more direct than the one from directed non-reachability in the proof of \cref{lem:cc-compl}.
	Let $G = (V,E)$ be a directed graph and $s,t \in V$.
	Define $\calT' \coloneqq \{\, A_v \sqsubseteq A_w \mid (v,w) \in E \,\}$.
	Now there is an $s$-$t$-path in $G$ iff $\tup{\calT', \{A_s(a)\}} \models A_t(a)$.
	To obtain a \brave-abduction problem, we add an artifical inconsistency.
	Let $\calK \coloneqq \tup{\calT, \calA}$, where $\calT \coloneqq \calT' \cup \{B_1 \sqsubseteq \neg B_2\}$ and $\calA \coloneqq \{B_1(b), B_2(b)\}$.
  Further, let $\Hyp \coloneqq \{A_s(a)\}$.
	Obviously, $\calK \models \bot$ and $\calK \not\models_\brave A_t(a)$.
  It is easy to see that this reduction is correct, which is stated in the following claim for later reference.
  \begin{claim}\label{cl:redu-reach}
    There is an $s$-$t$-path in $G$ iff \Hyp is a \brave-hypothesis for $\tup{\calK, A_t(a)}$.
  \end{claim}
  Further, since \Hyp is a singleton set, \Hyp is a \brave-hypothesis iff it is a $\leq$-minimal \brave-hypothesis.
  The reduction can be computed in logarithmic space, showing \NL-hardness under logspace many-one reductions.

  We now show \coNP-hardness for verification of general hypotheses in case of $\calS = \ar$.
  Here, we reuse the following reduction from unsatisfiability to AR-entailment~\cite{bienvenu2019computing}.
	Let $\varphi = \{c_1, \dots, c_k\}$ over propositions $X=\{x_1,\dots, x_n\}$, where the $c_i$ are clauses.
  We construct $\calK = \tup{\calT, \calA}$ using a single concept name $A$ and role names $N = \{P, N, U\} $, defining
	\begin{align*}
	  \calT \coloneqq {} & \{\exists U \subsum A, \exists P^{-} \subsum \neg \exists N^{-} \} \cup {} \\
    & \{ \exists P \subsum \neg \exists U^{-}, \exists N \subsum \neg \exists U^{-} \}, \text{ and} \\
	  \calA \coloneqq {} & \{P(c_j,x_i) \mid x_i \in c_j \} \cup \{N(c_j,x_i) \mid \neg x_i \in c_j \} \cup {} \\
    & \{U(a,c_j) \mid j \leq k \}.
  \end{align*}
	Moreover, let $\alpha \dfn A(a)$.
  The correctness, stated in the following fact, was shown in~\cite{bienvenu2019computing}.
	\begin{fact}\label{fac:redu-meghyn}
	  It holds that $\calK \models_\ar A(a)$ iff $\varphi$ is unsatisfiable.
	\end{fact}
	To show hardness of the verification problem at hand, let $\Hyp \dfn \{U(a,c_j) \mid j \leq k\}$ and $\calK' \dfn \tup{\calT, \calA \setminus \Hyp}$.
  Clearly, we have $\calK' \models \bot$ and $\calK' \not\models_\ar \alpha$.
	Further, it is easy to see that $\Hyp$ is an $\ar$-hypothesis for $\alpha$ in $\calK'$ iff $\tup{\calT, \calA \cup \Hyp} \models_\ar \alpha$ iff $\varphi$ is unsatisfiable.

  Finally, we adapt the above reduction to the case of $\leq$-minimal hypotheses.
  More precisely, we modify the given CNF-formula before applying the reduction to ensure that a specific singleton ABox is an \ar-hypothesis iff the CNF-formula is unsatisfiable.
  Let $\varphi = \{c_1, \dots, c_k\}$ over variables $X = \{x_1, \dots, x_n\}$.
  Define
  \begin{align*}
    c_i' &\coloneqq c_i \cup \{x_{n+1}\} \text{ for } 1 \leq i \leq k, \\
    c_{k+1}' &\coloneqq \neg x_{n+1} \lor x_{n+2}, \text{ and} \\
    c_{k+2}' &\coloneqq \neg x_{n+2}
  \end{align*}
  and let $\varphi_1 \coloneqq \{c_1', \dots, c_{k+1}'\}$ and $\varphi_2 \coloneqq \varphi_1 \cup \{c_{k+2}'\}$.
  Analogously to the construction of \calK from $\varphi$ in the hardness proof for general hypotheses above, we construct knowledge bases $\calK_i = \tup{\calT,\calA_i}$ from $\varphi_i$ for $i \in \{1,2\}$.
  Further, define $\calK_2' \coloneqq \tup{\calT, \calA_2 \setminus \{U(a, c_{k+2})\}}$.
  The following claim now establishes $\coNP$-hardness. 
  \begin{claim}
    $\tup{\calK_2', A(a)}$ is a valid \ar-abduction problem and $\Hyp = \{U(a, c_{k+2})\}$ is a ($\leq$-minimal) solution to it iff $\varphi$ is unsatisfiable.
  \end{claim}
  \begin{claimproof}
    Obviously, every satisfying asignment of $\varphi_2$ assigns $x_{n+2}$ and $x_{n+1}$ to $0$, and hence $\varphi_2$ and $\varphi$ are equisatisfiable.
    The formula $\varphi_1$, however, is always satisfiable, as we can simply assign both $x_{n+1}$ and $x_{n+2}$ to $1$.
    First, this means that that $\varphi$ is unsatisfiable iff $\calK_2 \models_\ar A(a)$, because of correctness of the original reduction and equisatisfiability of $\varphi$ and $\varphi_2$. 
    Secondly, again by correctness of the original reduction, it means that $\calK_1 \not\models_\ar A(a)$.
    Since $\calK_2'$ only adds the assertion $N(c_{k+2}, x_{n+2})$ compared to $\calK_1$, and it is easy to see that this assertion cannot help entailment of $A(a)$ under \ar semantics, $\calK_1 \not\models_\ar A(a)$ implies $\calK_2' \not\models_\ar A(a)$.
    Combining these two observations, it follows that $\tup{\calK_2', A(a)}$ is a valid \ar-abduction problem, and $\calH \coloneqq \{U(a, c_{k+2})\}$ is a solution to it iff $\varphi$ is unsatisfiable.
    Furthermore, if $\calH$ is a solution, then it is also $\leq$-minimal as $|\calH| = 1$.
  \end{claimproof}
\end{proof}

\dlliteSingletonHyps*
\begin{proof}
  By definition, we have $\tup{\calT, \calA \cup \Hyp} \models_\brave \alpha$, so there is a repair \calR of $\tup{\calT, \calA \cup \Hyp}$ s.t.
  \[\tup{\calT, \calR} \models \alpha.\]
  Since \calT is a \DLLite TBox, minimal $\calT$-supports of $\alpha$ are of size $1$.
  Combined with the fact that $\calK \not\models_\brave \alpha$, this implies that there is some $\beta \in \calR \cap \Hyp$ s.t.\ $\tup{\calT, \{\beta\}} \models \alpha$, so $\{\beta\}$ is a \brave-hypothesis for $\tup{\calK, \alpha}$.

  Further, consider any conflict $\mathcal{C}$ of $\tup{\calT, \calA \cup \{\beta\}}$.
  As $\mathcal{C} \subseteq \calA \cup \{\beta\} \subseteq \calA \cup \Hyp$, it is also a conflict of $\tup{\calT, \calA \cup \Hyp}$.
  Finally, consider a conflict $\{\alpha, \gamma\} \in \Conf(\tup{\calT, \calA \cup \{\alpha\}})$.
  Since $\tup{\calT, \{\beta\}} \models \alpha$, we have $\{\beta, \gamma\} \in \Conf(\tup{\calT, \calA \cup \{\beta\}})$.
\end{proof}

\dlliteExistCc*
\begin{proof}
  Consider an \calS-abduction problem $\tup{\calK, \alpha}$, where $\calK = \tup{\calT, \calA}$ is a \DLLite KB and $\alpha = A(a)$.
  We begin with the case of existence of general \brave-hypotheses.
  For this, it is sufficient to check whether $A$ is satisfiable w.r.t.\ \calT:
  If this is the case, then $\{\alpha\}$ is a \brave-hypothesis for $\tup{\calK, \alpha}$.
  Otherwise, there is no such hypothesis.
  Satisfiability for \DLLite is in \NL, yielding the desired upper bound.

  For \NL-hardness of existence of general hypotheses, we use a reduction from directed non-reachability, similar to the reduction seen in the proof of \cref{lem:cc-compl}.
  For a directed graph $G = (V,E)$ and $s,t \in V$, define
  \begin{align*}
    \calT \coloneqq {} &\{\, A_v \sqsubseteq A_w \mid (v,w) \in E \text{ and } t \in \{v,w\} \,\} \cup {} \\
    &\{A_t \sqsubseteq A_s\} \cup \{\, A_v \sqsubseteq \neg A_t \mid (v,t) \in E \,\}.
  \end{align*}
  We add a dummy inconsistency to obtain the KB \calK from \calT.
  Obviously, $\calK \models \bot$ due to the dummy inconsistency, and $\calK \not\models_\brave A_t(a)$.
  \begin{claim}\label{cl:redu-reach-2}
    There is no $s$-$t$-path in $G$ iff $A_t$ is satisfiable w.r.t.\ \calT. 
  \end{claim} 
  This implies that there is no $s$-$t$-path in $G$ iff there is a \brave-hypothesis for $\tup{\calK, A_t(a)}$, and hence correctness of the reduction.
  
  The remaining three cases can be handled simultaneously.
  By \cref{lem:ar-triv}, there is any \ar-hypothesis for $\tup{\calK, \alpha}$ iff $\{\alpha\}$ is conflict-confining for \calK.
  \cref{lem:dllite-singleton-hyps} yields an analogous result for conflict-confining \brave-hypotheses: If there is any \brave-conflict-confining hypothesis for $\tup{\calK, \alpha}$, then $\alpha$ does not introduce any conflicts by the last part of the lemma.
  Consequently, determining the complexity of checking whether $\{\alpha\}$ is conflict-confining resolves the complexity for the remaining cases.
  \cref{lem:cc-compl} now yields \NL-completeness, as hardness in the proof of that lemma already holds for singleton ABoxes.
\end{proof}

\dlliteVerifCcCmin*
\begin{proof}
  We now turn to verification of conflict-confining \calS-hypotheses, beginning with membership for both semantics.
  For $\calS = \brave$, we need to check for a given ABox \Hyp whether $\tup{\calT, \calA} \models_\brave \alpha$ and \Hyp is conflict-confining for \calK.
  The former is \brave entailment and hence in \NL, while the latter can be checked similar to how we checked that $\{\alpha\}$ is conflict-confining in the proof of \cref{thm:dllite-exist-cc}:
  As conflicts for \DLLite are of size at most $2$, \Hyp is not conflict-confining for \calK iff \Hyp is \calT-inconsistent or there is a pair $(\beta_1, \beta_2) \in \calA \times \Hyp$ s.t.\ $\beta_1$ is \calT-consistent and $\{\beta_1, \beta_2\}$ is \calT-inconsistent.
  As we can iterate over all pairs from $\calA \times \Hyp$ in logarithmic space and both consistency and inconsistency can be checked in \NL, this yields an algorithm running in \NL in total. 
  
  For $\calS = \ar$, it turns out that conflict-confining hypotheses have a simpler structure than general hypotheses, resulting in a lower complexity than that for verification of general hypotheses.
  This property is stated in the following claim.
  \begin{claim}
    \Hyp is a conflict-confining \ar-hypothesis for $\tup{\calK, \alpha}$ iff \Hyp is conflict-confining for \calK and $\tup{\calT, \Hyp} \models \alpha$.  
  \end{claim}
  \begin{claimproof}
    The direction from right to left is obvious.
    For the direction from left to right, note that $\calK \not\models_\ar \alpha$.
    Hence, there is a repair \calR of \calK s.t.\ $\tup{\calT, \calR} \not\models \alpha$.
    As \Hyp is conflict-confining, the ABox $\calR \cup \Hyp$ is \calT-consistent, which implies that $\calR \cup \Hyp$ is a repair of $\tup{\calT, \calA \cup \Hyp}$.
    Since $\tup{\calT, \calA \cup \Hyp} \models_\ar \alpha$, we have $\tup{\calT, \calR \cup \Hyp} \models \alpha$.
    As minimal \calT-supports are of size $1$, there is some $\beta \in \calR \cup \Hyp$ s.t.\ $\tup{\calT, \{\beta\}} \models \alpha$.
    If $\beta \in \calR$, we would have $\tup{\calT, \calR} \models \alpha$, so we must have $\beta \in \Hyp$.
    Consequently, we have $\tup{\calT, \Hyp} \models \alpha$.
  \end{claimproof}
  Consequently, it is sufficient to check that \Hyp is conflict-confining for \calK and $\tup{\calK, \Hyp} \models \alpha$.
  The former can be done in \NL using the algorithm from the case $\mathcal{S} = \brave$ above, while the latter is classical entailment and hence also in \NL, resulting in an \NL-algorithm in total.

  Hardness for both semantics can be shown using the same reduction from reachability as in the proof of \cref{thm:dllite-verif-gen-card}, with correctness following from \cref{cl:redu-reach} and the fact that the constructed ABox \Hyp does not introduce any conflicts.

  It remains to handle the case of $\preceq_c$-minimal hypotheses.
  For $\calS = \ar$, \cref{lem:ar-triv} tells us that there is any \ar-hypothesis iff there is a conflict-confining one.
  Hence, an ABox \Hyp is a $\preceq_c$-minimal \ar-hypothesis for $\tup{\calK, \alpha}$ iff it is a conflict-confining \ar-hypothesis, so the problem inherits the complexity from verification of conflict-confining \ar-hypotheses.
 
  We now turn to membership for the case $\calS = \brave$.
  By \cref{lem:dllite-singleton-hyps} there is some \brave-hypothesis $\{\beta\} \in \Hyp$ of $\alpha$.
  Hence, \Hyp can only be $\leq_c$-minimal, if $\Hyp \setminus \{\beta\}$ does not add any additional conflicts over $\beta$.
  We use this as follows:
  Iterate over all $\beta \in \Hyp$ in logarithmic space and check that $\{\beta\}$ is a \brave-hypothesis for $\tup{\calK, \alpha}$ and $\Hyp \setminus \{\beta\}$ is conflict-confining for $\tup{\calT, \calA \cup \{\beta\}}$.
  Both checks can be done in \NL by previous considerations, yielding an \NL-algorithm in total.
  Finally, we need to check that $\beta$ introduces the same number of conflicts as $\alpha$.
  We argue \NL-membership by showing that the complement of this check is in \NL.
  By the last part of \cref{lem:dllite-singleton-hyps}, $\beta$ introduces more conflicts than $\alpha$ iff there is any assertion $\gamma \in \calA$ s.t.\ $\{\alpha, \gamma\}$ is \calT-consistent and $\{\beta, \gamma\}$ is \calT-inconsistent.
  This can again be checked in \NL.
  
  Regarding membership for $\subseteq_c$-minimal \brave-hypotheses, we can again first look for some $\beta \in \Hyp$ s.t. $\{\beta\}$ is a \brave-hypothesis for $\tup{\calK, \alpha}$ and $\Hyp \setminus \{\beta\}$ is conflict-confining for $\tup{\calT, \calA \cup \{\beta\}}$.
  It remains to check whether $\{\beta\}$ is a $\subseteq_c$-minimal \brave-hypothesis.
  For this, we check whether there is any conflict-confining \brave-hypothesis for $\tup{\calK, \alpha}$ in \NL using the algorithm for existence above.
  If this is the case, it means that $\{\beta\}$ is $\subseteq_c$-minimal iff it is conflict-confining, which can again be checked in \NL.
  Otherwise, $\{\beta\}$ is trivially $\subseteq_c$-minimal, as any \brave-hypothesis introduces some new conflict, and any brave hypothesis introducing any conflict containing $\beta$ introduces at least all the conflicts introduced by $\{\beta\}$.

  Hardness for $\preceq_c$-minimal \brave-hypotheses can be shown using the same reduction from reachability as in the proof of \cref{thm:dllite-verif-gen-card} (and reused for conflict-confining hypotheses above), with correctness again following from \cref{cl:redu-reach} and the fact that the constructed ABox \Hyp does not introduce any conflicts.
\end{proof}

\dlliteExistBraveSigNontriv*
\begin{proof}
  We begin with membership for \SignatureRestricted hypotheses.
  Let $\tup{\calK, \alpha, \Sigma}$ be a \SignatureRestricted \brave-abduction problem.
  Let $\Hyp_m$ be the set of all assertions over $\Sigma$ and individuals from $\calK$ and $\alpha$.
  Using \cref{lem:dllite-singleton-hyps}, it is easy to see that there is a \brave-hypothesis over $\Sigma$ for $\alpha$ in \calK iff for any $\beta \in \Hyp_m$, $\{\beta\}$ is such a hypothesis.
  As verification of \brave-hypotheses is in \NL, the upper bound follows.

  Hardness for \SignatureRestricted hypotheses follows from a straightforward reduction from reachability, using the construction from the proof of \cref{thm:dllite-verif-gen-card}.
  It is easy to see that for this construction, \Hyp is a \brave-hypothesis for $\tup{\calK, \alpha}$ iff there is any \brave-hypothesis over signature $\{A_s\}$, so correctness follows from \cref{cl:redu-reach}.

  We now turn to the case of non-trivial hypotheses.
  Let $\tup{\calK, \alpha}$ be a \brave-abduction problem.
  Membership can be shown similar as for \SignatureRestricted hypotheses above, by simply defining $\Hyp_m$ as the set of assertions over the individuals, concept names and role names occurring in \calK and $\alpha$, except for the trivial hypothesis $\{\alpha\}$.

  For \NL-hardness we adapt the reduction from directed reachability first seen in the proof of \cref{thm:dllite-verif-gen-card}.
  We prevent using unintended hypotheses by using only a single individual and using existential restrictions that cannot be satisfied by that individual alone.
  In detail, this can be done as follows.
  Consider a directed graph $G = (V,E)$ and some $s, t \in V$ with $s \neq t$.
  (The case $s=t$ can be handled separately.)
  Let $\calK = \tup{\calT, \calA}$, where
  \begin{align*}
    \calT \coloneqq {} & \{A_s \sqsubseteq \exists r_s, \quad \exists r_t \sqsubseteq A_t \} \cup {} \\
    & \{\, \exists r_u \sqsubseteq \exists r_v \mid (u,v) \in E \,\} \cup {} \\
    & \{\, \exists r_u \sqsubseteq \neg \exists r_u^- \mid u \in V \,\} \cup {} \\
    & \{B_1 \sqsubseteq \neg B_2\} \text{ and} \\
    \calA \coloneqq {} & \{B_1(a), B_2(a)\}
  \end{align*}
  The following claim shows correctness of the reduction.
  \begin{claim}
    There is a non-trivial \brave-hypothesis for $\tup{\calK, A_t(a)}$ iff there is an $s$-$t$-path in $G$.
  \end{claim}
  \begin{claimproof}
    ($\Leftarrow$) In this case, it is easy to see that $A_s(a)$ is a non-trivial \brave-hypothesis for $\tup{\calK, A_t(a)}$.

    ($\Rightarrow$) Consider some non-trivial \brave-hypothesis \Hyp for $\tup{\calK, A_t(a)}$.
    By definition, there is some repair \calR of $\tup{\calT, \calA \cup \Hyp}$ s.t.\ $\tup{\calT, \calR} \models A_t(a)$.
    Since minimal \calT-supports are of size $1$, there is some $\beta \in \calR$ s.t.\ $\tup{\calT, \{\beta\}} \models A_t(a)$.
    Further, we have $\beta \in \Hyp$, since $\calK \not\models_\brave A_t(a)$.
    As $a$ is the only occurring individual, $\beta$ must be of the form $B(a)$ or $r(a,a)$ for some concept name $B$ or role name $r$.
    Note that $r(a,a)$ is \calT-inconsistent for all $r \in \{\, r_u \mid u \in V \,\}$, so the only \calT-support of this form is $A_s(a)$ and we have $\beta = A_s(a)$.
    Consequently, we have $\tup{\calT, \{A_s(a)\}} \models A_t(a)$ and by construction of \calT there is an $s$-$t$-path in $G$.
  \end{claimproof} 
\end{proof}

\dlliteVerifBraveSubset*
\begin{proof} 
  Let $\tup{\calK, \alpha}$ be a \brave-abduction problem and \Hyp an ABox.
  By \cref{lem:dllite-singleton-hyps}, \Hyp is a $\subseteq$-minimal \brave-hypothesis iff it is a $\leq$-minimal \brave-hypothesis, so the complexity is inherited from verification of $\leq$-minimal \brave-hypotheses. 
\end{proof}

\dlliteArHyps*
\begin{proof}
  ($\Leftarrow$) Let $\calR_1, \dots, \calR_m$ be the repairs of $\tup{\calT, \calA}$ and $\beta_1, \dots, \beta_m \in \calA \cup \calB$ be \calT-supports of $\alpha$ s.t.\ $\calR_i \cup \{\beta_i\}$ is \calT-consistent f.a.\ $i$.
  We show that $\{\beta_1, \dots, \beta_m\} \setminus \calA$ is an \ar-hypothesis for $\tup{\calK, \alpha}$.
  Note that $\calA \cup (\{\beta_1, \dots, \beta_m\} \setminus \calA) = \calA \cup \{\beta_1, \dots, \beta_m\}$.
  Consider any repair \calR of $\tup{\calT, \calA \cup \{\beta_1, \dots, \beta_m\}}$.
  If $\beta_i \in \calR$ for some $i$, then $\tup{\calT, \calR} \models \alpha$.
  Otherwise, $\calR \subseteq \calA$.
  But then \calR is a repair of $\tup{\calT, \calA}$, so $\calR = \calR_j$ for some $j$.
  Hence, $\calR \cup \{\beta_j\}$ is \calT-consistent, which is a contradiction to \calR being subset-maximal.

  ($\Rightarrow$)
  Let $\Hyp \subseteq \calB$ be an \ar-hypothesis for $\tup{\calK, \alpha}$.
  Consider any repair \calR of $\tup{\calT, \calA}$.
  Then there is some repair $\calR'$ of $\tup{\calT, \calA \cup \Hyp}$ s.t.\ $\calR \subseteq \calR'$, since \calR is a \calT-consistent subset of $\calA \cup \Hyp$ and each such set is contained in a maximal one.
  As $\tup{\calT, \calA \cup \Hyp} \models_\ar \alpha$, we have $\tup{\calT, \calR'} \models \alpha$.
  As \calT is a \DLLite TBox, minimal \calT-supports are of size $1$.
  Hence, this means that there is some \calT-support $\{\beta\}$ of $\alpha$ s.t.\ $\beta \in \calR' \subseteq \calA \cup \calB$.
  But as $\calR'$ is a repair, this implies that $\calR \cup \{\beta\}$ is \calT-consistent.
\end{proof}

\dlliteExistArSigNontriv*
\begin{proof}
  We begin with the case of existence of \SignatureRestricted hypotheses.
  Membership in \coNP is shown by Algorithm~\ref{alg:dllite-ex-sig}.
  \begin{algorithm}[t]
 	  \DontPrintSemicolon
 	  \SetKwInOut{Input}{Input}
 	  \caption{\coNP-algorithm for existence of \SignatureRestricted \ar-hypotheses in case of \DLLite}\label{alg:dllite-ex-sig}
    \Input{\SignatureRestricted \ar-abduction problem $\tup{\calK, \alpha, \Sigma}$ for \DLLite KB $\calK = \tup{\calT, \calA}$}
    $\calB \gets$ set of all assertions over $\Sigma$\;
    universally guess a set $\calR \subseteq \calA$\;\label{ln:dllite-ex-sig-calB}
    \If{\calR is \calT-inconsistent\label{ln:dllite-ex-sig-repair1}}{
      {\bf accept}
    }
    \ForAll{extensions $\calR'$ of $\calR$ by a single new element from \calA}{
      \If{$\calR'$ is \calT-consistent}{
        {\bf accept}\label{ln:dllite-ex-sig-repair2}
      }
    }
    \ForAll{$\beta \in \calA \cup \calB$\label{ln:dllite-ex-sig-supp1}}{
      \If{$\tup{\calT, \{\beta\}} \models \alpha$ {\bfseries and} $\calR \cup \{\beta\}$ is \calT-consistent}{
        {\bf accept}\label{ln:dllite-ex-sig-supp2}
      }
    }
    {\bf reject}
  \end{algorithm}
  For correctness, note that lines~\ref{ln:dllite-ex-sig-repair1}-\ref{ln:dllite-ex-sig-repair2} accept iff \calR is not a repair of \calK.
  Hence, the for-loop in lines~\ref{ln:dllite-ex-sig-supp1}-\ref{ln:dllite-ex-sig-supp2} checks that for all repairs \calR of \calK, there is some \calT-support $\{\beta\}$ of $\alpha$ from $\calA \cup \calB$ such that $\calR \cup \{\beta\}$ is \calT-consistent.
  By \cref{thm:dllite-ar-hyps}, this is equivalent to the existence of any \ar-hypothesis $\Hyp \subseteq \calB$ for $\tup{\calK, \alpha}$.
  Finally, this is equivalent to existence of any \SignatureRestricted \ar-hypothesis for $\tup{\calK, \alpha, \Sigma}$, since \calB contains all assertions over $\Sigma$.

  The algorithm only guesses universally and runs in polynomial time:
  The sets \calB can be computed in polynomial time.
  The for-loops iterate over poly-size sets, and consistency and inconsistency as well as entailment of concept assertions under classical semantics can be checked in $\NL \subseteq \Ptime$.

  We now show \coNP-hardness for the \SignatureRestricted setting by adapting the reduction from unsatisfiability to \ar-entailment presented in the proof of \cref{thm:dllite-verif-gen-card}, with correctness following from \cref{fac:redu-meghyn}.
  For a given formula $\varphi$ in CNF, construct the TBox \calT and the ABoxes \calA and \Hyp as in that construction.
  Now, define
  \begin{align*}
    \calT' &\coloneqq \calT \cup \{\exists U \sqsubseteq \neg \exists U^-\} \text{ and} \\
    \calA' &\coloneqq \calA \setminus \{\, U(a,c_j) \mid j \leq k \,\}
  \end{align*}
  and let $\calK' \coloneqq \tup{\calT', \calA'}$ and $\Sigma \coloneqq \{\, U, a, c_j \mid j \leq k \,\}$.
  Due to the new axiom $\exists U \sqsubseteq \neg \exists U^-$, the assertion $U(a,a)$ is not $\calT'$-consistent, so it cannot help entailment under brave semantics.
  The only remaining assertions over $\Sigma$ that can help entailment of $A(a)$ via the axiom $\exists U \sqsubseteq A$ are assertions in \Hyp.
  Further, since all assertions in \Hyp are $\calT'$-support of $A(a)$, if any subset of them is an \ar-hypothesis for $\tup{\calK', A(a)}$, then so is \Hyp.
  Hence, existence of any \ar-hypothesis for $\tup{\calK', A(a), \Sigma}$ is equivalent to \Hyp being such a hypothesis.
  By \cref{fac:redu-meghyn}, this implies correctness of the reduction.

  We now turn to the case of existence of non-trivial \ar-hypotheses.
  For membership in \coNP, we call \cref{alg:dllite-ex-sig} with $\Sigma$ set to the signature of \calK and $\alpha$ and modify the algorithm slightly:
  Define \calB as the set of all assertions over the signature $\Sigma$, except for the assertion $\alpha$ itself.
  As the latter is not allowed in any non-trivial hypothesis, correctness follows by \cref{thm:dllite-ar-hyps} using the same arguments as for the \SignatureRestricted setting above.

  For \coNP-hardness for existence of non-trivial \ar-hypotheses, we reduce from checking whether a given $\forall$-QBF $\forall X \varphi(X)$ is true, where $\varphi = \{C_1, \dots, C_m\}$ is in DNF, i.e., the $C_i$ are conjunction terms
  Let $X = x_1, x_2, \dots, x_n$ be the variables in $\varphi$.
  We construct a \DLLite KB $\calK = \tup{\calT, \calA}$ such that there is a non-trivial \ar-hypothesis for $\tup{\calK, A(a)}$ iff $\forall X \varphi(X)$ is true.
  Define
  \begin{align*}
    \calT \coloneqq {} &\{\, C_j \sqsubseteq A \mid 1 \leq j \leq m \,\} \cup {} \\
    & \{\, T_{x_i} \sqsubseteq \neg C_j \mid \neg x_i \in C_j \,\} \cup {} \\
    & \{\, F_{x_i} \sqsubseteq \neg C_j \mid x_i \in C_j \,\} \cup {} \\
    & \{\, T_{x_i} \sqsubseteq \neg F_{x_i} \mid 1 \leq i \leq n \,\} \text{ and} \\
    \calA \coloneqq {} & \{\, T_{x_i}(a), F_{x_i}(a) \mid 1 \leq i \leq n \,\}.
  \end{align*}
  We now show that $\forall X \varphi(X)$ is true iff there is a non-trivial \ar-hypothesis for $\tup{\calK, A(a)}$.

  ($\Rightarrow$) Assume that $\forall X \varphi(X)$ is true, that is, for each assignment $s$ over $X$ there is some $j$ such that $s \models C_j$.
  Let $\Hyp \coloneqq \{\, C_j(a) \mid 1 \leq j \leq m \,\}$.
  Consider any repair \calR of $\tup{\calT, \calA \cup \Hyp}$.
  First note that by the axioms of the form $T_{x_i} \sqsubseteq \neg F_{x_i}$, the assertions of the form $T_{x_i}(a)$ and $F_{x_i}(a)$ present in \calR correspond to an assignment $s$ over $X$.
  By assumption, there is some $j$ such that $s \models C_j$.
  Hence, $\calR \cup \{C_j(a)\}$ is \calT-consistent by construction of \calT, so $C_j(a) \in \calR$ by maximality of repairs.
  Consequently, $\tup{\calT, \calR} \models A(a)$ by the axiom $C_j \sqsubseteq A$.

  ($\Leftarrow$) Assume that there is a non-trivial \ar-hypothesis \Hyp for $\tup{\calK, A(a)}$.
  Consider any assignment $s$ over $X$.
  Due to the axioms of the form $T_{x_i} \sqsubseteq \neg F_{x_i}$, there is a repair \calR of $\tup{\calT, \calA \cup \Hyp}$ that contains exactly the assertions of the form $T_{x_i}(a)$ and $F_{x_i}(a)$ corresponding to assignment $s$.
  As $\tup{\calT, \calR} \models A(a)$ by assumption, there is some $j$ s.t.\ $C_j(a) \in \calR$.
  But then $s \models C_j$ by construction of \calT.
  As this applies to all assignments $s$ over $X$, this implies that $\forall X \varphi(X)$ is true.
\end{proof}

\dlliteVerifArSubset*
\begin{proof}
  For membership, let $\tup{\calK, \alpha}$ be an \ar-abduction problem, where $\calK = \tup{\calT, \calA}$ is a \DLLite KB and let \Hyp be an ABox.
  To verify that \Hyp is a $\subseteq$-minimal \ar-hypothesis for $\tup{\calK, \alpha}$, we need to check that
  \begin{enumerate}
    \item $\tup{\calT, \calA \cup \Hyp} \models_\ar \alpha$ and
    \item \label{it:min1} there is no \ar-hypothesis $\Hyp' \subsetneq \Hyp$ for $\tup{\calK, \alpha}$.
  \end{enumerate}
  The former can be done in \coNP, as it is just \ar-entailment.

  To show \DP-membership, we now argue that \ref{it:min1} is in \NP by showing that we can answer the complement in \coNP.
  Let $\Hyp = \{\beta_1, \dots, \beta_n\}$ and consider the subsets $\calB_i = \Hyp \setminus \{\beta_i\}$ for all $i$.
  Obviously, there is an \ar-hypothesis $\Hyp' \subsetneq \Hyp$ for $\tup{\calK, \alpha}$ iff there is an \ar-hypothesis $\Hyp' \subseteq \calB_i$ for $\tup{\calK, \alpha}$ for some $i$.
  By \cref{thm:dllite-ar-hyps}, setting \calB to $\calB_i$, there is an \ar-hypothesis $\Hyp' \subseteq \calB_i$ for $\tup{\calK, \alpha}$ iff for each repair \calR of \calK, there is a \calT-support $\{\beta\} \subseteq \calA \cup \calB_i$ of $\alpha$ s.t.\ $\calR \cup \{\beta\}$ is \calT-consistent.
  Hence, for each $i$ we can check whether there is an \ar-hypothesis $\Hyp' \subseteq \calB_i$ for $\tup{\calK, \alpha}$ using a slight modification of \cref{alg:dllite-ex-sig}.
  Namely, instead of relying on the signature $\Sigma$, simply redefine \calB in that algorithm to $\calB_i$.
  Consecutively checking the above for all $i$ yields a \coNP-algorithm checking whether there is any \ar-hypothesis $\Hyp' \subsetneq \Hyp$ for $\tup{\calK, \alpha}$. 

	For $\DP$-hardness, we again use the reduction from unsatisfiability also used in the proof of \cref{thm:dllite-verif-gen-card}, but first introduce some terminology.
	Given a formula $\varphi$ in CNF, a collection $\psi\subseteq \varphi$ of clauses is a  MUS of $\varphi$ if $\psi$ is unsatisfiable but $\psi'$ is satisfiable for every $\psi' \subset\psi$.
  It is known that the problem to decide if a set of clauses is a MUS is $\DP$-hard~\cite{Liberatore05}.
  Now, constructing $\calK'$ as in the reduction in the proof of \cref{thm:dllite-verif-gen-card}, it can be observed that the $\subseteq$-minimal \ar-hypotheses $\calH$ for $\tup{\calK', \alpha}$ correspond precisely to MUSes $\psi_\Hyp$ for $\varphi$ by taking $c_j\in \psi_\calH\iff U(a,c_j)\in \calH$.
  Hence, given a CNF $\varphi$ and set of clauses $\psi$, we obtain a reduction to verification of $\subseteq$-minimal \ar-hypotheses by mapping it to the \ar-abduction problem $\tup{\calK', \alpha}$ and the candidate hypothesis $\Hyp_\psi= \{U(a,c_j)\mid c_j\in\psi\}$.	
\end{proof}

\section{Full Proofs for \cref{sec:elbot}}

\elbotVerifGeneral*
\begin{proof}
	We reduce our problem to/from $\calS$-entailment for $\EL_\bot$ KBs, which is \NP-complete for \brave semantics~\cite{BienvenuB16} and $\co\NP$-complete for \ar~\cite{Rosati11}.

	Let $\tup{\calK, \alpha}$ be an \calS-abduction problem and \Hyp an ABox.
	Observe that the question whether \Hyp is an $\calS$-hypothesis for $\tup{\calK, \alpha}$ only requires a check for $\calS$-entailment.
	Thus, the complexity of $\calS$-entailment yields an upper bound.
	Furthermore, to check $\leq$-minimality it is sufficient to check whether $|\Hyp| = 1$, since $\{\alpha\}$ is an \calS-hypothesis for $\tup{\calK, \alpha}$ iff there is any such hypothesis by Lemmata~\ref{prop:brave-triv} and~\cref{lem:ar-triv}.
  Thus, checking $\leq$-minimality does not increase the complexity.
	
	For hardness, we reduce $\calS$-entailment to $\calS$-verification.
	Given a KB $\calK = \tup{\calT,\calA}$ and a concept assertion $C(a)$, we let $\calT' = \calT \cup \{C\sqcap B \subsum A\}$ for fresh concepts $A$ and $B$, and consider $\calK'=\tup{\calT',\calA}$.
	Moreover, we let $\alpha \dfn A(a)$ and $\Hyp\dfn\{B(a)\}$.
	It is easy to see that $\calK' \not\models_\calS \alpha$ and hence $\tup{\calK, \alpha}$ is a valid abduction problem.
	Now, we observe that
	\[\calK \models_\calS C(a)\iff \calK'\models_\calS C(a) \iff \tup{\calT',\calA\cup\calH}\models_\calS \alpha.\]
	Consequently, $\calH$ is a ($\leq$-minimal) $\calS$-hypothesis for $\tup{\calK',\alpha}$ if and only if $\calK \models_\calS C(a)$.
	Thus we obtain corresponding hardness in each case due to the complexity of $\calS$-entailment.
\end{proof}

\elbotExistGeneral*
\begin{proof}
  For $\calS = \brave$, consider a \brave-abduction problem $\tup{\calK, \alpha}$ with $\calK = \tup{\calT, \calA}$ and $\alpha = A(a)$.
  By \cref{prop:brave-triv}, existence of any \brave-hypothesis for $\tup{\calK, \alpha}$ is equivalent to $\{\alpha\}$ being such a hypothesis.
  The latter is equivalent to $A$ being satisfiable w.r.t.\ \calT, so the problem is equivalent to satisfiabilty of named concepts w.r.t.\ \ELbot-TBoxes, which is known to be \Ptime-complete.
  
  We now turn towards \ar semantics.
  Let $\tup{\calK, \alpha}$ be an \ar-abduction problem.
	By \cref{lem:ar-triv}, checking whether there is an \ar-hypothesis for $\tup{\calK, \alpha}$ is equivalent to checking whether $\{\alpha\}$ is conflict-confining for \calK.
  \cref{lem:cc-compl} now yields \coNP-completeness, as the hardness proof already applies to singleton ABoxes.
\end{proof}

Interestingly the complexity of abduction for consistent KBs in the signature-restricted setting is not well-mapped.
Precisely, the work by \cite{Koopmann21a} does not yield results in our setting as the signature is restricted differently (the mentioned work does not restrict individuals in the signature).
The following proposition closes this gap by characterizing the complexity of signature-restricted abduction in $\EL_\bot$ under classical semantics.

\begin{proposition}\label{thm:consistent-signature-elbot}
  For \ELbot, the existence problem for \SignatureRestricted hypotheses under classical semantics is \NP-complete.
\end{proposition}
\begin{proof}
	\emph{Membership:} Guess a set \Hyp of assertions over the signature $\Sigma$, i.e.\ $\Hyp \subseteq \{A(x), r(x,y) \mid A, r, x, y \in \Sigma\}$.
  Then verify that $\calA \cup \Hyp$ is \calT-consistent and $\tup{\calT, \calA \cup \Hyp} \models \alpha$. 
	Both checks can be performed in polynomial time.
	
	\emph{Hardness:} We reduce from propositional satisfiability.
	To this aim, let $\varphi = \{c_1,\dots,c_p\}$ be a CNF formula over propositions $X =\{x_1,\dots,x_n\}$, where each $c_i$ is a clause.
	A literals $\ell$ is a variable $x$ or a negated variable $\neg x$.
	For a literal $\ell$, we denote by $\bar\ell$ its ``opposite'' literal.
	For a set $Z$ of variables, $\Lit(Z)$ denotes the collection of literals over $Z$.
	We construct the KB $\calK = \tup{\calT, \calA}$ with concept names $N=\{A_x, A_{\bar x} \mid x\in X\} \cup \{A_c \mid c \in\varphi\} \cup \{A_\varphi\}$, where
	\begin{align*}
		\calT\dfn {} & \{\,A_x \sqcap A_{\bar x} \subsum \bot \mid x\in X\, \} \cup {} \\ & \{\,A_\ell \subsum A_c \mid \ell \in c, c\in\varphi\,\} \cup {} \\
		& \left\{\,\bigsqcap_{c\in \varphi}A_c\subsum A_\varphi\,\right\}, \\
		\calA \dfn {} & \emptyset.
	\end{align*} 
	Furthermore, let $\alpha \coloneqq A_\varphi(m)$ and
  \[\Sigma =\{A_x, A_{\bar x}, \mid x\in X\} \cup \{m\}\]
  for an individual name $m$.
	Now, $\tup{\calK, \alpha, \Sigma}$ is the desired abduction problem.
	Clearly, $\calK \not\models \alpha$, since no assertion in $\calK$ involves $m$.
	\begin{claim}
		$\varphi$ is satisfiable iff $\alpha$ has a hypothesis over $\Sigma$ in \calK.
	\end{claim}
	\begin{claimproof}
		($\Rightarrow$) Let $s \subseteq \Lit(X)$ be a satisfying assignment for $\varphi$ seen as a set of literals.
		Then, for each clause $c\in\varphi$ there is some $\ell\in s\cap c$.
		We define $\Hyp = \{A_\ell(m) \mid \ell\in s\}$. 
		Since $s$ is an assignment, the set $\Hyp$ is consistent with $\calT$:
		No inconsistency is triggered due to axioms $A_x\sqcap A_{\bar x}\subsum \bot$, as $\Hyp$ only contains exactly one assertion for each variable.
		Now, we prove that $\calK_{\Hyp} \models \alpha$, where $\calK_{\Hyp} \coloneqq \tup{\calT, \Hyp}$.
		Since each clause is satisfied, we have $\calK_{\Hyp} \models A_c(m)$ for each $c\in\varphi$, which in turn implies that $\calK_{\Hyp} \models A_\varphi(m)$ due to the last TBox axiom.
		
		($\Leftarrow$)
		Let $\Hyp$ be a hypothesis for $\alpha$ in $\calK$.
		Observe that there are sets $X_1, X_2 \subseteq X$ of variables such that $X_1 \cap X_2 = \emptyset$ and $\Hyp$ takes the following form:
		\[\Hyp = \{A_x(m) \mid x\in X_1\}\cup \{A_{\bar x}(m)\mid x\in X_2\}.\]
		This holds, since some disjointness axiom is violated otherwise and one can remove corresponding assertions from $\calH$ without breaking the entailment of $\alpha$.
		We define $s_{\Hyp} = \{\ell \mid A_\ell(m) \in \Hyp \}$.
		Then, $s_{\Hyp}$ is a potentially partial assignment over $X$, as for any variable it may neither contain the positive nor the negative literal.
		Now, we prove that $s_{\Hyp} \models \varphi$.
		However, this is easy to see, since $\tup{\calT, \Hyp} \models A_c(m)$ for every clause $c \in \varphi$. 
		Consequently, for each $c \in \varphi$, there is some $\ell \in c$, such that $A_\ell(m) \in \Hyp$, which in turn implies that $\ell \in s_{\Hyp}$.
		Obviously, $s_\Hyp$ can also be extended to a full assignment that still satisfies $\varphi$.
	\end{claimproof}
\end{proof}

\elbotExistSig*
\begin{proof}
  For (1): An \NP-algorithm for the problem can guess a hypothesis $\Hyp$ over the signature $\Sigma$ and, at the same time, guess a repair $\calR$ of the ABox.
  Then, verify that $\tup{\calT, \calR \cup \Hyp} \not\models \bot$ and $\tup{\calT, \calR \cup \Hyp} \models \alpha$ in polynomial time.
	The \NP-hardness can be shown by slightly adapting the reduction in Proposition~\ref{thm:consistent-signature-elbot}, adding a dummy inconsistency over fresh concepts not in $\Sigma$.	
	For (2): The following algorithm shows \SigmaP-membership:
  Guess a set $\Hyp$ such that for all repairs $\calR$ of $\tup{\calT, \calA \cup \Hyp}$, we have $\tup{\calT, \calR} \models \alpha$.
	This requires $\NP$-time to guess the set $\Hyp$ and an $\NP$-oracle to guess a repair $\calR$ as a counter example to the entailment, thus resulting in $\SigmaP$-membership.
	For hardness, we reduce from the standard $\SigmaP$-complete problem of validity for $\exists\forall$-QBFs:
	Instances of this problem are of the form $\Phi \coloneqq \exists Y \forall Z \varphi'$, where $\varphi'$ is a Boolean formula over variables $X = Y \cup Z$.
	Without loss of generality, we can assume that $\varphi' = \neg \varphi$ for some Boolean formula $\varphi$ in CNF.
	The problem asks whether $\Phi$ is valid (or true).
	We construct the following KB $\calK = \tup{\calT, \calA}$, using concept names $N = \{A_x, A_{\bar x} \mid x \in X\} \cup \{V_y \mid y \in Y\} \cup \{A_c \mid c \in \varphi\} \cup \{A_{\varphi}, A_{\bar\varphi}, C\}$.
  We define
  \begin{align*}
 	\calT \dfn {} 	& \{\,A_x \sqcap A_{\bar x} \subsum \bot \mid x\in X\,\} \cup {} \\ 
		& \{\, A_\ell \subsum A_c \mid \ell \in c, c\in\varphi\,\} \cup {} \\ 
		& \left\{ \bigsqcap_{c\in \varphi}A_c\subsum A_\varphi, \quad A_\varphi \sqcap A_{\bar\varphi}\subsum \bot \right\} \cup {} \\ 
		& \{\, A_y \subsum V_y, A_{\bar y} \subsum V_y \mid y\in Y \,\} \cup {} \\ 
		& \left\{ \bigsqcap_{y\in Y} V_y \sqcap A_{\bar \varphi} \subsum C \right\} \text{ and} \\ 
		\calA \dfn {} &\{\,A_z(m), A_{\bar z}(m)\mid z\in Z\,\}  \cup {} \{\,A_{\bar\varphi}(m)\,\}
  \end{align*}
	Finally, let $\Sigma \coloneqq \{m\} \cup \{A_y, A_{\bar y} \mid y \in Y\}$ and $\alpha \coloneqq C(m)$.
	Intuitively, the axioms in $\calT$ ensure: a valid assignment over $X$, satisfaction of each clause for its corresponding literals, and the satisfaction of the formula $\varphi$.
	The final two sets of axioms are needed to encode that hypotheses over $\Sigma$ are assignments over variables from $Y$.
	The ABox $\calA$ intuitively corresponds to encoding all the assignments over variables from $Z$.
	Now $\tup{\calK, \alpha, \Sigma}$ is the desired abduction problem.

  We first observe that $\tup{\calK, \alpha}$ is a valid \ar-abduction problem:
  Obviously, $\calK \models \bot$ when $Z$ is non-empty, due to both $A_z(m)$ and $A_{\bar z}(m)$ being present in the ABox for every $z \in Z$.
  Also, $\calK\not\models_{\ar} \alpha$, as $\calA$ does neither involve the concept name $C$ nor any of the concept names $A_y$, $A_{\bar y}$, or $V_y$ for $y \in Y$.
  The following claim states correctness of the reduction.
	\begin{claim}
		$\Phi$ is true if and only if $\alpha$ has an AR-hypothesis over the signature $\Sigma$ in $\calK$.
	\end{claim}
	
	\begin{claimproof}
	($\Rightarrow$)
	Suppose $\Phi$ is true.
	Then there is an assignment $s \subseteq \Lit(Y)$ such that for all assignments $t \subseteq \Lit(Z)$, $\neg\varphi[s,t]$ is true.
\begin{shortonly}
  Here, $\Lit(\cdot)$ denotes the set of literals over a given set of variables.
\end{shortonly}
	We construct an AR-hypothesis for $\alpha$ from $s$.
	Let $\Hyp_s = \{A_p(m) \mid p\in s\}$.
	Obviously, $\Hyp_s$ is an ABox over $\Sigma$.
	Also, it does not violate any axiom of the form $A_y \sqcap A_{\bar y} \subsum \bot$, since $s$ is an assignment.
	
	We prove that $\tup{\calT, \calA \cup \Hyp_s} \models_{\ar} \alpha$.
	Consider any repair $\calR$ of $\tup{\calT, \calA \cup \Hyp_s}$.
	As $\tup{\calT, \calR} \not\models \bot$, \calR does not violate any axiom of the form $A_x \sqcap A_{\bar x} \sqsubseteq \bot$.
	Hence, $\calR \cap \{A_x(m), A_{\bar x}(m) \mid x \in X\}$ corresponds to (potentially partial) assignments $s_\calR \subseteq s$ and $t_\calR$ over $Y$ and $Z$, respectively.
	We first prove that $\tup{\calT, \calR} \not\models A_\varphi(m)$.
	Suppose to the contrary, that $\tup{\calT, \calR} \models A_\varphi(m)$.
	As \calR is consistent with \calT, this only happens by triggering the axiom $\sqcap_{c \in \varphi} A_c \sqsubseteq A_\varphi$, and in turn an axiom of the form $A_\ell \sqsubseteq A_c$ for each clause $c \in \varphi$.
	But this means that $s_\calR \cup t_\calR$, and hence also $s \cup t_\calR$, is a satisfying assignment for $\varphi$, which is a contradiction to $\neg\varphi[s,t]$ being true for all assignments $t$ over $Z$.
  Indeed, as this argument covers the case $s_\calR = s$, subset-maximality of repairs further yields that $\Hyp_s \subseteq \calR$.
	Moreover, subset-maximality together with the fact that $\tup{\calT, \calR} \not\models A_\varphi(m)$ yields that $A_{\bar \varphi}(m) \in \calR$.
	Consequently, $\tup{\calT, \calR} \models C(m)$.
	
	($\Leftarrow$)
	Suppose $\Phi$ is false.
	Then, for each assignment $s \subseteq \Lit(Y)$, there is an assignment $t \subseteq \Lit(Z)$ such that $\neg\varphi[s,t]$ is false or, equivalently, $\varphi[s,t]$ is true.
	The latter can be stated as: each clause $c\in \varphi$ contains some literal $\ell \in c$ with $\ell\in s\cup t$.
	
	We now prove that $\alpha$ does not have any AR-hypothesis over $\Sigma$ in $\calK$.
	For	contradiction, assume that $\Hyp \subseteq \{A_p(m) \mid p \in \Lit(Y)\}$ is such a hypothesis and consider any repair \calR of $\tup{\calT, \calA  \cup \Hyp}$.
	As $\calR$ does not violate any axiom of the form $A_x \sqcap A_{\bar x} \sqsubseteq \bot$, the subset $\calR_Y \coloneqq \calR \cap \{A_y(m), A_{\bar y}(m) \mid y \in Y\}$ corresponds to a potentially partial assignment $s_\calR$ over $Y$.
	On the other hand, as $\tup{\calT, \calR} \models C(m)$, we also have $\tup{\calT, \calR} \models \sqcap_{y \in Y} V_y(m)$.
	Therefore, $\calR$ contains at least one assertion from $\{A_y(m), A_{\bar y}(m)\}$ for each $y\in Y$, i.e.\ that $s_\calR$ is a full assignment over $Y$.
	By our assumption, there is an assignment $t$ over $Z$ s.t.\ $\varphi[s_\calR, t]$ is true.
	
	Let $\calR_t \coloneqq \calR_Y \cup \{A_\ell(m) \mid \ell \in t\}$.
	Obviously, $\calR_t$ does not violate any of the disjointness axioms in \calT, as it does not contain $A_{\bar \varphi}(m)$ and $s_\calR \cup t$ is an assignment over $X$.
	This further means that $\tup{\calT, \calR_t} \not\models C(m)$.
	Furthermore, $\calR_t$ is subset-maximal: As both $s_\calR$ and $t$ are full assignments, we cannot add any assertion of the form $A_x(m)$ or $A_{\bar x}(m)$ for $x \in X$ without violating one of the disjointness axioms.
	Also, as $\varphi[s_\calR, t]$ is true, we have $\tup{\calT, \calR_t} \models A_\varphi(m)$.
	Hence, we also cannot add $A_{\bar \varphi}(m)$ without violating the corresponding disjointness axiom.
	This shows that $\calR_t$ is a repair of $\tup{\calT, \calA \cup \Hyp}$ that does not entail $\alpha$, contradicting our assumption.
	
	This proves the correctness of the claim.%
	\end{claimproof}
  We conclude by observing that the reduction can be achieved in polynomial time. 
\end{proof}

Once again, the complexity of existence for nontrivial hypothesis is not addressed in the setting of consistent $\ELbot$ KBs.
Moreover, the reduction from the proof of Proposition~\ref{thm:consistent-signature-elbot} does not apply here.
Nevertheless, the following proposition closes this gap by characterizing the complexity of existence for non-trivial hypothesis in $\EL_\bot$ under classical semantics.

\begin{proposition}\label{thm:consistent-nt-elbot}
	For \ELbot, the existence problem for non-trivial hypotheses under classical semantics is \NP-complete.
\end{proposition}
\begin{proof}
	Membership follows similarly as in the proof of Proposition~\ref{thm:consistent-signature-elbot}. We guess a set \Hyp of assertions over the signature of $\calK$.
	Then, verify that $\calA \cup \Hyp$ is \calT-consistent and $\tup{\calT, \calA \cup \Hyp} \models \alpha$. 
	Both checks can be performed in polynomial time.
	
	\emph{Hardness:} We reduce from propositional satisfiability.
	To this aim, let $\varphi = \{c_1,\dots,c_p\}$ be a CNF formula over propositions $X =\{x_1,\dots,x_n\}$, where each $c_i$ is a clause.
	We construct the KB $\calK = \tup{\calT, \calA}$ with concept names $N=\{A_x, A_{\bar x} \mid x\in X\} \cup \{r_c \mid c \in\varphi\} \cup \{B, A_\varphi\}$, where
	\begin{align*}
		\calT\dfn {} & \{\,A_x \sqcap A_{\bar x} \subsum \bot \mid x\in X\, \} \cup {} \\ & \{\,A_\ell \subsum \exists r_c.B \mid \ell \in c, c\in\varphi\,\} \cup {} \\
		& \left\{\,\bigsqcap_{c\in \varphi}\exists r_c.B\subsum A_\varphi\,\right\} \cup {} \\
		& \{\, A_\varphi\sqcap B\subsum \bot\,\} \text{ and} \\
		\calA \dfn {} & \emptyset.
	\end{align*} 
	Furthermore, let $\alpha \coloneqq A_\varphi(m)$ for an individual name $m$.
	Now, $\tup{\calK, \alpha}$ is our desired abduction problem.
	Since the ABox is empty, we clearly have $\calK \not\models \alpha$.
	\begin{claim}
		$\varphi$ is satisfiable iff $\alpha$ has a non-trivial hypothesis in \calK.
	\end{claim}
	\begin{claimproof}
		($\Rightarrow$) Let $s \subseteq \Lit(X)$ be a satisfying assignment for $\varphi$ seen as a set of literals.
		Then, for each clause $c\in\varphi$ there is some $\ell\in s\cap c$.
		We define the ABox
    \[\Hyp \coloneqq \{A_\ell(m) \mid \ell\in s\}.\]
		Since $s$ is an assignment, the set $\Hyp$ is consistent with $\calT$:
		No inconsistency is triggered due to axioms $A_x\sqcap A_{\bar x}\subsum \bot$, as $\Hyp$ only contains exactly one assertion for each variable.
		Now, we prove that $\calK_{\Hyp} \models \alpha$, where $\calK_{\Hyp} \coloneqq \tup{\calT, \Hyp}$.
		Since each clause is satisfied, we have $\calK_{\Hyp} \models \exists r_c.B(m)$ for each $c\in\varphi$ due to the corresponding TBox axiom ``$A_\ell\subsum \exists r_c.B$'', which in turn implies that $\calK_{\Hyp} \models A_\varphi(m)$ due to the TBox axiom $\bigsqcap_{c\in\varphi} \exists r_c.B\subsum A_\varphi$.

		($\Leftarrow$)
		Let $\Hyp$ be a non-trivial hypothesis for $\alpha$ in $\calK$ and thus
		we have $A_\varphi(m) \not\in \calH$.
		Hence, the entailment of $A_\varphi(m)$ must be obtained via the axiom $\left\{ \bigsqcap_{c\in\varphi} \exists r_c.B \subsum A_\varphi \right\}$.
		Further, we have $B(m) \not\in \calH$, as this assertion conflicts with $A_\varphi(m)$.
		As $\Hyp$ only contains assertions over the individual $m$, this means that the existence of $r_c$-successors of $m$ in $B$ must be achieved by anonymous individuals.
		But these can only be entailed using axioms of the form $A_\ell \subsum \exists r_c.B$, which means that for each $c \in \varphi$, there is some $\ell \in c$ s.t.\ $A_\ell(m) \in \calR$.
		Due to the axioms of the form $A_x \sqcap A_{\bar x} \subsum \bot$, there is at most one of the assertions $A_x(m)$ and $A_{\bar x}(m)$ in \calH for each $x$, so \Hyp corresponds to a (partial) assignment $s_\calH$ over $X$.
		From the above, we have that for each $c \in \varphi$, there is some $\ell \in c \cap s_\calH$, i.e., $s_\calH$ satisfies $\varphi$.
	\end{claimproof}
\end{proof}

\elbotExistNontriv*
\begin{proof}
  We begin with the case of \brave semantics.
  We first show \NP-membership:
  Existence of a non-trivial \brave-hypothesis for some \brave-abduction problem $\tup{\tup{\calT, \calA}, \alpha}$ can be checked by guessing a candidate hypothesis \Hyp over the signature of $\tup{\calK, \alpha}$, as well as a candidate repair \calR of $\tup{\calT, \calA \cup \Hyp}$ and verifying that \Hyp is non-trivial, \calR is a repair, and $\tup{\calT, \calR} \models \alpha$ in polynomial time.

  The \NP-hardness can be shown by slightly adapting the reduction in Proposition~\ref{thm:consistent-nt-elbot}.
  We add a dummy inconsistency and follow the similar idea as in the proof of Theorem~\ref{thm:elbot-exist-sig}. 
  	 
  We now turn to \ar semantics.
  The following algorithm shows \SigmaP-membership:
  Given an \ar-abduction problem $\tup{\calK, \alpha}$ over some signature $\Sigma$, guess a hypothesis \Hyp over $\Sigma$, and verify that $\alpha \not\in \Hyp$ and $\tup{\calT, \calA \cup \Hyp} \models_\ar \alpha$.
  The latter is \ar-entailment and can hence be handled by querying a \coNP oracle.
  
  For hardness, we reduce from validity problem for $\exists\forall$-QBFs.
  Let $\Psi = \exists Y\forall Z \; \psi$ be a $\exists\forall$-QBF, where $\psi$ is in DNF, and let $X = Y \cup Z$.
  We construct a TBox \calT using concept names $N = \{A_x, A_{\bar x} \mid x \in X\} \cup \{C, A_\psi\}$.
  Intuitively, \calT expresses that for any term $t \in \psi$, the conjunction of concepts corresponding to the literals in $t$ entails the concept $A_\psi$, which represents satisfaction of $\psi$.
  To this end, we define   $\calK\dfn \tup{\calT,\calA}$, where 
  \begin{align*}
  	\calT \dfn {} &\{\,A_x \sqcap A_{\bar x} \subsum \bot \mid x\in X\, \} \cup {} \\ 
  	& \left\{\,C\sqcap \bigsqcap_{\ell\in t}A_\ell \subsum A_\psi \mid t\in\psi\,\right\}, \\ 
  	\calA \coloneqq {} & \{\,A_z(m), A_{\bar z}(m)\mid z\in Z\}, 
  \end{align*}
	for an individual name $m$.
  Here, the first set of axioms ensures that repairs encode (potentially partial) assignments over $X$ and the second encodes that $\psi$ is satisfied iff at least one its terms is satisfied.
  Now, $\tup{\calK, A_\psi(m)}$ is the desired abduction problem.
  It is a valid abduction instance, as \calK is obviously inconsistent, and $\calK \not\models_{\ar} A_\psi(m)$, as \calK does not contain the assertion $C(m)$.
  Observe that $\calH_{\text{triv}} \dfn \{A_\psi(m)\}$ is the trivial \ar-hypothesis for $\tup{\calK, \alpha}$.
  However, we are considering existence of non-trivial hypotheses here, i.e., one that does not contain $A_\psi(m)$.
  For correctness, we prove the following claim.
  \begin{claim}\label{claim:nt-verification-el}
  	$\Psi$ is true iff $\tup{\calK, \alpha}$ admits a non-trivial \ar-hypothesis..
  \end{claim}
  \begin{claimproof}
  	($\Rightarrow$) Let $\theta_Y$ be an assignment over $Y$, represented as a set of literals, witnessing that $\Psi$ is true, i.e., f.a.\ assignments $\theta_Z$ over $Z$, $\theta_Y \cup \theta_Z \models t$ for some $t \in \psi$.
  	Define $\Hyp \dfn \{C(m)\} \cup \{\, A_\ell(m) \mid \ell \in \theta_Y \,\}$.
    We show that $\calH$ is an \ar-hypothesis for $\alpha$ in $\calK$.
    Consider any repair \calR of $\calK' \coloneqq \tup{\calT, \calA \cup \Hyp}$.
    Since $\theta_Y$ is an assignment, \Hyp is \calT-consistent.
    Even more, no subset of \Hyp is contained in any conflict of $\calK'$.
    Hence, we have $\Hyp \subseteq \calR$.
    Now let $\theta_Z \dfn \{\, \ell \in Z \cup \bar Z \mid A_\ell(m) \in \calR \,\}$.
    We argue that $\theta_Z$ is a (full) assignment over $Z$.
    As \calR is \calT-consistent, at most one of the assertions $A_x(m)$ and $A_{\bar x}(m)$ is in \calR for each $x \in X$ due to the first form of axioms in \calT.
    On the other hand, since \calR is subset-maximal among the \calT-conistent subsets of $\calA \cup \Hyp$, it contains at least one of those to assertions for each $x \in X$.    Hence, there is some term $t \in \psi$ s.t.\ $t \subseteq \theta_Y \cup \theta_Z$ by our assumption on $\theta_Y$.
    Since $\Hyp \subseteq \calR$ and by construction of \calT, this implies that $\tup{\calT, \calR} \models A_\psi(m)$.
  	This proves that $\calK' \models_\ar A_\psi(m)$, since the argument applies to all repairs \calR of $\calK'$.

  	($\Leftarrow$) Let \Hyp be a non-trivial \ar-hypothesis for $\tup{\calK, \alpha}$.
    We first argue that w.l.o.g., \Hyp only contains assertions of the form $A_\ell(m)$ and $C(m)$:
    It can only contain assertions over individual $m$, as it may not use fresh individuals.
    W.l.o.g.\ it does not contain assertions using concept and role names not occuring in $\tup{\calK, \alpha}$ or already present in \calA, as including these assertions in \Hyp does not contribute to the entailment of $A_\psi(m)$.
    Finally, $A_\psi(m) \not\in \Hyp$, as \Hyp is non-trivial.

    It is now easy to see that $C(m) \in \Hyp$, as otherwise $A_\psi(m)$ could not be entailed.
    Further, we can w.l.o.g. assume that \Hyp contains at most one of the assertions $A_y(m)$ and $A_{\bar y}(m)$ for each $y \in Y$.
    This holds since we consider \ar semantics, so $\alpha$ is entailed for all repairs of $\calK' \dfn \tup{\calT, \calA \cup \Hyp}$.
  	More precisely, assume that $\{A_y(m),A_{\bar y}(m)\} \subseteq \Hyp$ for some $y\in Y$ and let $\Hyp \coloneqq \Hyp \setminus \{A_y(m)\}$.
    Then any repair \calR of $\tup{\calT, \calA \cup \Hyp'}$ is also a repair of $\calK'$, and hence we have $\tup{\calT, \calR} \models \alpha$.
  	Assume further that $\tup{\calT,\calA\cup \calH_d}\not\models_\ar \alpha$ for $\calH_d\dfn \calH\setminus\{A_d(m)\}$ and $d\in\{y,\bar y\}$, i.e., both assertions are necessary to entail $\alpha$.
    Consequently, $\Hyp'$ is also a non-trivial \ar-hypothesis for $\tup{\calK, \alpha}$.

  	Now, define the (potentially partial) assignment $\theta_Y \dfn \{\, \ell \mid A_\ell(m) \in \calH \,\}$ over $Y$.
  	To finish the proof, we show that for each assignment $\theta_Z$ over $Z$, we have $\theta_Y \cup \theta_Z \models \psi$.
    Let $\calR_Z \dfn \{\, A_\ell(m) \mid \ell \in \theta_Y \cup \theta_Z \,\} \cup \{C(m)\}$.
    As $\theta_Y$ and $\theta_Z$ are assignments, it is easy to see that $\calR_Z$ is a repair of $\calK'$.
    Hence, we have $\tup{\calT, \calR_Z}$ by assumption.
    But this implies that $\tup{\calT,\calR_Z} \models A_\psi(m)$ by assumption, and hence there is some term $t \in \psi$, such that $t \subseteq \calR_Z$ by construction of \calT.
    Consequently, $\theta_Y \cup \theta_Z \models \psi$.
    Since $\theta_Z$ is an arbitrary assignment over $Z$, we conclude that $\Psi$ is true.	
  \end{claimproof}
\end{proof}

\elbotExistCc*
\begin{proof}
  We begin with the case of \ar semantics.
  Consider some \ar-abduction problem $\tup{\calK, \alpha}$.
  By \cref{lem:ar-triv}, existence of a conflict-confining \ar-hypothesis for $\tup{\calK, \alpha}$ is equivalent to existence of any \ar-hypothesis for $\tup{\calK, \alpha}$.
  Hence, the complexity coincides with that for existence of general \ar-hypotheses, which is \coNP-complete by \cref{thm:elbot-exist-general}.
	
	We now turn to \brave semantics.
  For membership, consider a \brave-abduction problem $\tup{\calK, \alpha}$.
  To check for existence of a conflict-confining \brave-hypothesis, we can guess an ABox \Hyp over the signature of \calK and $\alpha$, and verify that it is a \brave-hypothesis as well as conflict-confining.
  This algorithm shows membership in $\NP^{\NP}$, as both checks can be done by querying an \NP-oracle:
  The former is \brave-entailment and therefore in \NP, while the latter is in \coNP by \cref{lem:cc-compl}.
	
	For hardness, we reduce from non-validity of $\forall\exists$-QBFs to existence of conflict-confining \brave-hypotheses.
  Let $\Phi\dfn \forall Y\exists Z.\varphi$, where $\varphi$ is a CNF represented as a set of clauses, where clauses are sets of literals.
  Let $X \coloneqq Y \cup Z$.
	We construct the KB $\calK = \tup{\calT, \calA}$, using concept names $N = \{A_x, A_{\bar x} \mid x \in X\} \cup \{V_y \mid y \in Y\} \cup \{A_c \mid c \in \varphi\} \cup \{A_{\varphi}, A_{\bar\varphi}, C\}$ as in the proof of \cref{thm:elbot-exist-sig}.
  As we do not have an explicit \SignatureRestriction, we use the condition of being conflict-confining to encode this restriction:
  Using the idea from \cref{ex:brave-cc}, we ensure that any hypothesis that uses symbols outside the intended signature $\Sigma = \{\, A_y, A_{\bar y} \mid y \in Y \,\} \cup \{m\}$ introduces new conflicts, and is thus not conflict-confining.
  For this, we use additional concept names $C_d$ and $B_d$.
	
	Using these ideas, we define
	\begin{align*}
		\calT \coloneqq {} & \{\, C_d\sqcap \bigsqcap_{y\in Y} V_y \sqcap A_{\bar \varphi} \subsum C\, \} \cup {} \\
		& \{\,C_d\sqcap A_y \subsum V_y, C_d\sqcap A_{\bar y} \subsum V_y \mid y\in Y\,\} \cup {} \\
		& \{\,A_x \sqcap A_{\bar x} \subsum \bot \mid x\in X\, \}  \cup {} \\
		& \{\, C_d\sqcap A_\ell \subsum A_c \mid \ell \in c, c\in\varphi\, \}\cup {} \\ 
		& \{\,C_d\sqcap \bigsqcap_{y\in Y} V_y\sqcap \bigsqcap_{c\in \varphi}A_c\subsum A_\varphi\, \} \cup {} \\ 
		& \{\, A_\varphi \sqcap A_{\bar\varphi}\subsum \bot\,\} \cup \{\, A_\varphi\sqcap B_d \subsum \bot\, \}  \cup {} \\ 
		& \{\, C_d\sqcap B_d \subsum \bot\, \} \cup {} 
		\{\, C\sqcap B_d \subsum \bot\, \} \cup {}\\ 
		& \{\, A_c\sqcap B_d \subsum \bot \mid c\in\varphi\, \} \cup {} 
		 \{\, V_y\sqcap B_d \subsum \bot\, \} \text{ and} \\
    \calA \coloneqq & \{\, A_z(m), A_{\bar z}(m) \mid z\in Z \,\} \cup \{A_{\bar\varphi}(m), B_d(m), C_d(m)\},
	\end{align*}
	where $m$ is an individual.
	Finally, let $\calK \coloneqq \tup{\calT, \calA}$ and $\alpha \coloneqq C(m)$.

	Axioms in each line of $\calT$ encode the following intuition:
	(i) and (ii) enforce a complete assignment over $Y$ as a counter-example to the satisfaction of $\varphi$,
	(iii) ensures a valid assignment over $X$,
	(iv) satisfaction of clauses via their literals,
	(v) satisfaction of $\varphi$, but not solely by a partial assignment over $Z$ variables, and 
	(vi) satisfaction of $\varphi$ causes a new conflict.
  The remaining axioms use the idea from \cref{ex:brave-cc} to implicitly enforce a \SignatureRestriction, disallowing certain assertions in any conflict-confining hypothesis.

	Now $\tup{\calK, \alpha}$ is our desired abduction problem.
	It is a \brave-abduction problem:
  The KB \calK is inconsistent and as it contains no assertions of the form $A_\ell(m)$ for any $\ell\in\{y,\bar y\}$, we have $\calK \not\models C(m)$.
	
	For correctness, we prove the following claim.
	
	\begin{claim}\label{claim:el-cc-brave}
		$\Phi$ is false iff $\tup{\calK, \alpha}$ admits a conflict-confining \brave-hypothesis.
	\end{claim}
	\begin{claimproof}
		Suppose $\Phi$ is false and let $\theta_Y$ (seen as a set of literals over $Y$) be an assignment over $Y$ s.t.\ $\forall Z \varphi[\theta_Y]$ is false, where $\varphi[\theta_Y]$ denotes the formula obtained from $\varphi$ by applying the partial assignment $\theta_Y$.
		Define $\calH \coloneqq \{\, A_\ell(m) \mid \ell \in \theta_Y \,\}$. 
		We now show that \Hyp is a conflict-confining \brave-hypothesis for $\tup{\calK, \alpha}$ by showing (i) $\tup{\calT, \calA\cup\calH} \models_\brave \alpha$ and (ii) $\calH$ is conflict-confining.

		For (i), let $\calB \coloneqq \Hyp \cup \{C_d(m), A_{\bar\varphi}(m)\}$.
    It is easy to see that \calB is \calT-consistent, so there is a repair $\calR \supseteq \calB$ of $\tup{\calT, \calA \cup \Hyp}$.
    Further, we have $\tup{\calT, \calB} \models C(m)$ and therefore $\tup{\calT, \calR} \models C(m)$.

		To see (ii), we first observe that $\calH$ does not trigger any inconsistency due to an axiom of the form $A_y\sqcap A_{\bar y}\subsum \bot$, as $\theta_Y$ is an assignment.
    Hence, \Hyp is \calT-consistent.
    Next, we show that for any repair $\calR'$ of \calK, the set $\calR' \cup \Hyp$ is \calT-consistent.
    Due to the axioms of the form $A_x \sqcap A_{\bar x} \sqsubseteq \bot$, the set $\calR' \cap \{\, A_z(m), A_{\bar z}(m) \mid z \in Z \,\}$ corresponds to a (potentially partial) assignment over $Z$.
    Since $\varphi[\theta_Y]$ is false for all assignments over $Z$ by assumption, we have $\tup{\calT, \calR' \cup \Hyp} \not\models A_\varphi$ by construction of \calT.
    This means that $A_\varphi \sqcap A_{\bar \varphi} \sqsubseteq \bot$ does not trigger.
    The remaining disjointness axioms use the idea from \cref{ex:brave-cc} to avoid conflicts in $\tup{\calT, \calR' \cup \Hyp}$:
    Due to the axiom $C_d \sqcap B_d \sqsubseteq \bot$, we have $C_d(m) \not\in \calR'$ or $B_d(m) \not\in \calR'$.
    In both cases, none of the remaining disjointness axioms can trigger, as either the assertion $B_d(m)$ is missing (and not entailed) or no assertion of the form $C(m)$, $A_\varphi(m)$, $A_c(m)$, or $V_y(m)$ is entailed from $\tup{\calT, \calR' \cup \Hyp}$.

    Conversely, suppose that there is a conflict-confining \brave-hypothesis \Hyp for $\tup{\calK, \alpha}$.
    We first argue that, w.l.o.g., \Hyp only contains assertions of the form $A_y(m)$ and $A_{\bar y}(m)$.
    To this end, note that \Hyp may only contain assertions over the individual $m$, since it is the only individual in $\tup{\calK, \alpha}$.
    Further, we can assume w.l.o.g.\ that \Hyp only contains assertions using concept or role names occurring in \calK, as fresh concept and role names cannot help entailment of $\alpha$ in \calT.
    Additionally we can assume w.l.o.g.\ that it does not contain any assertions already present in \calA.
    Finally, we observe that all other assertions over the signature of $\tup{\calK, \alpha}$ would introduce new conflicts, contradicting the assumption that \Hyp is conflict-confining:
	  The assertion $C(m)$ would introduce the new conflict $\{C(m), B_d(m)\}$, while the assertion $A_\varphi(m)$ would introduce the new conflict $\{A_\varphi(m), B_d(m)\}$.
	  Any assertion $A_c(m)$ for $c \in \varphi$ would introduce the new conflict $\{A_c(m), B_d(m)\}$, and any assertion $V_y(m)$ for $y \in Y$ would introduce the new conflict $\{V_y(m), B_d(m)\}$.

    Next, it is easy to see that \Hyp contains exactly one of the assertions $A_y(m)$ and $A_{\bar y}(m)$ for each $y \in Y$:
    It contains at least one, since $C(m)$ is entailed in some repair, and this is only possible when $\bigsqcap_{y \in Y} V_y$ is entailed in that repair.
    It contains at most one, since otherwise the axiom $A_y \sqcap A_{\bar y} \sqsubseteq \bot$ would lead to a new conflict.
    Consequently, \Hyp corresponds to an assignment $\theta_Y$ over $Y$.
    
    Now consider an assignment $\theta_Z$ over $Z$, represented as a set of literals.
    We now show that $\varphi[\theta_Y \cup \theta_Z]$ is false.
    As this applies to all assignments over $Z$, it implies that $\forall Y \exists Z \varphi(Y,Z)$ is false.
    Let
    \[\calA_Z \coloneqq \{\, A_\ell(m) \mid \ell \in \theta_Z \,\},\]
    i.e., the subset of \calA corresponding to assignment $\theta_Z$.
    The set $\calA_Z$ is \calT-consistent, since $\theta_Z$ is an assignment.
    As for all $y \in Y$ we have $\tup{\calT, \calA_Z} \not\models V_y(m)$, $\calA_Z \cup \{A_{\bar\varphi}(m)\}$ is \calT-consistent.
    Due to \Hyp being conflict-confining, this implies that $\calA_Z \cup \{A_{\bar \varphi}(m)\} \cup \Hyp$ is also \calT-consistent.
    Consequently, we have $\tup{\calT, \calA_Z \cup \Hyp} \not\models A_\varphi(m)$, implying that $\varphi[\theta_Y \cup \theta_Z]$ is false by construction of \calT.
	\end{claimproof}
\end{proof}

\elbotVerifCc*
\begin{proof} 
	Let $\tup{\calK, \alpha}$ be an $\calS$-abduction problem with $\calK = \tup{\calT, \calA}$ and $\Hyp$ a given ABox.
	We begin with the case of \ar semantics.
	By definition, \Hyp is a conflict-confining \ar-hypothesis iff (1) $\tup{\calT, \calA\cup \Hyp} \models_\ar \alpha$, and (2) $\Hyp$ is conflict-confining for \calK, i.e., $\Conf(\tup{\calT, \calA \cup \Hyp}) = \Conf(\tup{\calT, \calA})$.
  Hence, we can check that \Hyp is not conflict-confining in \NP, using the following approach.
	Guess two subsets $\calR$ and $\calC$ of $\calA\cup\calH$, and verify that
  \begin{itemize}
	  \item $\tup{\calT,\calR}\not\models \alpha$, $\tup{\calT,\calR}\not\models \bot$,  and  $\tup{\calT, \calR\cup\{\beta\}}\models\bot$ for every $\beta\in (\calA\cup\calH)\setminus \calR$, and
	  \item $\calC \not\subseteq \calA$, $\tup{\calT, \calC}\models \bot$, and $\tup{\calT, \calC\setminus\{\gamma\}}\not\models \bot$ for any $\gamma\in C$.
  \end{itemize}
 	Both checks can be performed in polynomial time.
  For correctness, note the following.
	If the algorithm acccepts, either we obtain a counter example to \calH being an \ar-hypothesis (if the first check succeeds) or to \calH being conflict-confining (if the second check succeeds).
	This yields membership in \coNP.

	For hardness, we reduce from \ar-entailment.
	To achieve this, let $\calK=\tup{\calT,\calA}$ and $A(a)$ be an instance of \ar-entailment.
	Define $\calK' \coloneqq \tup{\calT',\calA}$, where
  \[\calT' \coloneqq \calT \cup \{X \sqcap A \subsum C\}.\]
	Finally, let $\alpha \dfn C(a)$ and $\Hyp \dfn \{X(a)\}$.
	We observe that $\calK \models_\ar A(a)$ iff $\tup{\calT', \calA \cup \Hyp} \models_\ar \alpha$.
	Note that $\Hyp$ is trivially conflict-confining as it uses a fresh concept name $X$ that cannot participate in any conflict.
	
	We now turn to \brave semantics.
	For membership, observe that $\Hyp$ is a conflict-confining $\brave$-hypothesis iff (1) $\tup{\calT, \calA\cup \Hyp} \models_\brave \alpha$, and (2)  $\Conf(\tup{\calT, \calA \cup \Hyp}) = \Conf(\tup{\calT, \calA})$.
	
	In the case of $\brave$-hypotheses, (1) is entailment of concept assertions for \ELbot and hence in \NP, while (2) can be checked in \coNP by universally guessing a conflict $\calC$ such that $\calC\in \Conf(\tup{\calT, \calA \cup \Hyp})$ and $\calC\not\in \Conf(\tup{\calT, \calA})$.
	The last check can be performed  in polynomial time.
	Hence, the problem is contained in \DP.
	
	For hardness, we reduce from a combination of entailment and non-entailment under \brave semantics to verification of conflict-confining \brave-hypotheses.
	Given an instance $\tup{\calK, \alpha_1, \alpha_2}$ for some inconsistent KB \calK, the problem asks whether $\calK\models_\brave \alpha_1$ and $\calK\not\models_\brave \alpha_2$.
	This problem is \DP-complete because the first question is \NP-complete and the second question is \coNP-complete under \brave semantics.
	For the reduction, assume w.l.o.g.\ that $\alpha_1$ and $\alpha_2$ are concept assertions over the same individual, but using different concept names, and let $\alpha_1 = A(a)$, $\alpha_2= B(a)$, and $\calK = \tup{\calT,\calA}$.
	We construct a KB $\calK'$, an observation $\alpha$, and a hypothesis \Hyp next.
	Let $\calK' \coloneqq \tup{\calT',\calA}$ with
  \[\calT' \coloneqq \calT \cup \{C \sqcap A\subsum D, C\sqcap B\subsum \bot\},\]
  $\alpha \coloneqq D(a)$, and $\Hyp \coloneqq \{C(a)\}$ for fresh concepts $C$ and $D$.
	The instance is a valid abduction problem, since $\calK'$ is inconsistent and $\calK' \not\models_\brave \alpha$.
	Intuitively, \Hyp is a Brave-hypothesis for $\tup{\calK', \alpha}$ iff $\calK \models_\brave A(a)$ and \Hyp is conflict-confining for $\calK'$ iff $\calK \not\models_\brave B(a)$.
	It remains to show correctness, which is stated in the following claim.
	\begin{claim}\label{claim:cc-verification-el}
		\Hyp is a conflict-confining \brave-hypothesis for $\alpha$ in $\calK'$ iff $\calK \models_\brave \alpha_1$ and $\calK \not\models_\brave \alpha_2$.
	\end{claim}
	\begin{claimproof}
	($\Rightarrow$)
	Suppose \Hyp is a conflict-confining \brave-hypothesis for $\tup{\calK', D(a)}$.
	Notice that the only way to obtain the entailment $\calK' \models_\brave D(a)$ is via the TBox axiom $C\sqcap A\subsum D$, since no axiom in \calK contains $D$.
	Therefore, we must have $\calK \models_\brave A(a)$.
	Moreover, we have $\calK \not\models_\brave B(a)$:
	Suppose to the contrary that $\calK \models_\brave B(a)$ and let $\calR$ be a witnessing repair such that $\tup{\calT, \calR} \models B(a)$.
	Then, we have that $\tup{\calT', \calR \cup \calH} \models \bot$, in particular due to the axiom $C \sqcap B \subsum \bot$ and the assertion $C(a)$.
  Since \calR is a repair, and hence \calT-consistent, there must be a conflict of $\tup{\calT', \calR \cup \Hyp}$ that is not a conflict of $\tup{\calT',\calA}$.
	But this leads to a contradiction to our assumption that \Hyp is conflict-confining for $\calK'$.
	As a result, $\calK'\not\models_\brave B(a)$ must be true.
	
	($\Leftarrow$) Suppose $\calK \models_\brave A(a)$ and $\calK \not\models_\brave B(a)$.
	Then, $\calK' \models_\brave A(a)$ and hence
  \[\tup{\calT', \calA\cup\Hyp} \models_\brave D(a).\]
	Therefore \Hyp is indeed a \brave-hypothesis for $\tup{\calK', \alpha}$.
	To show that \Hyp is conflict-confining for $\calK'$, suppose to the contrary that there is a conflict $\calC\in \Conf(\tup{\calT', \calA\cup\Hyp})$ such that $\calC\not\in \Conf(\tup{\calT', \calA})$.
	In particular, this implies that $C(a)\in\calC$ since $\calH=\{C(a)\}$.
	As a result, we have $\calC\setminus\{C(a)\}\subseteq \calA$ and $\calC$ is $\calT'$-consistent, and hence also \calT-consistent.
  But this implies that $\tup{\calT', \calC} \models B(a)$, since the only conflict involving $C(a)$ is via the axiom $C\sqcap B\subsum \bot$.
  Further, this means that $\tup{\calT, \calC} \models B(a)$, as the new axioms in $\calT'$ do not help entailment of $B(a)$.
	As \calC is \calT-consistent, there exists a repair $\calR \supseteq \calC$ of \calK, and we have $\tup{\calT, \calR} \models B(a)$.
  Consequently, $\calK \models_\brave B(a)$.
	But this is a contradiction to our assumption, so \Hyp must be conflict-confining.
\end{claimproof}
	We conclude by observing that the above reduction can be achieved in polynomial time. 
\end{proof}

\elbotVerifCminSubset*
\begin{proof}
  In the case of \ar semantics, the problem inherits the complexity from verification of conflict-confining \ar-hypotheses by \cref{lem:ar-triv}, using the same argument as in the corresponding part of the proof of \cref{thm:dllite-verif-cc-cmin}.

  We now turn to \brave semantics.
  For membership, consider a \brave-abduction problem $\tup{\calK, \alpha}$ and ABox \Hyp.
  Observe that \Hyp is a $\subseteq_c$-minimal \brave-hypothesis for $\tup{\calK, \alpha}$ iff (i) $\tup{\calT,\calA\cup\calH}\models_\brave $ and (ii) there is no \brave-hypothesis $\calH'$ for $\tup{\calK,\alpha}$ such that
  \[\Conf(\tup{\calT, \calA \cup \calH'}) \prec \Conf(\tup{\calT, \calA \cup \calH}).\]
	Here, one can guess a set $\calH'$ as a counter-witness to $\calH$ being a conflict-minimal hypothesis, whereas the \brave-entailment in (i) and the relationship between the conflict sets in (ii) can be performed via oracle calls.
	This yields membership in $\coNP^{\NP}$ or equivalently, in $\PiP$.

	For hardness, we reuse the construction from the hardness proof for Theorem~\ref{thm:elbot-exist-cc}.
  Let $\Phi$ be a $\forall\exists$-QBF, and $\calK = \tup{\calT,\calA}$ be the KB obtained $\Phi$ using that construction.
	We add new axioms and assertions to \calK to obtain the KB $\calK' \dfn \tup{\calT,\calA'}$ as follows:
	\begin{align*}
    \calT' &\dfn \calT \cup \{C_d \sqcap X \subsum C, X\sqcap Y\subsum \bot\}, \text{ and} \\
    \calA' &\dfn \calA\cup \{Y(m)\}
  \end{align*}
	Then, we let $\alpha\dfn C(m)$ as before and take $\calH\dfn\{X(m)\}$.
	Here, $\calH$ induces exactly one more conflict in $\calK'$, namely $\{X(m), Y(m)\}$.

	For correctness, it is easy to see that there is a conflict-confining \brave-hypothesis for $\tup{\calK, C(m)}$ iff there is such a hypothesis for $\tup{\calK', C(m)}$.
  Further, the latter is equivalent to \Hyp not being a $\subseteq_c$-minimal hypothesis for $\tup{\calK', C(m)}$, since \Hyp introduces exactly one new conflict for $\calK'$, while a conflict-confining \brave-hypothesis introduces no new conflicts.
	Thus the correctness follows due to the proof of Claim~\ref{claim:el-cc-brave}.
	This yields the mentioned $\PiP$-hardness.
\end{proof}

\elbotVerifCminCard*
\begin{proof}
  The problem again inherits the complexity from verification of conflict-confining \ar-hypotheses by \cref{lem:ar-triv}, using the same argument as in the corresponding part of the proof of \cref{thm:dllite-verif-cc-cmin}.
\end{proof}

\elbotVerifSubset*
\begin{proof}
	We prove the result for both semantics separately as it uses different techniques.
	The result for \brave and \ar semantics is proven in Theorem~\ref{thm:elbot-verif-brave-subset} and~\ref{thm:elbot-verif-ar-subset}, resp.
\end{proof}

\begin{theorem}
	\label{thm:elbot-verif-brave-subset}
	For \ELbot, verification of $\subseteq$-minimal \brave-hypotheses is \DP-complete.
\end{theorem}
\begin{proof}
  For membership, observe that \Hyp is a $\subseteq$-minimal \brave-hypothesis iff (1) $\tup{\calT, \calA\cup \Hyp} \models_\calS \alpha$ and (2) for all subsets $\Hyp' \subsetneq \Hyp$, we have $\tup{\calT, \calA \cup \Hyp} \not\models_\calS \alpha$.
  Here, (1) is entailment of concept assertions for \ELbot and hence in \NP.
  For (2), we universally guess a subset $\Hyp' \subseteq \Hyp$ and subset $\calR \subseteq \calA \cup \Hyp'$, and verify in polynomial time that either \calR is not a repair of $\tup{\calT, \calA \cup \Hyp'}$, or we have $\tup{\calT, \calR} \not\models \alpha$.
  Intuitively, this can be seen as universally guessing a subset $\calH'$ and repair \calR of $\tup{\calT, \calA \cup \Hyp'}$ and checking that $\tup{\calT, \calR} \not\models \alpha$.
	As membership is checked as the conjunction of an \NP- and a \coNP-problem, this shows membership in \DP.
	
	For hardness, we reduce from a combination of entailment and non-entailment under \brave semantics verification of $\subseteq$-minimal \brave-hypotheses, similar to hardness under \brave semantics in the proof of \cref{thm:elbot-verif-cc}.
	Given an instance $\tup{\calK, \alpha_1, \alpha_2}$ for some inconsistent KB \calK, the problem asks whether $\calK \models_\brave \alpha_1$ and $\calK \not\models_\brave \alpha_2$.
	For the reduction, assume w.l.o.g.\ that $\alpha_1$ and $\alpha_2$ are concept assertions over the same individual, but using different concept names, and let $\alpha_1 = D(a)$, $\alpha_2= C(a)$, and $\calK = \tup{\calT,\calA}$.
  We now construct a KB $\calK'$, an observation $\alpha$, and a hypothesis \Hyp.
	Let $\calK' \coloneqq \tup{\calT',\calA}$ with
  \[\calT' \coloneqq \calT \cup \{C \subsum A, A\sqcap B\sqcap D \subsum Q\},\]
  $\alpha \coloneqq Q(a)$, and $\Hyp \coloneqq \{A(a), B(a)\}$ for fresh concepts $A,B,Q$.
	The instance is a valid abduction problem, since $\calK'$ is inconsistent and $\calK' \not\models_\calS \alpha$.
	Intuitively, $\Hyp$ is a Brave-hypothesis for $\alpha$ in $\calK'$ iff $\calK \models_\brave D(a)$ and $\Hyp$ is subset-minimal iff $\calK \not\models_\brave C(a)$.
	We next prove the correctness of the reduction.
	\begin{claim}
		\Hyp is a $\subseteq$-minimal hypothesis for $\alpha$ in $\calK'$ iff $\calK \models_\brave \alpha_1$ and $\calK \not\models_\brave \alpha_2$.
	\end{claim}
	\begin{claimproof}
    ($\Rightarrow$)
	  Suppose $\Hyp$ is a $\subseteq$-minimal hypothesis for $Q(a)$ in $\calK'$.
	  Observe that the only way to obtain the entailment $\calK' \models_\brave Q(a)$ is via the TBox axiom $A\sqcap B\sqcap D\subsum Q$, since no axiom in $\calK$ contains $Q$ and $Q(a) \not\in \Hyp$.
	  Therefore, $\calK \models_\brave D(a)$ since otherwise, $\calK' \not\models_\brave D(a)$ and hence $\calK' \not\models_\brave Q(a)$.
	  Moreover, we have $\calK \not\models_\brave C(a)$:
	  Suppose to the contrary that $\calK \models_\brave C(a)$.
	  Then $\calK' \models A(a)$ due to the axiom $C \sqsubseteq A$.
	  Consequently, $\{B(a)\}$ is a \brave-hypothesis for $\alpha$ in $\calK'$, which is a contradiction to $\subseteq$-minimality of $\Hyp$.
	  
	  ($\Leftarrow$) Suppose $\calK \models_\brave D(a)$ and $\calK \not\models_\brave C(a)$.
	  Then, we also have $\calK' \models_\brave D(a)$ and hence
    \[\tup{\calT', \calA\cup\Hyp} \models_\brave Q(a),\]
    so $\Hyp$ is a $\brave$-hypothesis for $\tup{\calK', \alpha}$.
	  For $\subseteq$-minimality, suppose to the contrary that there is a \brave-hypothesis $\Hyp' \subsetneq \Hyp$ for $\alpha$ in $\calK'$.
	  Since $B(a)$ cannot be entailed via any axiom in $\calK'$, we have $\Hyp' = \{B(a)\}$.
	  However, this implies that $\calK' \models_\brave A(a)$, which can only be true if $\calK' \models_\brave C(a)$.
    But then $\calK \models_\brave C(a)$, which is a contradiction.
  \end{claimproof} 
\end{proof}

\begin{theorem}\label{thm:elbot-verif-ar-subset}
	For \ELbot, verification of $\subseteq$-minimal \ar-hypotheses is \PiP-complete.
\end{theorem}
\begin{proof}
  Membership in \PiP is shown similar as \DP-membership for \cref{thm:elbot-verif-brave-subset}.
  Given an \ar-abduction problem $\tup{\calK, \alpha}$ with $\calK = \tup{\calT, \calA}$ and an ABox \Hyp, \Hyp is a $\subseteq$-minimal $\ar$-hypothesis for $\tup{\calK, \alpha}$ iff (1) $\tup{\calT, \calA \cup \Hyp} \models_\ar \alpha$ and (2) for all subsets $\Hyp' \subsetneq \Hyp$, we have $\tup{\calT, \calA \cup \Hyp} \not\models_\ar \alpha$.
  Here, (1) is \ar-entailment and hence in \coNP.
  For (2) we can guess a subset $\Hyp' \subsetneq \Hyp$ and check that $\tup{\calT, \calA \cup \Hyp'} \not\models_\ar \alpha$ using an oracle for non-entailment under \ar semantics, resulting in \PiP-membership.

  For hardness, we reduce from checking whether a given $\Pi_2$-QBF is true.
  Let $\Phi = \forall X \exists Y \varphi(X,Y)$ be a $\Pi_2$-QBF, where $X = \{x_1, \dots, x_n\}$ and $Y = \{y_1, \dots, y_m\}$.
  Note that we can assume w.l.o.g.\ that $\neg \varphi$ is in DNF with set of terms $\{c_1, \dots, c_k\}$.
  We construct an \ar-abduction problem $\tup{\calK, A(a)}$ and \ar-hypothesis $\Hyp$ for it s.t.\ $\forall X \exists Y \varphi$ is true iff $\Hyp$ is a $\subseteq$-minimal \ar-hypothesis.
  Equivalently, $\exists X \forall Y \neg \varphi$ is true iff there is some subset $\Hyp' \subsetneq \Hyp$ s.t.\ $\Hyp'$ is an \ar-hypothesis for $\tup{\calK, A(a)}$.

  The proof idea is as follows.
  We construct \calK and \Hyp in such a way that that subsets $\Hyp' \subsetneq \Hyp$ encode assignments over $X$, and repairs of $\tup{\calT, \calA \cup \Hyp'}$ range over encodings of all assignments over $Y$.
  Entailment of $A(a)$ in $\tup{\calT, \calA \cup \Hyp'}$ is then equivalent to the corresponding assignment over $X \cup Y$ satisfying $\varphi$.
  Further, we use the idea from \cref{ex:ar-non-convex}:
  To ensure entailment of $A(a)$ in $\tup{\calT, \calA \cup \Hyp}$, we use an additional axiom in \calT that circumvents the construction encoding $\Phi$, but cannot be triggered in $\tup{\calT, \calA \cup \Hyp'}$ for any subset $\Hyp' \subsetneq \Hyp$.
  Another important component of the construction will be two disjoint concepts $B_1$ and $B_2$ that allow us to split the set of repairs in two parts:
  The repairs containing the assertion $B_1(a)$ will ensure that $\varphi$ is satisfied for all assignments over $Y$, while the repairs containing the assertion $B_2(a)$ will ensure that \Hyp contains a full assignment over $X$.

  We now provide the definition of \calK, argue that $\tup{\calK, A(a)}$ is an \ar-abduction problem, and define the ABox \Hyp.
  We then provide further intuition on the components of the construction, followed by the proof of correctness.
  Define $\calK = \tup{\calT, \calA}$, where $\calT \coloneqq \calT_0 \cup \calT_1 \cup \calT_2$, and $\calT_0$, $\calT_1$, $\calT_2$ and \calA are defined as follows:
  \[
    \calT_0 \coloneqq \{\, B_1 \sqcap B_2 \sqsubseteq \bot, \quad B_1' \sqcap B_2  \sqsubseteq \bot \,\},
  \]
  \begin{align*}
    \calT_1 \coloneqq {} & \left\{ B_1 \sqcap \bigsqcap_{1 \leq i \leq n} (T_{x,i} \sqcap F_{x,i}) \sqsubseteq A \right\} \cup {} \\
    & \{\, B_1 \sqcap C_j \sqsubseteq A \mid 1 \leq j \leq k \,\} \cup {} \\
    & \{\, B_1' \sqcap T_{y,i} \sqcap F_{y,i} \sqsubseteq \bot \mid 1 \leq i \leq m \,\} \cup {} \\
    & \{\, B_1' \sqcap T_{x,i} \sqcap C_j \sqsubseteq \bot \mid \neg x_i \in c_j \,\} \cup {} \\
    & \{\, B_1' \sqcap F_{x,i} \sqcap C_j \sqsubseteq \bot \mid x_i \in c_j \,\} \cup {} \\
    & \{\, B_1' \sqcap T_{y,i} \sqcap C_j \sqsubseteq \bot \mid \neg y_i \in c_j \,\} \cup {} \\
    & \{\, B_1' \sqcap F_{y,i} \sqcap C_j \sqsubseteq \bot \mid y_i \in c_j \,\}
  \end{align*}
  \begin{align*}
    \calT_2 \coloneqq {} & \left\{ B_2 \sqcap \bigsqcap_{1 \leq i \leq n} \mathrm{HAVE}_i \sqsubseteq A \right\} \cup {} \\
    & \{\, T_{x,i} \sqsubseteq \mathrm{HAVE}_i,\; F_{x,i} \sqsubseteq \mathrm{HAVE}_i \mid 1 \leq i \leq n \,\}
  \end{align*}    
  \begin{align*}
    \calA \coloneqq {} & \{ B_1(a), B_1'(a), B_2(a) \} \cup {} \\
    & \{\, T_{y,i}(a), F_{y,i}(a) \mid 1 \leq i \leq m \,\} \cup {} \\
    & \{\, C_j(a) \mid 1 \leq j \leq k \,\}.
  \end{align*}
  First, note that $\tup{\calK, A(a)}$ is an \ar-abduction problem:
  the KB \calK is inconsistent, for example it has the conflict $\{B_1(a), B_2(a)\}$.
  Also, $\calK \not\models_\ar A(a)$ as there is a repair \calR of \calK with $B_2(a) \in \calR$.
  By the axioms in $\calT_0$, we have $B_1(a) \not\in \calR$, so $A(a)$ cannot be entailed by the aximos in $\calT_1$.
  But there is also no assertion of the form $T_{x,i}(a)$ or $F_{x,i}(a)$ in \calR, as these are not contained in \calA.
  Hence, $A(a)$ cannot be entailed by the first axiom in $\calT_2$, so $\tup{\calT, \calR} \not\models A(a)$.
  Now, define the \ar-hypothesis \Hyp for $\tup{\calK, A(a)}$ by
  \begin{align*}
    \Hyp \coloneqq \{\, T_{x,i}(a), F_{x,i}(a) \mid 1 \leq i \leq n \,\}.
  \end{align*}

  We now provide some intuition on the construction of \calK and \Hyp.
  The axioms in $\calT_0$ split the set of repairs into two parts:
  Those repairs that contain $B_1(a)$ and potentially $B_1'(a)$, and those repairs that contain $B_2(a)$.
  The axioms in $\calT_1$ only affect the former repairs while those in $\calT_2$ only affect the latter.
  Regarding the encoding of $\varphi$ in the construction, presence of an assertion of the form $T_{z,i}(a)$ or $F_{z,i}(a)$ encodes that variable $z$ is assigned to \texttt{true} or \texttt{false}, resp.
  Accordingly, repairs containing $B_1(a)$ encode assignments over $Y$, due to the axioms of the form $B_1' \sqcap T_{y,i} \sqcap F_{y,i} \sqsubseteq \bot$, while similarly subsets of \Hyp encode assignments over $Z$.
  (W.l.o.g., one can assume that \Hyp contains at most one of the assertions $T_{y,i}$ and $F_{y,i}$ for each $i$, due to how satisfaction of $\varphi$ is encoded.)
  Further, the assertions of the form $C_j(a)$ correspond to the terms $c_j$ of $\varphi$.

  Based on these ideas, $\calT_1$ ensures that for a given subset $\Hyp' \subsetneq \Hyp$, $A(a)$ is entailed in all repairs containing $B_1(a)$, iff the assignment $\theta_{\Hyp'}$ corresponding to $\Hyp'$ satisfies the formula $\forall Z \varphi$:
  The last $4$ kinds of axioms in $\calT_1$ ensure that $C_j(a)$ conflicts with assignments that falsify $c_j$, so some $C_j(a)$ only remains in all repairs (yielding entailment of $A(a)$), if at least one term of $\varphi[\theta_{\Hyp'}]$ is satisfied by each assignment over $Z$.
  This construction is circumvented for the full set \Hyp: Here, entailment of $A(a)$ is ensured via the first axiom in $\calT_1$.
  On the other hand, $\calT_2$ ensures that each \ar-hypothesis $\Hyp' \subsetneq \Hyp$ of $\tup{\calK, A(a)}$ contains at least one of the assertions $T_{y,i}(a)$, $F_{y,i}(a)$ (using the repairs containing $B_2(a)$).

  The full proof of correctness is now split into three parts, namely the proof that \Hyp is an \ar-hypothesis for $\tup{\calK, A(a)}$ and the two directions of the equivalence that there is an \ar-hypothesis $\Hyp' \subsetneq \Hyp$ of $\tup{\calK, A(a)}$ iff $\Phi$ is false, shown in Claims~\ref{cl:elbot-subsetmin-1}, \ref{cl:elbot-subsetmin-2} and~\ref{cl:elbot-subsetmin-3}.

  \begin{claim}\label{cl:elbot-subsetmin-1}
    \Hyp is an \ar-hypothesis for $\tup{\calK, A(a)}$.
  \end{claim}
  \begin{claimproof}
    Consider any repair $\calR$ of $\tup{\calT, \calA \cup \Hyp}$.
    We consider two cases.

    Case ``$B_1(a) \in \calR$''.
    By the axioms in $\calT_0$, we have $B_2(a) \not\in \calR$, so $A(a)$ cannot be entailed by the first axiom in $\calT_2$.
    If we have $C_j(a) \in \calR$ for some $j$, $A(a)$ is entailed by the axiom $B_1 \sqcap C_j \sqsubseteq A$.
    Otherwise, the only remaining disjointness axioms that could trigger are those of the form $B_1' \sqcap T_{y,i} \sqcap F_{y,i} \sqsubseteq \bot$, which do not affect the concepts of the form $T_{x,i}$ or $F_{x,i}$.
    Hence, by maximality of repairs, we have $T_{x,i}(a), F_{x,i}(a) \in \calR$ f.a.\ $1 \leq i \leq n$.
    Consequently, $A(a)$ is entailed by the first axiom in $\calT_1$.

    Case ``$B_1(a) \not\in \calR$''.
    First note that $B_1'(a) \not\in \calR$: Otherwise, we would have $B_2(a) \not\in \calR$ by the axioms in $\calT_0$.
    But then, $B_1(a) \in \calR$ by maximality of repairs, as $B_1$ does not occur in any disjointness axioms in $\calT_1$ or $\calT_2$.
    As $B_1'(a) \not\in \calR$, none of the disjointness axioms in $\calT_1$ can trigger, so we only need to consider those in $\calT_2$.
    Hence, we have $B_2(a), T_{x,i}(a), F_{x,i}(a) \in \calR$ for all $1 \leq i \leq n$ by maximality of repairs.
    Consequently, $A(a)$ is entailed via the axioms of the form $T_{x,i} \sqsubseteq \mathrm{HAVE}_i(a)$ and $F_{x,i} \sqsubseteq \mathrm{HAVE}_i(a)$ as well as the axiom $B_2 \sqcap \bigsqcap_{1 \leq i \leq n} \mathrm{HAVE}_i \sqsubseteq A$.
  \end{claimproof}

  \begin{claim}\label{cl:elbot-subsetmin-2}
    If $\forall X \exists Y \varphi(X,Y)$ is false, then there is a subset $\Hyp' \subsetneq \Hyp$ s.t.\ $\Hyp'$ is an \ar-hypothesis for $\tup{\calK, A(a)}$.
  \end{claim}
  \begin{claimproof}
    In this case, there is an assignment $s_X$ over $X$ s.t.f.a. assignments $s_Y$ over $Y$, there is conjunction term $c_j$ in $\neg \varphi$ with
    \[s_X \cup s_Y \models c_j.\]
    Define
    \[\Hyp' = \{\, T_{x,i}(a) \mid s_X(x_i) = 1 \,\} \cup \{\, F_{x,i}(a) \mid s_X(x_i) = 0 \,\}.\]
    Consider any repair $\calR$ of $\tup{\calT, \calA \cup \Hyp'}$.
    Similar to the above proof of claim, we distinguish two cases based on whether $B_1(a) \in \calR$.

    Case ``$B_1(a) \in \calR$''.
    By the axioms in $\calT_0$, we have $B_2(a) \not\in \calR$, so $A(a)$ cannot be entailed by the first axiom in $\calT_2$.
    If we have $B_1'(a) \not\in \calR$, then none of the disjointness axioms in $\calT_1$ can trigger, so there is some $C_j(a) \in \calR$ by maximality of repairs and we get the entailment of $A(a)$ via the axiom $B_1 \sqcap C_j \sqsubseteq A$.
    Otherwise, the disjointness axioms of the form $B_1' \sqcap T_{y,i} \sqcap F_{y,i}$ ensure that \calR contains at most one of the assertions $T_{y,i}(a)$ and $F_{y,i}(a)$ for each $1 \leq i \leq m$.
    Hence, the present assertions in \calR correspond to the assignment $s_\calR$ over $Y$ defined by
    \[s_\calR(y_i) \coloneqq \begin{cases}
      1, & \text{if } T_{y,i}(a) \in \calR \\
      0, & \text{if } F_{y,i}(a) \in \calR.
    \end{cases}\]
    By our assumption, there is some $c_j$ s.t. $s_X \cup s_\calR \models c_j$.
    By construction of $\calT_1$, this means that $C_j(a)$ is not in conflict with any of the assertions in $\Hyp'$ or any of the assertions $T_{y,i}(a)$ or $F_{y,i}(a)$ remaining in \calR.
    Hence, $A(a)$ is entailed by the axiom $B_1 \sqcap C_j \sqsubseteq A$.

    Case ``$B_1(a) \not\in \calR$''.
    First note that $B_1'(a) \not\in \calR$ by the same argument as in the previous proof of claim.
    As $B_1'(a) \not\in \calR$, none of the disjointness axioms in $\calT_1$ can trigger.
    Hence, we have $B_2(a) \in \calR$ by maximality of repairs.
    Further, by maximality of repairs and the fact that $s_X$ is an assignment over $X$, we have $T_{x,i}(a) \in \calR$ or $F_{x,i}(a) \in \calR$ for each $1 \leq i \leq n$.
    Consequently, $A(a)$ is entailed via the axioms of the form $T_{x,i} \sqsubseteq \mathrm{HAVE}_i$ and $F_{x,i} \sqsubseteq \mathrm{HAVE}_i$ as well as the axiom $B_2 \sqcap \bigsqcap_{1 \leq i \leq n} \mathrm{HAVE}_i \sqsubseteq A$.
  \end{claimproof}

  \begin{claim}\label{cl:elbot-subsetmin-3}
    If there is a subset $\Hyp' \subsetneq \Hyp$ s.t.\ $\Hyp'$ is an \ar-hypothesis for $\tup{\calK, A(a)}$, then $\forall X \exists Y \varphi(X,Y)$ is false.
  \end{claim}
  \begin{claimproof}
    By assumption we have $\tup{\calT, \calA \cup \Hyp'} \models_\ar A(a)$.
    First, note that $\Hyp' \cup \{ B_2(a) \}$ is \calT-consistent.
    Hence, there is a repair $\calR$ of $\tup{\calT, \calA \cup \Hyp'}$ with $\Hyp' \cup \{B_2(a)\} \subseteq \calR$.
    By the axioms in $\calT_0$, we have $B_1(a) \not\in \calR$.
    Hence, $A(a)$ cannot be entailed using axioms in $\calT_1$, so it must be entailed by the axiom $B_2 \sqcap \bigsqcap_{1 \leq i \leq n} \mathrm{HAVE}_i \sqsubseteq A$.
    This means that we have $T_{x,i}(a) \in \calR$ or $F_{x,i}(a) \in \calR$ for each $1 \leq i \leq n$.
    As these assertions do not occur in \calA, they must be contained in $\Hyp'$.

    Define the assignment $s_{\Hyp'}$ over $X$ by
    \[s_{\Hyp'}(x_i) \coloneqq \begin{cases}
      1, & \text{if } T_{x,i}(a) \in \Hyp' \\
      0, & \text{otherwise}
    \end{cases}\]
    and consider any assignment $s_Y$ over $Y$.
    As $s_Y$ is a function, the set
    \begin{align*}
      \calA_{s_Y} \coloneqq {} & \{B_1(a), B_1'(a)\} \cup \Hyp' \cup {} \\
      & \{\, T_{y,i}(a) \mid s_Y(y_i) = 1 \,\} \cup \{\, F_{y,i}(a) \mid s_Y(y_i) = 0 \,\}
    \end{align*}
    is $\calT$-consistent, so there are repairs containing all assertions in $\calA_{s_Y}$.
    In these repairs, $A(a)$ cannot be entailed by the axioms in $\calT_2$, since $B_2(a)$ is in conflict with $B_1(a)$.
    Further, it cannot be entailed by the first axiom in $\calT_1$, as $\Hyp' \subsetneq \Hyp$.
    Hence, there is some index $j$ s.t.\ $C_j(a)$ is not in conflict with any of the assertions in $\calA_{s_Y}$.
    (Otherwise, there is a repair $\calR$ among these that does not contain any of the assertions of the form $C_j(a)$, meaning that $\tup{\calT, \calR} \not\models A(a)$.)
    In particular, this is also true for the subset $\calA_{s_Y}' \subseteq \calA_{s_Y}$ defined by
    \[\calA_{s_Y}' \coloneqq \calA_{s_Y} \setminus \{\, F_{x,i}(a) \mid 1 \leq i \leq n, T_{x,i}(a) \in \calA_{s_Y}' \,\}.\]
    But the assertions of the forms $T_{x,i}(a)$, $F_{x,i}(a)$, $T_{y,i}(a)$ and $F_{y,i}(a)$ present in $\calA_{s_Y}'$ correspond to the assignment $s_{\Hyp'} \cup s_Y$, so this means that $c_j$ is not falsified by $s_{\Hyp'} \cup s_Y$ by construction of $\calT_1$. 
    Consequently, $s_{\Hyp'} \cup s_Y \models \neg \varphi$.
    As $s_Y$ was picked arbitrarily, this means that $\exists X \forall Y \neg \varphi(X,Y)$ is true and $\forall X \exists Y \varphi(X,Y)$ is false.
  \end{claimproof}
  Combined, the three above claims show that $\forall X \exists Y \varphi(X,Y)$ is true iff \Hyp is a $\subseteq$-minimal \ar-hypothesis for $\tup{\calK, A(a)}$, finishing the hardness proof.  
\end{proof}

\elbotExistSigCC*
\begin{proof}
	For membership in each case, observe that one can guess a hypothesis $\calH$ over the given signature $\Sigma$.
	The verification then requires to determine that (i) $\calH$ is a $\calS$-hypothesis for $\tup{\calK, \alpha}$, and (ii) $\calH$ is conflict-confining.
	This leads to a complexity of $\NP^{\NP}$, or equivalently $\SigmaP$ in both cases.
	The hardness for \brave semantics follows from the case of conflict-confining hypotheses (Theorem~\ref{thm:elbot-exist-cc}), whereas that of \ar follows from the case of \SignatureRestricted hypotheses (Theorem~\ref{thm:elbot-exist-sig}).
\end{proof}

\end{document}